\documentclass{article}
\usepackage{latexsym}
\usepackage{amsmath, amssymb,amsthm}
\usepackage{amsfonts,euscript}
\usepackage{setspace}
\usepackage{mathrsfs}

\newtheorem{theoreme}{Theorem}[section]
\newtheorem{proposition}[theoreme]{Proposition}
\newtheorem{lemma}[theoreme]{Lemma}

\newtheorem{corollary}[theoreme]{Corollary}
\newtheorem{remark}[theoreme]{Remark}

\DeclareMathAlphabet{\mathpzc}{OT1}{pzc}{m}{it}

\def\supp{{\rm supp}}
\def\R{\mathbb R}
\def\C{\mathbb C}
\def\Z{\mathbb Z}
\def\S{\mathbb S}
\def\Dom{\mathcal D}
\def\Im{{\rm Im}}
\def\Re{{\rm Re}}

\def\rr{{\mathbb R}}
\def\zz{{\mathbb Z}}
\def\cc{{\mathbb C}}

\def\e{{\rm e}}
\def\p{{\rm p}}
\def\i{{\rm i}}
\def\r{{\rm r}}
\def\d{{\rm d}}
\def\ess{{\rm ess}}

\def\Ln{{\rm Ln}}
\def\Q{X} 
\def\H{\mathscr H}
\def\B{\mathscr B}
\def\Ju{J}

\renewcommand\bar{\overline}

\newcommand\cF{{\mathscr F}}

\newcommand\cG{{\mathcal G}}
\newcommand\cD{{\mathcal D}}

\newcommand\cH{{\mathcal H}}

\newcommand\cY{{\mathscr Y}}

\def\bbbone{{\mathchoice {\rm 1\mskip-4mu l} {\rm 1\mskip-4mu l}
{\rm 1\mskip-4.5mu l} {\rm 1\mskip-5mu l}}}
\def\one{\bbbone}
\def\slim{{\rm s-}\lim}
\def\wlim{{\rm w-}\lim}

\def\Ka{\mathcal K}
\def\Ia{\mathcal I}
\def\Ha{\mathcal H}
\def\Ja{\mathcal J}
\def\Ya{\mathcal Y}
\def\La{\mathcal L}

\def\t{{\scriptscriptstyle\#}}
\begin{document}

\title{On Schr\"odinger operators with inverse square potentials on the half-line}

\author{
Jan Derezi\'{n}ski\footnote
  {The financial support of the National Science
Center, Poland, under the grant UMO-2014/15/B/ST1/00126, is gratefully
acknowledged.}\\
Department of Mathematical Methods in Physics, Faculty of Physics\\
University of Warsaw,  Pasteura 5, 02-093, Warszawa, Poland\\
email: jan.derezinski@fuw.edu.pl
\\  \\
Serge Richard\footnote{On leave of absence from
Univ.~Lyon, Universit\'e Claude Bernard Lyon 1, CNRS UMR 5208, Institut Camille Jordan,
43 blvd.~du 11 novembre 1918, F-69622 Villeurbanne cedex, France.} \footnote{
Supported by JSPS Grant-in-Aid for Young Scientists A
no 26707005.}\\
Graduate school of mathematics,
Nagoya University, \\
Chikusa-ku,
Nagoya 464-8602, Japan \\
email: richard@math.nagoya-u.ac.jp}

\maketitle

\vspace{-5mm}

\begin{abstract}
The paper is devoted to  operators given formally by the
expression
\begin{equation*}
-\partial_x^2+\big(\alpha-\frac14\big)\frac{1}{x^{2}}.
\end{equation*}
This expression is  homogeneous of degree minus 2. However,
when we try to realize it as a self-adjoint operator for real
$\alpha$, or closed operator for complex $\alpha$, we find that this homogeneity
can be broken.

This leads to a definition of two holomorphic families of closed
operators on $L^2(\R_+)$, which we denote $H_{m,\kappa}$ and $H_0^\nu$,
with $m^2=\alpha$, $-1<\Re(m)<1$, and where $\kappa,\nu\in\C\cup\{\infty\}$ specify the boundary
condition at $0$.
We study these operators using their explicit solvability in terms of
Bessel-type functions and the Gamma function.
In particular, we show that their point
spectrum has a curious shape: a string of eigenvalues
on a piece of a spiral. Their continuous spectrum is always
$[0,\infty[$. Restricted to their continuous spectrum, we diagonalize
these operators using a generalization of the Hankel
transformation. We also study their scattering theory.

These operators are usually non-self-adjoint. Nevertheless, it
is possible to use concepts typical for the self-adjoint case to study them.
Let us also stress that $-1<\Re(m)<1$ is the maximal region of parameters for which the operators
$H_{m,\kappa}$ can be defined within the framework of the Hilbert space  $L^2(\R_+)$.
\end{abstract}


\tableofcontents

\section{Introduction}
\setcounter{equation}{0}

The family of differential operators
\begin{equation}\label{qfq}
-\partial_x^2+\big(\alpha-\frac14\big)\frac{1}{x^{2}}
\end{equation}
is very special. They are homogeneous of degree $-2$. They appear in
numerous applications, {\it e.g.}~as the radial part of the Laplacian in any
dimension. Their eigenfunctions can be expressed in terms of Bessel-type
functions, and they have a surprisingly long and intricate theory, see for example \cite{Case,CH,EG,FST,MSS,PP70}.

It is natural to try to interpret \eqref{qfq} as a closed operator on
$L^2(\R_+)$. The  most natural interpretation was
given and extensively studied in \cite{BDG}. It
involves setting $\alpha=m^2$ and considering a family of closed
operators $H_m$ defined for $\Re(m)>-1$, depending on $m$
holomorphically.
In fact, the region $\Re(m)>-1$ can be divided into two parts.
For $\Re(m)\geq1$ the operator $H_m$ corresponds to the closure of
\begin{equation}\label{qfq1}
-\partial_x^2+\big(m^2-\frac14\big)\frac{1}{x^{2}}
\end{equation}
restricted to $C_{\rm c}^\infty(\R_+)$.
For $-1<\Re(m)<1$ it is necessary to specify a boundary condition:
$H_{\pm m}$ are both extensions of
\eqref{qfq1} with elements in their domain behaving like $c x^{\frac12\pm m}$ near $0$
for some $c\in \C$.
Note that all operators $H_m$ are homogeneous of degree $-2$.

For $-1<\Re(m)<1$ more general boundary conditions can be studied.
By considering elements of $L^2(\R_+)$ behaving like $c\big(\kappa x^{1/2-m} +x^{1/2+m}\big)$
for some $c\in \C$, one naturally obtains a two-parameter family of closed
operators $H_{m,\kappa}$, with $\kappa\in\C\cup\{\infty\}$. Note
that these operators are no longer homogeneous, except for $H_{m,0}=H_m$ and
$H_{m,\infty}=H_{-m}$.

A separate analysis is required for $m=0$. Possible boundary
conditions for this case are $c\big(x^{1/2}\ln(x) + \nu x^{1/2}\big)$ for some $c\in \C$.
They lead to a family of closed operators that we denote by $H_0^\nu$, where $\nu\in
\C\cup\{\infty\}$.
Note that $H_0^\infty=H_0$, and that this operator is the only homogeneous one amongst
the operators $H_0^\nu$.

The study of the two families of operators $H_{m,\kappa}$ and $H_0^\nu$
is the object of our paper.
We extensively use their explicit
solvability, so that for these operators we can give exact formulas for
many constructions of operator theory. Let us mention part of
the analysis performed below.

The operators $H_m$ have no point spectrum. However, $H_{m,\kappa}$
and $H_0^\nu$ usually have point spectrum, which coincides with the
discrete spectrum. All the eigenvalues are simple
and depend quite sensitively on the parameters.
The number of these eigenvalues can be finite, but also infinite.
Their position form rather interesting patterns on the
complex plane: typically, it is a sequence situated along a piece of a spiral.

The continuous spectrum always coincides with the positive half-line
$[0,\infty[$. One can express the resolvent of our operators in terms
of the \emph{MacDonald} and \emph{modified Bessel functions} $K_m$ and $I_m$.
The resolvent has boundary values at $[0,\infty[$, which can be expressed in
terms of the \emph{Hankel} and \emph{Bessel functions} $H_m^\pm$ and $J_m$.
We make sense out of these boundary values as bounded operators between
appropriate weighted Hilbert spaces, a property often called the
\emph{Limiting Absorption Principle}. We provide formulas for
the spectral projections onto parts of the continuous spectrum,
and introduce bounded invertible operators that diagonalize
$H_{m,\kappa}$ and $H_0^\nu$, except for a small set of parameters that we call {\em exceptional}.
These operators can be called {\em generalized Hankel
transformations}, see also \cite{Tuan} for related results.
Finally, we describe the scattering theory
for $H_{m,\kappa}$ and $H_0^\nu$, giving formulas for the \emph{M{\o}ller (wave) operators} and
for the \emph{scattering operator}.

Let us mention that most of the operators $H_{m,\kappa}$
and $H_0^\nu$ are not self-adjoint, but that
the subfamilies of self-adjoint ones are also exhibited.
More precisely, $H_{m,\kappa}$ are self-adjoint for real $m$ and $\kappa$, or for
purely imaginary $m$ and $|\kappa|=1$. Similarly, $H_0^\nu$ are
self-adjoint for real $\nu$. Thus our analysis fits into a recent
fashion of studying spectral properties of non-self-adjoint
operators. Indeed, we observe that in the non-self-adjoint cases our operators have
quite interesting discrete spectrum.

We also stress that the operators that we study are extremely natural
and appear in many situations. For example, they describe
oscillations of a conical membrane, the
Aharonov-Bohm effect \cite{BDG,PR}, and sticky diffusion \cite{GH}.
They also play an important role in the study of the wave and Klein-Gordon equations on anti-de Sitter spacetime,
see for example \cite{Ba,Ga,IW} and references therein.
We are convinced that they
possess many more applications we are not aware of. Indeed, since
\eqref{qfq1} is a homogeneous expression and the only nonhomogeneity
is due to boundary conditions, we expect that the studied operators
appear in various scaling limits, or \emph{renormalization group analysis}.

The analysis of our paper can also be considered as a large part of
modern theory of the Bessel equation and the Gamma
function (which belong to the oldest objects of mathematics,
going back at least to Euler in 17th century).
Indeed, many identities for Bessel-type functions and the
Gamma function have a meaning in the theory of $H_{m,\kappa}$
and $H_0^\nu$. Let us give some simple examples: the identity
$z\Gamma(z)=\Gamma(z+1)$ is related to the M{\o}ller operators
for the pair $(H_{m+2}, H_m)$, and the identity
$\Gamma(z)\Gamma(1-z)=\frac{\pi}{\sin\pi z}$ is related to the
M{\o}ller operators for the pair $(H_{-m}, H_m)$.
Note that our point of view on Bessel functions is further developed
in Section \ref{sec_Bessel_eq} and in Appendix \ref{secB1}.

Many authors studied various classes of
one-dimensional Schr\"odinger operators on
the half-line, also called Sturm-Liouville operators on the half-line.
Among classic works on this topic,
which included complex potentials and various boundary conditions, let us mention \cite[Chap.~XX]{DS3} and \cite{Na}.

The methods used in our paper are of course adaptations of known approaches.
Let us also note that the resolvents of $H_{m,\kappa}$, resp.~$H_0^\nu$ are rank one perturbations of the resolvents of the operators $H_m$, resp.~$H_0$.
Therefore our analysis can be interpreted as an example of the theory of singular rank one perturbation, which is a well-studied subject.
Nevertheless, it seems that
a large part of our analysis of the operators $H_{m,\kappa}$ and $H_0^\nu$, especially in the non-self-adjoint case, is new. Let us also mention that
there are still some topics about these operators, which are open. In particular, in this paper we do not analyze fully the \emph{exceptional case},
in which the generalized Hankel transformation are unbounded. In this case \emph{spectral singularities} appear, first noted in a similar context in \cite{Na}.

In our paper
we restrict ourselves to a rather narrow, explicit and solvable family of operators. However, it is not easy to find papers about more general classes of operators
that cover the whole family we consider.
In many works on one dimensional Schr\"odinger operators the singularity $\frac{1}{x^2}$ is excluded by restrictive assumptions. This is e.g.~the case
of \cite{DS3,Na}.
Many papers also assume that the operator is dissipative or accretive (its numerical range is contained in the upper, resp.~lower complex half-plane).
Our operators are often neither dissipative nor accretive. In fact it was noted already in \cite{BDG} that for $\Re( m)<0$ and $\Im(m)\neq0$,
the numerical range of $H_m$ is the whole $\C$.

A vast majority of papers about  Schr\"odinger operators
considers only  self-adjoint cases.
One often assumes the essential self-adjointness on $C_{\rm c}^\infty(\R_+)$.
Then $H_m$ with $|m|<1$ are not seen at all.
If one considers the Friedrichs extension, only the case $m\geq0$ is
covered. If one considers both the Friedrichs and Krein extensions,
one throws away the interesting region $m^2<0$. Looking at the
exactly solvable potential $(m^2-\frac14)\frac1{x^2}$ and using theory
of the Bessel equation,
we can check what are the natural
assumptions in our case. In particular, we can notice that it is natural to
include  non self-adjoint cases, which interpolate in an interesting
way between self-adjoint cases,
and to study a holomorphic family of closed operators depending on two complex parameters.

With this idea in mind, let us compare our work with some of the recent
papers dealing with the operator \eqref{qfq1}.
Note that many papers are related to this operator, and therefore we mention only a few of them
(see also the references in these papers).
First of all, let us mention \cite{GZ} in which an extensive study
of Schr\"odinger operators of the form $-\partial_x^2+V$ with singular $V$
is performed. However, when the special case of $V(x)=\big(\alpha-\frac14\big)\frac{1}{x^{2}}$
is considered, only the parameters $m\geq 1$ are considered.
In \cite{EK} the operator \eqref{qfq1} is also thoroughly analyzed in the range
$m\in [0,1[$ but at the end of the day only the Friedrichs self-adjoint extension
is considered.

In \cite{AHM} the Friedrichs realization $H_m$ of the operator \eqref{qfq1}
with $m\in [0,1[$ is considered and the scattering theory is developed
for the pair $(H_m,H_{\rm D})$ with $H_{\rm D}$ the Dirichlet Laplacian on the half-line.
Again, the set of $m$ considered is rather restrictive.
In addition a sentence like ``for $m\geq 1$ no scattering is possible between $H_{\rm D}$ and $H_m$''
(see page 85 of that paper) is in contradiction with the scattering theory developed
in our Section \ref{sec_homogeneous} and even further extended in the subsequent sections.

In the paper \cite{KT} and in the preprint \cite{AB} the Friedrichs realization
of the operator \eqref{qfq1} is also considered for $m\geq 0$. In the former paper
some dispersive estimates are provided for the evolution group $\{\e^{-itH_m}\}_{t\in \R}$
with an emphasis in the dependence in $m$.
In \cite{AB} a new study of the expression \eqref{qfq1} on finite intervals or on the half-line
is performed with the recently introduced approach of boundary triplets.
The Krein and the Friedrichs extensions are indeed considered on the half-line,
but the parameter $m$ is always real.

Finally, let us mention the recent paper \cite{KTT1} and the related subsequent preprint \cite{KTT2}.
In the former one, dispersive estimates are provided for the operator
$-\partial_x^2+\big(m^2-\frac14\big)\frac{1}{x^{2}}+q(x)$ under some additional conditions on the potential $q$.
Here, only the Friedrichs extension for $m>0$ is considered.
Obviously, the additional potential enlarges the set of operators under investigation,
but on the other hand only the simplest realization of these operators is analyzed.
In the preprint \cite{KTT2} only the initial operator \eqref{qfq1} is considered ({\it i.e.}~$q=0$)
for $m\in ]0,1[$ but a rather large family of self-adjoint realizations of this operator
is introduced. Dispersive estimates for these operators are obtained, and
their dependence on the boundary condition at $0$ is emphasized.

The family of operators $H_m$ with $\Re( m)>-1$ was introduced in
\cite{BDG}. Thus our study of $H_{m,\kappa}$ and $H_0^\nu$
can be viewed as a continuation of \cite{BDG}.
It seems, however, that some of the properties of $H_m$,
notably about scattering in the non-self-adjoint case, are described
in our present paper for the first time.

One could ask how complete our analysis is. In particular,
it is natural to ask whether the operators  $H_m$ can be extended holomorphically outside the domain $-1<\Re(m)$, and $H_{m,\kappa}$ outside
$-1<\Re(m)<1$.
Most probably, if we stick to the framework of the Hilbert space $L^2(\rr_+)$,
the answer is negative. A question about whether $H_m$ can be extended holomorphically was formulated in \cite{BDG} and has not been settled rigorously yet.
However,  one can extend the operators $H_{m,\kappa}$
and $H_m$ to larger domains of parameters, if one goes beyond
the framework  of the Hilbert space $L^2(\rr_+)$.
This subject will be considered in a separate paper.

\section{Inverse square potential}
\setcounter{equation}{0}

\subsection{Notation}\label{sec_notation}

$\C^\times$ denotes $\C\setminus\{0\}$.
$\bar \alpha$ means the complex conjugate of $\alpha\in\C$.
If $Y$ is a set, then $\#Y$ denotes the number of elements of $Y$.
If $\Xi$ is a subset of $\R$, then $\one_\Xi$ represents the
characteristic function of $\Xi$.
$C_{\rm c}^\infty(\R_+)$ denotes
the set of smooth functions on $\R_+:=]0,\infty[$ with compact support.

For an operator $A$, we denote by $\cD(A)$ its domain and by $\sigma_\p(A)$
the set of its eigenvalues (its point spectrum). We also use the notation
$\sigma(A)$ for its spectrum and $\sigma_\ess(A)$ for its essential spectrum.
If $z$ is an isolated point of $\sigma(A)$, then $\one_{\{z\}}(A)$
denotes the Riesz projection of $A$ onto $z$. Similarly, if $A$ is
self-adjoint and $\Xi$ is a Borel subset of $\sigma(A)$, then
$\one_\Xi(A)$ denotes the spectral projection of $A$ onto $\Xi$.

A (possibly unbounded)
operator $A$ on a Hilbert space $\cH$
is invertible (with a bounded inverse) if its null space is $\{0\}$,  its range is $\cH$ and
$A^{-1}$ is bounded.
If $A$ is a positive invertible operator on $\cH$ and $s\geq0$, then
$A^{-s}\cH$ denotes  $\Dom( A^s)$ and $A^s\cH$ denotes its (antilinear)
dual. Thus we obtain a nested scale of Hilbert spaces $A^s\cH$, $s\in\R$.

In the sequel we will usually work with the Hilbert space $L^2(\R_+)$ with the
generic variable denoted by $x$ or $y$, and sometimes also by $k$.
It is equipped with the norm denoted by $\|\cdot\|$, the scalar
product
\begin{equation*}
( f|g):=\int_0^\infty \bar{f(x)}g(x)\d x,
\end{equation*}
as well as the bilinear form
\begin{equation}\label{bili}
\langle f|g\rangle:=\int_0^\infty f(x)g(x)\d x.
\end{equation}
If $B$ is an operator on $L^2(\R_+)$, then $B^*$ denotes the usual
Hermitian adjoint of $B$, whereas
$B^\t$  denotes the  adjoint (the transpose) of $B$ w.r.t.~the
\eqref{bili}. Clearly, if $B$ is a bounded linear operator on $L^2(\R_+)$ with
\begin{equation*}
\big(B f\big)(k):=\int_0^\infty B(k,x)g(x)\d x,
\end{equation*}
then
\begin{equation*}
\big(B^* g\big)(x)=\int_0^\infty\bar{ B (k,x)}g(k)\d k,
\end{equation*}
while
\begin{equation*}
\big(B^\t g\big)(x)=\int_0^\infty B (k,x)g(k)\d k.
\end{equation*}

We shall use the symbol $\Q$ to denote the operator of multiplication
by the variable in $\R_+$, {\it i.e.}~$\big(\Q f\big)(x) = xf(x)$
for $f\in\cD(\Q)\subset L^2(\R_+)$ and $x\in \R_+$.
Note that if it is clear that the name of the variable is $x$,
we will also write $x$ instead of $\Q$.
We will often use the scale of Hilbert spaces based on the operator $\langle \Q\rangle:=(1+\Q^2)^{1/2}$,
denoted by $\langle\Q\rangle^{-s}L^2(\R_+)$.
The Sobolev spaces
$\H_0^1(\R_+)$ and  $\H^1(\R_+)$ are the subspaces of
$L^2(\R_+)$ defined as the form domain of the Dirichlet and Neumann
Laplacian respectively.

We will also consider the unitary group $\{U_\tau\}_{\tau\in\R}$ of dilations
acting on $f\in L^2(\R_+)$
as $\big(U_\tau f\big)(x) = \e^{\tau/2}f(\e^\tau x)$.
An operator $B$ on $L^2(\R_+)$ is said to be
\emph{homogeneous of degree $\alpha \in \R$} if
$U_{\tau}\Dom(B)\subset \Dom(B)$ for any $\tau \in \R$ and if the
equality $U_{\tau}BU_{\tau}^{-1} = \e^{\alpha \tau}B$ holds on
$\Dom(B)$. For instance, $\Q$ is an operator of degree $1$.
The generator of dilations is denoted by $A$, so that
$U_\tau=\e^{\i\tau A}$. On suitable functions $f$
one has
$$
Af(x)=\frac{1}{2\i}(\partial_x x+x\partial_x)f(x).
$$

The following holomorphic functions are understood as their \emph{principal
bran\-ches}, that is, their domain is $\C\setminus]-\infty,0]$ and on
$]0,\infty[$ they coincide with their usual definitions from real analysis:
$\ln(z)$,  $\sqrt z$, $z^\lambda$, $\arg (z):=\Im \big(\ln(z)\big)$. $\Ln(z)$ will
denote the multivalued logarithm. This means, if $w_0$
satisfies $\e^{w_0}=z$, then $\Ln(z)$ is the \emph{set}
$w_0+2\pi\i\Z$.

The Wronskian $W(f,g)$ of two continuously differentiable functions
$f,g$ on $\R_+$ is given by
the expression
\begin{equation}\label{eq_Wronsk}
W_x(f,g)\equiv W(f,g)(x):=f(x)g'(x)-f'(x)g(x).
\end{equation}

\subsection{Maximal and minimal homogenous Schr\"odinger operators}

For any $\alpha \in \C$ we consider the differential expression
\begin{equation*}
L_\alpha :=-\partial_x^2+\big(\alpha-\frac14\big)\frac{1}{x^{2}}
\end{equation*}
acting on distributions on $\R_+$, and denote by $L_\alpha ^{\min}$ and $L_\alpha ^{\max}$
the corresponding minimal and maximal operators
associated with it in $L^2(\R_+)$, see \cite[Sec.~4 \& App.~A]{BDG}
for details. We simply recall from this reference that
\begin{equation*}
\Dom(L_\alpha ^{\max}) = \{f\in L^2(\R_+)\mid L_\alpha  f\in L^2(\R_+)\}
\end{equation*}
and that $\Dom(L_\alpha ^{\min})$ is
the closure of the restriction of $L_\alpha$ to
$C_{\rm c}^\infty(\R_+)$.
Clearly, both $L_\alpha ^{\min}$ and $L_\alpha ^{\max}$ are homogeneous of degree
$-2$.
Note that from now on we shall simply say \emph{homogeneous},
without specifying the degree $-2$. In addition, the following relation holds:
\begin{equation*}
\big(L_\alpha ^{\min}\big)^* = L_{\bar \alpha}^{\max}.
\end{equation*}

Let us recall some additional results which have been obtained in \cite[Sec.~4]{BDG}.
For that purpose, we say that $f\in \Dom(L_\alpha ^{\min})$ around
$0$, (or, by an abuse of notation, $f(x)\in \Dom(L_\alpha ^{\min})$ around $0$) if there exists
$\zeta\in C_{\rm c}^\infty\big([0,\infty[\big)$ with
$\zeta=1$ around $0$ such that $f\zeta\in
\Dom(L_\alpha ^{\min})$.
In addition, it turns out that it is useful to introduce a parameter $m\in\C$ such
that $\alpha=m^2$, even though there are two $m$ corresponding to a single $\alpha\neq0$.
In other words, we shall consider from now on the operator
\begin{equation*}
L_{m^2}:=-\partial_x^2+\big(m^2-\frac14\big)\frac{1}{x^{2}}.
\end{equation*}
With this notation, if $|\Re(m)|\geq 1$ then $L_{m^2}^{\min} = L_{m^2}^{\max}$, while if $|\Re(m)|<1$ then
$L_{m^2}^{\min} \subsetneq L_{m^2}^{\max}$ and $\Dom(L_{m^2}^{\min})$ is a closed subspace of codimension $2$ of $\Dom(L_{m^2}^{\max})$.
More precisely, if $|\Re(m)|<1$ and if $f\in \Dom(L_{m^2}^{\max})$, then there exist $a,b\in \C$
such that:
\begin{align*}
f(x)  - ax^{1/2-m} - bx^{1/2+m}&\in\Dom(L_{m^2}^{\min})\hbox{
around }0  \qquad \hbox{ if } m \neq 0, \\
f(x)-ax^{1/2}\ln(x)- bx^{1/2}&\in\Dom(L_{0}^{\min})\hbox{
around }0.
\end{align*}
In addition, the behavior of any function $g\in \Dom(L_{m^2}^{\min})$ is known, namely $g\in \H^1_0(\R_+)$ and as $x\to 0$:
\begin{align*}
g(x) & = o\big(x^{3/2}\big)\ \hbox{ and }\ g'(x)= o\big(x^{1/2}\big)  \quad \hbox{ if } m \neq 0, \\
g(x) & = o\big(x^{3/2}\ln(x)\big)\ \hbox{ and }\ g'(x)= o\big(x^{1/2}\ln(x)\big)  \quad \hbox{ if } m = 0.
\end{align*}

\subsection{Two families of  Schr\"odinger operators with inverse square potentials}

Let us first recall from \cite[Def.~4.1]{BDG} that for any $m\in\C$ with $\Re (m)>-1$
the operator $H_m$ has been defined as the
restriction of $L_{m^2}^{\max}$ to the domain
\begin{align*}
\Dom(H_{m}) & = \big\{f\in \Dom(L_{m^2}^{\max})\mid
\hbox{ for some }  c \in \C,\\
&\qquad f(x)- c
x^{1/2+m}\in\Dom(L_{m^2}^{\min})\hbox{ around }
0\big\}.
\end{align*}
It is then proved in this reference that $\{H_m\}_{\Re(m)>-1}$
is a holomorphic family of closed homogeneous operators in $L^2(\R_+)$.
In addition, if $\Re(m)\geq 1$, then
\begin{equation*}
H_m=L_{m^2}^{\min}=L_{m^2}^{\max}.
\end{equation*}
For this reason, we shall concentrate on the case $-1 <\Re(m)<1$,
considering a larger family of operators.

For $|\Re(m)|<1$ and for any $\kappa\in \C\cup \{\infty\}$ we define a
family of operators $H_{m,\kappa}$~:
\begin{align}
\label{eq_do1}
\Dom(H_{m,\kappa}) & = \big\{f\in \Dom(L_{m^2}^{\max})\mid
\hbox{ for some }  c \in \C,\\
&\qquad f(x)- c
\big(\kappa x^{1/2-m} +
x^{1/2+m}\big)\in\Dom(L_{m^2}^{\min})\hbox{ around }
0\big\},\qquad\kappa\neq\infty; \nonumber \\
\label{eq_do2}
\Dom(H_{m,\infty}) & = \big\{f\in \Dom(L_{m^2}^{\max})\mid
\hbox{ for some }  c \in \C,\\
&\qquad f(x)- c
x^{1/2-m}\in\Dom(L_{m^2}^{\min})\hbox{ around } 0\big\}.\nonumber
\end{align}

For $m=0$, we introduce an additional family of operators $H_0^\nu$ with
$\nu\in\C\cup \{\infty\}$~:
\begin{align}
\label{eq_do3}
\Dom(H_0^\nu) & = \big\{f\in \Dom(L_{0}^{\max})\mid
\hbox{ for some }  c \in \C,\\
&\qquad f(x)- c
\big( x^{1/2}\ln(x) + \nu x^{1/2}\big)\in\Dom(L_0^{\min})\hbox{ around } 0\big\},\qquad\nu\neq\infty;\nonumber \\
\label{eq_do4}
\Dom(H_0^\infty) & = \big\{f\in \Dom(L_0^{\max})\mid
\hbox{ for some }  c \in \C,\\
&\qquad f(x)- c
x^{1/2}\in\Dom(L_0^{\min})\hbox{ around }
0\big\}.\nonumber
\end{align}

The following properties of these families of operators are immediate:

\begin{lemma}
\begin{enumerate}
\item[(i)] For any $|\Re(m)|<1$ and any $\kappa\in \C\cup \{\infty\}$,
\begin{equation}\label{Eq_duality}
H_{m,\kappa}=H_{-m,\kappa^{-1}}.
\end{equation}
\item[(ii)] The operator $H_{0,\kappa}$ does not depend on $\kappa$,
and all these operators coincide with $H_0^\infty$.
\end{enumerate}
\end{lemma}

As a consequence of (ii),
all the results about the case $m=0$
will be formulated in terms of the family $H_0^\nu$.

Let us now derive two simple results for this family of operators. The first one is related to the action
of the dilation group, while the second is dealing with
the Hermitian conjugation.

\begin{proposition}
For any $m$ with $|\Re(m)|<1$ and any $\kappa,\nu \in \C\cup\{\infty\}$,
we have
\begin{align*}
U_\tau H_{m,\kappa}U_{-\tau}&=\e^{-2\tau}H_{m,\e^{ -2\tau m}\kappa},\\
U_\tau H_0^\nu U_{-\tau}&=\e^{-2\tau}H_0^{\nu +\tau},
\end{align*}
with the convention that $\alpha \;\!\infty=\infty$ for any $\alpha
\in \C\setminus\{0\}$ and $\infty + \tau = \infty$ for any $\tau \in \C$.
In particular,
\begin{enumerate}
\item[(i)] Amongst the family of operators $H_{m,\kappa}$ with
$m\neq0$, only
\begin{equation*}
H_{m,0} = H_{m} \qquad \hbox{ and } \qquad
H_{m,\infty} =  H_{-m}
\end{equation*}
are homogeneous,
\item[(ii)] Amongst the family $H_0^\nu$, only
\begin{equation*}
H_0^\infty = H_0
\end{equation*}
is homogeneous.
\end{enumerate}
\end{proposition}

\begin{proof}
If $f\in \Dom(H_{m,\kappa})$, then $U_\tau f \in \Dom(H_{m,e^{-2m\tau}\kappa})$. Thus, the only domains which are left invariant are
$\Dom(H_{m,0})$ and $\Dom(H_{m,\infty})$. Since $L_{m^2}^{\max}$ is homogeneous, the same applies for $H_{m,0}$ and $H_{m,\infty}$.

If $f\in \Dom(H_0^\nu)$, then $U_\tau f\in
\Dom(H_0^{\nu+\tau})$. Thus, only $\Dom(H_0^\infty)$ is left
invariant,
and consequently only $H_{0}^\infty$ is homogeneous.
\end{proof}

\begin{proposition}
For any $m\in \C$ with $|\Re(m)|<1$ and for any $\kappa, \nu \in \C\cup \{\infty\}$, one has
\begin{equation}\label{Eq_adjoint}
(H_{m,\kappa})^*=H_{\bar m,\bar\kappa}\qquad \hbox{ and }\qquad
(H_0^\nu)^* = H_0^{\bar \nu}
\end{equation}
with the convention that $\bar \infty = \infty$.
\end{proposition}

\begin{proof}
Recall from \cite[App.~A]{BDG} that for any $f\in \Dom(L_{m^2}^{\max})$ and $g\in \Dom(L_{\bar m^2}^{\max})$,
the functions $f,f',g,g'$ are continuous on $\R_+$, and that the equality
\begin{equation*}
( L_{m^2}^{\max}f|g) - ( f|L_{\bar m^2}^{\max}g) = -W_0(\bar f,g)
\end{equation*}
holds with $W_0(\bar f,g) = \lim\limits_{x\to 0}W_x(\bar f,g)$ and $W_x$ defined in \eqref{eq_Wronsk}.
In particular, if $f\in \Dom(H_{m,\kappa})$, one infers that
\begin{equation*}
( H_{m,\kappa}f|g) = ( f|L_{\bar m^2}^{\max}g)  -W_0(\bar f,g).
\end{equation*}
Thus, $g\in \Dom\big((H_{m,\kappa})^*\big)$ if and only if $W_0(\bar f,g)=0$, and then
$(H_{m,\kappa})^*g=L_{\bar m^2}^{\max}g$.
Then, by taking into account the explicit description of $\Dom(H_{m,\kappa})$,
straightforward computations show that $W_0(\bar f,g)=0$ if and only
if $g\in \Dom(H_{\bar m,\bar \kappa})$. One then deduces that
$(H_{m,\kappa})^*= H_{\bar m,\bar \kappa}$.

A similar computation leads to the equality $(H_0^\nu)^* = H_0^{\bar \nu}$.
\end{proof}

\begin{corollary}\label{corol_SA}
\begin{enumerate}
\item[(i)] The operator $H_{m,\kappa}$ is self-adjoint for $m\in]-1,1[$ and $\kappa\in\R\cup\{\infty\}$,
and for $m\in \i\R$ and $|\kappa|=1$.
\item[(ii)] The operator $H_0^\nu$ is self-adjoint for $\nu\in \R\cup \{\infty\}$.
\end{enumerate}
\end{corollary}

\begin{proof}
For the operators $H_{m,\kappa}$ one simply has to take formula \eqref{Eq_adjoint} into account for the first case,
and the same formula together with \eqref{Eq_duality} in the second case.
Finally for the operators $H_0^\nu$, taking formula \eqref{Eq_adjoint} into account leads directly to the result.
\end{proof}

\begin{remark} By blowing up around $m=0$, it is possible to make
one single holomorphic function out of $H_{m,\kappa}$ and $H_0^\nu$.
In fact, we can extend the function $H_0^\nu$ by setting
\begin{equation*}
H_m^\nu:=H_{m,\frac{\nu m-1}{\nu m+1}},\ \ \ m\neq0.
\end{equation*}
Then $H_m^\nu$ is holomorphic for $1<\Re(m)<1$, $\nu\in\C\cup\{\infty\}$ and covers
all values of $H_{m,\kappa}$. To see the holomorphy at $m=0$ (at least at the
level of the boundary conditions) note that for $|m|$ very small one has
\begin{equation*}
\kappa x^{1/2-m} +
x^{1/2+m}\approx m(1-\kappa)\Big(
x^{1/2}\ln(x) +\frac{1+\kappa}{m(1-\kappa)}x^{1/2}\Big),
\end{equation*}
which is $c\big(x^{1/2}\ln(x)+\nu x^{1/2}\big)$ for
$\kappa=\frac{\nu m-1}{\nu m+1}$.
\end{remark}

\section{Bessel-type equations and functions}\label{sec_Bessel_eq}
\setcounter{equation}{0}

In this section we establish the link between our initial operator \eqref{qfq}
and different forms of the Bessel equation. By considering the dependence of the space dimension in these equations,
one is naturally led to introduce a new basic family of Bessel-type functions.

Let us recall that the Laplace operator in $d$ dimensions and in spherical coordinates is given by
\begin{equation*}
-\Delta_d=-\partial_r^2-\frac{d-1}r\partial_r-\frac1{r^2}\Delta_{\S^{d-1}},
\end{equation*}
where $r$ is the radial coordinate and $\Delta_{\S^{d-1}}$ is the
{\em Laplace-Beltrami operator on the sphere $\S^{d-1}$}.
For simplicity, we also use in this section the notation $r$
for the corresponding operator of multiplication by the radial coordinate.
Eigenvalues of $-\Delta_{\S^{d-1}}$ for $d=2,3,\dots$ are
\begin{equation}\label{spek}
l(l+d-2),\ \ l\in \{0,1,2,\dots\},
\end{equation}
where $l$ corresponds to the order of spherical harmonics.

Note that in the special case $d=2$, one has
$\Delta_{\S^1}=\partial_\phi^2$ and the \emph{angular momentum operator}
$-\i\partial_\phi$ has eigenvalues $m\in\zz$. As a consequence
\eqref{spek} can also be written as
\begin{equation*}
m^2,\ \ m\in\zz.
\end{equation*}
For $d=1$ the sphere reduces to a pair of points and $\Delta_{\S^0}$ corresponds to the $\bf 0$-operator.
It has the eigenvalue $0$ of multiplicity $2$, and this corresponds to the values $l=0$ and $l=1$ in \eqref{spek}.

Thus, if one sets $m:=l+\frac{d}{2}-1$ then the radial part of the Laplacian takes the following form in any dimension
\begin{align}\label{lap2}
\nonumber &-\partial_r^2-\frac{d-1}{r}\partial_r+l(l+d-2)\frac{1}{r^2}\\
& = -\partial_r^2-\frac{d-1}{r}\partial_r + \Big(m^2-\big(\frac{d}{2}-1\big)^2\Big)\frac{1}{r^2}
\end{align}
with an appropriately restricted range of $m$.
In what follows we will call \eqref{lap2} the \emph{Bessel operator of
dimension $d$}, allowing then the parameter $m$
to take arbitrary complex values.
In particular, the Bessel operator of dimension $2$ is
\begin{equation}\label{lap3}
-\partial_r^2-\frac{1}{r}\partial_r+\frac{m^2}{r^2},
\end{equation}
while the Bessel operator of dimension $1$ is
\begin{equation}\label{lap4}
-\partial_r^2+\Big(m^2-\frac14\Big)\frac{1}{r^2}.
\end{equation}

Operators \eqref{lap2} are related to one another for different $d$ by a
simple similarity transformation. Indeed, by a short computation performed on $C_{\rm c}^\infty(\R_+)$
one easily observes that the following equality hold:
\begin{align*}
& -\partial_r^2-\frac{d-1}{r}\partial_r+ \Big(m^2-\big(\frac{d}{2}-1\big)^2\Big)\frac{1}{r^2}\\
&=r^{-\frac{d}{2}+1}\left(
-\partial_r^2-\frac{1}{r}\partial_r+
\frac{m^2}{r^2}\right)
r^{\frac{d}{2}-1}\\
&=r^{-\frac{d}{2}+\frac12}\left(
-\partial_r^2+
\Big(m^2-\frac14\Big)\frac{1}{r^2}\right)
r^{\frac{d}{2}-\frac12}.
\end{align*}
It is then a matter of taste, convenience and historical circumstances
whether the operator  \eqref{lap3} or \eqref{lap4} is taken as
the basic one. In the literature, at least since the times of
Bessel, it
seems that \eqref{lap3} has a distinguished status. We prefer
\eqref{lap4}, at least in the context of this paper, since our initial operator \eqref{qfq}
has a form similar to \eqref{lap4}.

A simple scaling argument shows that the eigenvalue problem for
\eqref{lap3} can be reduced to one of the following two equations:
\begin{align}
\rm{the}\ \emph{modified Bessel equation}&&\left(\partial_r^2+\frac{1}{r}\partial_r-\frac{m^2}{r^2}-1\right)v& = 0,
\label{lap5}\\\rm{the}\ \emph{(standard) Bessel equation}&&
\left(\partial_r^2+\frac{1}{r}\partial_r-\frac{m^2}{r^2}+1\right)v& = 0.
\label{lap6}
\end{align}
Certain distinguished solutions
of \eqref{lap5} are
\begin{align*}
\rm{the}\ \emph{modified Bessel function}&&I_m(z),\\
\rm{the}\ \emph{MacDonald function}&&K_m(z),
\end{align*}
and of \eqref{lap6} are
\begin{align*}\rm{the}\
\emph{Bessel function}&&J_m(z),\\
\rm{the}\ \emph{Hankel function of the 1st kind}&&H_m^+(z)=H_n^{(1)}(z),\\
\rm{the}\ \emph{Hankel function of the 2nd kind}&&H_m^-(z)=H_n^{(2)}(z),\\
\rm{the}\ \emph{Neumann function}&&Y_m(z).
\end{align*}
We call them jointly
\emph{the Bessel family}. They are probably the best known and the most
widely used special functions in mathematics and its applications \cite{AS,AAR,GR,W}.

\begin{remark} The notation $H_m^{(1)}$, $H_m^{(2)}$ for the two kinds of Hankel functions is more common in the literature. We use the notation
$H_m^+$, $H_m^-$, which is much more convenient.
\end{remark}

Instead of \eqref{lap5} and \eqref{lap6} we will prefer
to consider their analogs \emph{for dimension $1$}, namely
\begin{align}
\left(\partial_r^2-\Big(m^2-\frac14\Big)\frac{1}{r^2}-1\right)v& = 0,
\label{lap7}\\
\left(\partial_r^2-\Big(m^2-\frac14\Big)\frac{1}{r^2}+1\right)v& = 0.
\label{lap8}
\end{align}
We will also introduce new special functions
that solve \eqref{lap7}
\begin{align*}\rm{the}\
\emph{modified Bessel function for dimension $1$}&&\Ia_m(r):=\sqrt{\frac{\pi r}{2}} I_m(r),\\
\rm{the}\ \emph{MacDonald function for dimension $1$}&&\Ka_m(r):=\sqrt{\frac{2 r}{\pi}} K_m(r),
\end{align*}
and new special functions that solve \eqref{lap8}
\begin{align*}\rm{the}\
\emph{Bessel function for dimension $1$}&&\Ja_m(r):= \sqrt{\frac{\pi r}{2}} J_m(r),\\
\rm{the}\ \emph{Hankel function of the 1st kind for dimension $1$}&&\Ha_m^+(r):=\sqrt{\frac{\pi r}{2}} H^+_m(r),
\\
\rm{the}\ \emph{Hankel function of the 2nd kind for dimension $1$}&&\Ha_m^-(r):=\sqrt{\frac{\pi r}{2}} H^-_m(r),\\
\rm{the}\
\emph{Neumann function for dimension $1$}&&\Ya_m(r):= \sqrt{\frac{\pi r}{2}} Y_m(r).
\end{align*}

Jointly, they will be
called \emph{the Bessel family for dimension $1$}.
Accordingly, the Bessel family should be called \emph{the Bessel family for dimension $2$}.
As we shall show later on, the Bessel family for dimension $1$
contains a number of standard  elementary functions: the exponential
function, the trigonometric sine and cosine functions, and the hyperbolic sine and cosine functions.

Obviously, properties of the Bessel family for dimension $1$ can
easily be deduced from the corresponding properties of
the Bessel family for dimension $2$, and the other way round.
Most properties seem to have a simpler
form in the case of dimension $1$ than in the case of dimension $2$. There
are some exceptions, mostly involving integer values of $m$ which
play a distinguished role in dimension $2$ and, more generally, in even dimensions. We collect basic properties of the Bessel family for dimension $1$ in
Appendix \ref{secB1}

Let us note that in the literature one sometimes introduces the
modified and standard Bessel equations for all dimensions $d\geq2$,
as well as the corresponding functions, see for example \cite{M}.
In particular, there exists a standard notation for the functions of the Bessel family
for dimension $3$: $i_m$, $k_m$, $j_m$, $h_m^\pm$ and $y_m$. In our opinion,
however, this introduces too many unnecessary special functions:
the main issue is whether the dimension is even or odd.
One could argue that the Bessel family for dimension $2$,
that is $I_m$, $K_m$, $J_m$, $H_m^\pm$ and $Y_m$,
is better adapted for all even dimensions and integer values of $m$,
whereas the Bessel family for dimension $1$, that is $\Ia_m$, $\Ka_m$,
$\Ja_m$,  $\Ha_m^\pm$ and $\Ya_m$,
is better adapted for odd dimensions, as well
as for general values of $m$.

\section{The homogeneous case}\label{sec_homogeneous}
\setcounter{equation}{0}

In this section we consider the homogeneous operators $H_m$.
Part of the following results were proved in \cite{BDG}, but
we add some new material.  In particular,
some statements and proofs in \cite{BDG}
were restricted to the self-adjoint case. We extend them to all $\Re(m)>-1$.
In addition we shall now express everything in terms of the Bessel
family for dimension~$1$
instead of the Bessel family for dimension~$2$.

\subsection{Resolvent}

In \cite[Sec.~4.2]{BDG} the resolvent of
$H_m$ is constructed. For completeness we recall its construction below:

\begin{theoreme}
For any $m\in \C$ with
$\Re(m)>-1$ the spectrum of $H_m$ is $[0,\infty[$.
In addition, for $k\in \C$ with $\Re(k)>0$ the resolvent
\begin{equation*}
R_{m}(-k^2):=(H_m+k^2)^{-1}
\end{equation*}
has the kernel
\begin{equation}\label{eq_kernel_H_m}
R_m(-k^2;x,y)=
\frac{1}{k}\left\{\begin{matrix}\! \Ia_m(kx)\;\! \Ka_m(ky) & \hbox{ if } 0 < x < y, \\
\Ia_m(ky)\;\! \Ka_m(kx) & \hbox{ if } 0 < y < x.
\end{matrix}\right.
\end{equation}\end{theoreme}

\begin{proof}[Sketch of proof provided in \cite{BDG}]
It is first checked that the kernel provided by \eqref{eq_kernel_H_m}
defines a bounded operator which we denote by $R_m(-k^2)$. Then it is verified that
\begin{equation*}
\Big(\big( L_{m^2}+k^2\big)R_m(-k^2)\Big)(x,y)=\delta(x-y).
\end{equation*}
Next one has $R_m(-k^2)f\in \Dom(H_m)$ for any $f\in C_{\rm c}^\infty(\R_+)$.
Thus the previous equality can be reinterpreted as
$(H_m+k^2)R_m(-k^2)=\one$.
In addition, since $H_m^*=H_{\bar m}$ and $R_m(-k^2)^*=R_{\bar m}(-\bar k^2)$
one then infers that
\begin{equation*}
R_m(-k^2)( H_m+k^2)=\Big(( H_{\bar m}+\bar k^2)R_{\bar m}(-\bar k^2)\Big)^*= \one.
\end{equation*}
Therefore, $-k^2$ belongs to the resolvent set of $H_m$ and
$R_m(-k^2)$ is the resolvent of $H_m$.
\end{proof}

For $|\Re(m)|<1$, we also introduce the operator $P_m(-k^2)$ defined by its kernel
\begin{align}
P_m(-k^2;x,y) :=& \frac{ \sin (\pi m)}{m}k\;\!\Ka_m(kx)\;\!\Ka_m(ky)\quad \hbox{ if }m\neq0,
\label{proj1}\\
P_0(-k^2;x,y) :=&  \pi k \;\!\Ka_0(kx)\;\!\Ka_0(ky).\label{proj2}
\end{align}
By the bounds \eqref{tri4a} and \eqref{eq:besselfunc2_1},
the function $x\mapsto \Ka_m(kx)$ is square integrable.
Taking also \eqref{int3}, \eqref{int4} into account one easily infers that the
operator $P_m(-k^2)$ is a rank one projection. It is orthogonal if
$k$ and $m^2$ are real.
By taking the equality $\Ka_m = \Ka_{-m}$ and \eqref{macdo1} into account,
one can also deduce the following relation
for $|\Re(m)|<1$:
\begin{equation}\label{eq_dif_res_m}
R_{-m}(-k^2)-R_{m}(-k^2) =  \frac{m}{k^2}P_m(-k^2).
\end{equation}

\subsection{Boundary value of the resolvent and spectral density}\label{sec_spec_1}

In this section we show that a Limiting Absorption Principle
holds for the operators $H_m$. We also compute the kernels of the
boundary values of the resolvent and of the spectral density.

\begin{theoreme}\label{thm_boundary}
Let $m\in \C$ with $\Re(m)>-1$, and let $k>0$.
Then the \emph{boundary values of the resolvent}
\begin{equation*}
R_{m}(k^2\pm\i0)  :=\lim_{\epsilon\searrow0}R_{m}(k^2\pm\i\epsilon)
\end{equation*}
exist in the sense of operators from $\langle\Q\rangle^{-s}L^2(\R_+)$
to $\langle\Q\rangle^{s}L^2(\R_+)$ for any $s>\frac{1}{2}$,
uniformly in $k$ on each compact subset of $\R_+$. They have the kernels
\begin{equation}\label{eq_boundary}
R_m(k^2\pm \i 0;x,y)=\pm \frac{\i}{k}\left\{\begin{matrix}\! \Ja_m(kx)\;\!
\Ha_m^\pm(ky) & \hbox{ if } \ 0 < x \leq y, \\
\Ja_m(ky)\;\! \Ha_m^\pm (kx) & \hbox{ if } \ 0 < y < x.
\end{matrix}\right.
\end{equation}
\end{theoreme}

The above theorem describes a property that, at least in the
context of self-adjoint Schr\"odinger operators, is often called the
\emph{Limiting Absorption Principle}. Its proof will be based on an explicit
estimate of the resolvent kernel:

\begin{proposition}\label{prop_boundary4}
Let us consider $\Re(k) > 0$.
Then for $\Re(m)\geq 0$  with $m\neq 0$ one has
\begin{equation}\label{pwe}
|R_m(-k^2;x,y)| \leq
\frac{C_m^2}{|k|}\;\!\e^{-\Re(k)|x-y|}
\min(1,|xk|)^{\frac12}\min(1,|yk|)^{\frac12},
\end{equation}
for $-1< \Re(m)\leq 0$ with $m\neq 0$ one has
\begin{equation}\label{pwe1}
|R_m(-k^2;x,y)| \leq  \frac{C_m^2}{|k|}\;\!\e^{-\Re(k)|x-y|}
\min(1,|xk|)^{\Re(m)+\frac12}\min(1,|yk|)^{\Re(m)+\frac12},
\end{equation}
while in the special case $m=0$ one has
\begin{align}\label{pwe2}
\nonumber |R_0(-k^2;x,y) \leq & \frac{C_0^2}{|k|}\;\!\e^{-\Re(k)|x-y|} \min(1,|xk|)^{\frac12}\min(1,|yk|)^{\frac12} \\
& \ \times \big(1+\big|\ln(\min(1,|kx|))\big|\big)\;\! \big(1+\big|\ln(\min(1,|ky|))\big|\big).
\end{align}
The constants $C_m$ and $C_0$ are independent of $x,y$ and $k$.
\end{proposition}

\begin{proof}
The proof is based on the following estimates on the Bessel and MacDonald functions.
For $\epsilon>0$ and $|\arg z|<\pi-\epsilon$ one has
\begin{align}
\nonumber |\Ka_m(z)|&\leq C_m\e^{-\Re(z)}\min(1,|z|)^{-|\Re(m)|+\frac12} \qquad  m\neq0,\\
\label{eq_K_0} |\Ka_0(z)|&\leq C_0\e^{-\Re(z)}\min(1,|z|)^{\frac12}\big(1+\big|\ln(\min(1,|z|))\big|\big),
\end{align}
and
\begin{align*}
|\Ia_m(z)|&\leq C_m\e^{|\Re(z)|}\min(1,|z|)^{\Re(m)+\frac12}\qquad m\neq0,\\
|\Ia_0(z)|&\leq  C_0\e^{|\Re(z)|}\min(1,|z|)^{\frac12}.
\end{align*}

By using $\Re(k) > 0$ and for $m\neq 0$, observe first that for $0<x<y$ we obtain
\begin{align*}
& |R_m(-k^2;x,y)| \\
&\leq \frac{C_m^2}{|k|}\e^{x\Re(k)}\e^{-y\Re(k)}\min(1,|kx|)^{\Re(m)+\frac12}
\min(1,|ky|)^{-|\Re(m)|+\frac12},
\end{align*}
while for $0<y<x$ we have
\begin{align*}
& |R_m(-k^2;x,y)| \\
&\leq \frac{C_m^2}{|k|}\e^{y\Re(k)}\e^{-x\Re(k)}\min(1,|kx|)^{-|\Re(m)|+\frac12}
\min(1,|ky|)^{\Re (m)+\frac12}\ .
\end{align*}

If $\Re(m)\geq 0$ one observes then that $\frac{|kx|}{|ky|}<1$ in the first case, while
$\frac{|ky|}{|kx|}<1$ in the second case. This directly leads to \eqref{pwe}.
On the other hand for $\Re(m)<0$, one has $-|\Re(m)|=\Re(m)$, from which one infers \eqref{pwe1}.
 Finally, the special case $m=0$ is obtained by a straightforward computation.
\end{proof}

\begin{proof}[Proof of Theorem \ref{thm_boundary}]
Define the operator $R_m(k^2\pm\i0)$ by its kernel \eqref{eq_boundary}.
We will show that
\begin{equation}\label{boundary2}
\langle \Q\rangle^{-s}R_{m}(k^2\pm\i\epsilon) \langle \Q\rangle^{-s},
\end{equation}
whose kernel is
\begin{equation}\label{boundary3}
\langle x\rangle^{-s}R_{m}(k^2\pm\i\epsilon;x,y) \langle y\rangle^{-s},
\end{equation}
is a Hilbert-Schmidt operator and converges as $\epsilon\searrow0$ in the Hilbert-Schmidt norm to
\begin{equation}\label{boundary5}
\langle \Q\rangle^{-s}R_{m}(k^2\pm\i0) \langle\Q\rangle^{-s}.
\end{equation}

Consider first the (slightly more difficult) case $-1<\Re (m)\leq 0$ with $m\neq 0$. By
the estimate \eqref{pwe1} the expression \eqref{boundary3} can be bounded by
\begin{align}
&\frac{C}{|k|}\e^{-\Re (\sqrt{-k^2\mp\i\epsilon})|x-y|} \langle x\rangle^{-s}
\langle y\rangle^{-s} \min(1,|xk|)^{\Re (m)+\frac12}\min(1,|yk|)^{\Re (m)+\frac12}
\nonumber
\\
&\leq \frac{C}{|k|} \langle x\rangle^{-s} \langle y\rangle^{-s}
\min(1,|xk|)^{\Re (m)+\frac12}\min(1,|yk|)^{\Re (m)+\frac12},
\label{pwe3}
\end{align}
where $C$ is a constant independent of $x,y$ and $k$.
Note that in the computation the inequality $\Re( \sqrt{-k^2\mp\i\epsilon})\geq0$ has been used,
and that such an inequality holds by our choice of the principal branch of the square root.
One clearly infers that \eqref{pwe3} belongs to $L^2(\R_+\times\R_+)$ and dominates \eqref{boundary3}.
Since \eqref{boundary3} converges pointwise to
\begin{equation}\label{boundary4}
\langle x\rangle^{-s}R_{m}(k^2\pm\i0;x,y) \langle y\rangle^{-s},
\end{equation}
one concludes by the Lebesgue Dominated Convergence Theorem that
\eqref{boundary3} converges in $L^2(\R_+\times\R_+)$  to \eqref{boundary4}.
This is equivalent to the convergence of \eqref{boundary2} to  \eqref{boundary5} in the
Hilbert-Schmidt norm.
Note finally that the uniform convergence in $k$ on each compact subset of $\R_+$
can be checked directly on the above expressions.

For $\Re(m) \geq 0$ with $m\neq 0$, the same proof holds with the estimate \eqref{pwe} instead of \eqref{pwe1}.
Finally for $m=0$, the result can be obtained by using \eqref{pwe2}, and by observing that the factor with the logarithm
is also square integrable near the origin.
\end{proof}

Based on the previous theorem one can directly deduce the following statement.

\begin{proposition}
There exists the \emph{spectral density}
\begin{align*}
p_{m}(k^2):=&\lim_{\epsilon\searrow0}  \frac{1}{2\pi \i}\Big(R_{m}(k^2+\i\epsilon)-R_{m}(k^2-\i\epsilon)\\
=& \frac{1}{2\pi \i}\Big(R_{m}\big(k^2+\i0\big)-R_{m} \big(k^2-\i0\big)\Big),
\end{align*}
understood in the sense of operators from
 $\langle\Q\rangle^{-s}L^2(\R_+)$
to $\langle\Q\rangle^{s}L^2(\R_+)$ for any $s>\frac{1}{2}$.
The kernel of this operator is provided for $x,y\in \R_+$ by
\begin{equation*}
p_{m}(k^2;x,y) =\ \frac{\Ja_m(k x)\Ja_m(ky)}{\pi k}.
\end{equation*}
\end{proposition}

\subsection{Hankel transformation}\label{sec_Hankel}

This section is mostly inspired from Sections 5 and 6 of \cite{BDG}, from which
most of the statements are borrowed. We refer to this reference for more details,
or to the subsequent sections of the present paper for a more general approach.

For any $m\in \C$ with $\Re(m)>-1$  let us set
\begin{equation*}
\cF_{m}:C_{\rm c}(\R_+)\to L^2(\R_+)
\end{equation*}
with
\begin{equation*}
\big(\cF_{m} f\big)(x):=\int_0^\infty\cF_{m} (x,y)f(y)\d y
\end{equation*}
and
\begin{equation}\label{eq_kernel_Jm}
\cF_{m} (x,y) := \sqrt{\frac2\pi}\Ja_m(xy).
\end{equation}
We also define the unitary and self-adjoint transformation
$\Ju:L^2(\R_+)\to L^2(\R_+)$ by the formula
\begin{equation}\label{def_de_J}
\big(\Ju f\big)(x)=\frac{1}{x}f\Big(\frac{1}{x}\Big)
\end{equation}
for any $f\in L^2(\R_+)$ and $x\in \R_+$
Finally, we recall that $A$ denotes the generator of dilations.

\begin{proposition}\label{prop_sur_A}
For any $m\in \C$ with $\Re(m)>-1$
the map $\cF_{m}$ continuously extends to a bounded  invertible operator on $L^2(\R_+)$
satisfying $\cF_{m}^\t=\cF_{m} = \cF_m^{-1}$.
In addition, the following equalities hold
\begin{equation}\label{eq_explicit}
\cF_m=\Ju \Xi_m(A)=\Xi_m(-A)\Ju,\
\end{equation}
with
\begin{equation}\label{eq_def_Xi}
\Xi_m(t):= \e^{\i\ln(2)t} \frac{\Gamma(\frac{m+1+\i t}{2})}{\Gamma(\frac{m+1-\i t}{2})}\ .
\end{equation}
\end{proposition}

\begin{proof}
Let us start by proving the first equality in \eqref{eq_explicit}. For that purpose, consider the operator
$\Ju \cF_m: C_{\rm c}(\R_+)\to L^2(\R_+)$ whose kernel is given by
\begin{equation}\label{eq_ker_psi_A}
\sqrt{\frac2\pi}\frac{1}{x}\Ja_m\Big(\frac{y}{x}\Big)
=
\frac{1}{2\pi}\frac{1}{\sqrt{xy}}
\int_{-\infty}^{+\infty} \frac{\Gamma(\frac{m+\i
t+1}{2})}{\Gamma(\frac{m-\i t+1}{2})}
\left(\frac{1}{2}\right)^{-\i t}
\frac{y^{-\i t}}{x^{-\i t}}\;\!\d t,
\end{equation}
where \eqref{eq_Barnes} has been used for the second equality.
By taking into account the explicit formula for the kernel of an operator $\psi(A)$,
as recalled for example in \cite[Lem.~6.4]{BDG}, one infers that the r.h.s.~of
\eqref{eq_ker_psi_A} corresponds to the kernel
of an operator $\psi(A)$ with $\psi$ provided by the expression
\eqref{eq_def_Xi}. Then, from the density of $C_{\rm c}(\R_+)$ in $L^2(\R_+)$ and
since the map $t\mapsto \Xi_m(t)$ is bounded, one obtains that $\Ju \cF_m$ extends continuously
to the bounded operator $\Xi_m(A)$.
Since $\Ju$ is unitary and self-adjoint, one directly deduces the first equality in \eqref{eq_explicit}.

The second equality in \eqref{eq_explicit} is a straightforward consequence of the equality
\begin{equation*}
\Ju \;\! \e^{\i \tau A} \Ju = \e^{-\i \tau A}
\end{equation*}
which is easily checked. The equality $\cF_m^2 =\one$ can now be deduced from the equalities \eqref{eq_explicit}.
As a consequence $\cF_m^{-1}=\cF_m$, and this provides a direct proof of the boundedness of the inverse of $\cF_m$.
The equality $\cF_{m}^\t=\cF_{m}$ is finally a direct consequence of the expression \eqref{eq_kernel_Jm}
for the kernel of $\cF_m$.
\end{proof}

The map $\cF_m$ will be called the \emph{Hankel transformation of order $m$}.
It provides a kind of diagonalization of the operator $H_{m}$, as
shown in the next statement.

\begin{proposition}\label{priop1}
For any $m\in \C$ with $\Re(m)>-1$ and for any $k\in \C$ with $\Re(k)>0$ the following equality holds:
\begin{equation*}
\big(H_m+k^2\big)^{-1} = \cF_m \;\! (\Q^2+k^2)^{-1} \;\! \cF_m^\t\ .
\end{equation*}
\end{proposition}

\begin{proof}
The kernel of the operator on the r.h.s.~is given by the expression
\begin{equation*}
\big(\cF_m \;\! (\Q^2+k^2)^{-1} \;\! \cF_m^\t\big)(x,y)
= \frac{2}{\pi}\int_0^\infty \Ja_m(xp)\;\!\Ja_m(yp)
\frac{1}{(p^2+k^2)}\;\!\d p.
\end{equation*}
By \eqref{manipu}, it coincides with the
kernel of $R_m(-k^2)$ provided in \eqref{eq_kernel_H_m}.
\end{proof}

Proposition \ref{priop1} is convenient technically, because it
contains only bounded operators. One can rewrite it by using unbounded
operators as follows:

\begin{theoreme}
The operator $H_m$ is similar to a self-adjoint operator. More precisely, the following equalities hold:
\begin{equation*}
H_{m}=\cF_{m} \Q^2\cF_{m}^\t.
\qquad \hbox{and}\qquad
H_m=\Xi_m(-A)\Q^{-2}\Xi_m(A).
\end{equation*}
\end{theoreme}

\subsection{Spectral projections}\label{sec_proj_1}

For any $0<a<b$ let us now consider the operator
\begin{equation*}
\one_{[a,b]}(H_m)
:=\int_{\sqrt{a}}^{\sqrt{b}}p_m(k^2)\d (k^2) = 2\int_{\sqrt{a}}^{\sqrt{b}}p_m(k^2) k\;\!\d k,
\end{equation*}
which is defined as a bounded operator from  $\langle\Q\rangle^{-s}L^2(\R_+)$
to $\langle\Q\rangle^{s}L^2(\R_+)$ for
any $s>\frac{1}{2}$.
The kernel of this operator is given for $x,y\in \R_+$  by the expression
\begin{equation}\label{eq_kernel_m_x_y}
\one_{[a,b]}(H_m)(x,y)=\frac{2}{\pi}\int_{\sqrt{a}}^{\sqrt{b}}\Ja_m(kx)\;\!\Ja_m(ky)\;\!\d k.
\end{equation}

Clearly, by Stone's formula, for real $m>-1$, the above operator
extends to the self-adjoint operator $H_m$ onto the interval $[a,b]$.
For complex $m$, $H_m$  is not self-adjoint, hence strictly speaking
Stone's formula is not available. However, $H_m$ is similar to a
self-adjoint operator, hence the properties of
$\one_{[a,b]}(H_m)$ are almost the same as in the self-adjoint case.

\begin{proposition}\label{priop}
For any $0< a <b$ and any $m\in \C$ with $\Re(m)>-1$ one has
\begin{equation}\label{eq_diag_m}
\one_{[a,b]}(H_{m})=\cF_{m}\;\!\one_{[a,b]}(\Q^2)\;\!\cF_{m}^\t
\end{equation}
in $\B\big(L^2(\R_+)\big)$. In addition, $\one_{[a,b]}(H_{m})$ is a projection.
\end{proposition}

\begin{proof}
Let us first compute the r.h.s.~of \eqref{eq_diag_m}.
For that purpose recall that $\cF_m^{\t}=\cF_m$, and then one gets for any $f\in C_{\rm c}(\R_+)$
and $x\in \R_+$
\begin{align*}
\big(\cF_{m}\;\!\one_{[a,b]}(\Q^2)\;\!\cF_{m}^\t f\big)(x)
& = \sqrt{\frac{2}{\pi}}\int_{\sqrt{a}}^{\sqrt{b}} \Ja_m(kx) \;\!\big(\cF_m^\t f\big)(k)\;\!\d k \\
& = \frac{2}{\pi}\int_{\sqrt{a}}^{\sqrt{b}}\Ja(kx) \Big(\int_0^\infty\Ja(ky) f(y)\;\!\d y\Big)\d k \\
& = \int_0^\infty \Big(\frac{2}{\pi}\int_{\sqrt{a}}^{\sqrt{b}}\Ja_m(kx)\Ja_m(ky)\;\!\d k\Big) f(y)\;\!\d y,
\end{align*}
where Fubini's theorem has been applied for the last equality. By comparing the last expression
with \eqref{eq_kernel_m_x_y} one directly infers the equality \eqref{eq_diag_m}.
Note that since the r.h.s.~of \eqref{eq_diag_m} defines a bounded operator on $L^2(\R_+)$, this
equality provides a natural continuous extension of $\one_{[a,b]}(H_{m})$ as a bounded operator on $L^2(\R_+)$.
Finally, since $\cF_m$ satisfies $\cF_m^\t \cF_m = \one$, one directly infers that $\one_{[a,b]}(H_{m})$
is a projection.
\end{proof}

\subsection{M{\o}ller operators and scattering operator}

In this section we describe the scattering theory for the
operators $H_m$. Note that we also treat non-self-adjoint operators,
and therefore we cannot always invoke standard results.

Let us start by introducing the \emph{incoming} and \emph{outgoing Hankel transformations of order $m$}
defined by
\begin{equation}\label{def_cF_m_pm}
\cF_m^\mp:=\e^{\mp\i\frac{\pi}{2}m}\cF_m.
\end{equation}
Their kernel is provided by the expressions
\begin{equation*}
\cF_m^\mp (x,y)=\e^{\mp\i\frac{\pi}{2}m} \sqrt{\frac2\pi}\Ja_m(xy)
\end{equation*}
and the following relations trivially hold:
\begin{equation*}
\big(\cF_m^\mp\big)^{-1}=\cF_m^\pm= \cF_m^{\pm \t}.
\end{equation*}

Now, for any $m,m'\in \C$ with $\Re(m),\Re(m')>-1$, we define
\begin{equation}\label{eq_def_W}
W_{m,m'}^\pm :=  \cF_m^{\pm} \cF_{m'}^{\mp \t}.
\end{equation}
Some easy properties of these operators are gathered in the next statement.

\begin{proposition}
$W_{m,m'}^\pm$ are bounded invertible operators satisfying
\begin{align*}
W_{m,m'}^{\mp \t} \;\!W_{m,m'}^\pm & = \one, \\
W_{m,m'}^\pm \;\!W_{m,m'}^{\mp \t}&= \one, \\
W_{m,m'}^{\pm\t} & = W_{m',m}^{\mp},\\
W_{m,m'}^\pm H_{m'} & = H_{m}W_{m,m'}^\pm .
\end{align*}
\end{proposition}

Formally, the kernel of $W_{m,m'}^\pm$ is given by
\begin{equation*}
W_{m,m'}^\pm(x,y)=\e^{\pm\i\frac{\pi}{2}(m-m')}\frac{2}{\pi}
\int_0^\infty\Ja_m(kx)\Ja_{m'}(ky)\;\!\d k.
\end{equation*}
On the other hand, by using the expression derived in Section \ref{sec_Hankel}
one also gets
\begin{equation}\label{eq_wave_expli}
W_{m,m'}^\pm = \e^{\pm\i\frac{\pi}{2}(m-m')}
\frac{ \Gamma(\frac{m+1-\i A}{2})\Gamma(\frac{m'+1+\i A}{2})}{\Gamma(\frac{m+1+\i A}{2})\Gamma (\frac{m'+1-\i A}{2})}.
\end{equation}
Note also that the equality $(W_{m,m'}^\pm)^{-1}=W_{m,m'}^{\mp \t}$ holds.
Therefore, the scattering operator is defined by
\begin{equation*}
S_{m,m'}:=W_{m,m'}^{- \t} W_{m,m'}^-,
\end{equation*}
and for the operators considered above, one simply gets
\begin{equation*}
S_{m,m'}=\e^{-\i\pi(m-m')}\one.
\end{equation*}

Let us now make a link with the traditional approach of scattering theory.
By Proposition \ref{priop1} and the boundedness of the
Hankel transformations, the operators $H_m$ generate a bounded one-parameter
group by the formula:
\begin{equation*}
\e^{\i tH_m}=\cF_m\e^{\i t\Q^2}\cF_m^\t.
\end{equation*}
Therefore, we can try to apply time-dependent scattering
theory, even though the operators $H_m$ may be non-self-adjoint. (We
refer to \cite{Davies,Kato,Ry} for additional information on scattering
theory in the non-self-adjoint setting).
In the next statement we show that  $W^\pm_{m,m'}$ coincide with
the M{\o}ller operators for the pair $(H_m,H_{m'})$.

\begin{theoreme}\label{corol_existence}
For any $m,m'\in \C$ with $\Re(m),\Re(m')>-1$ one has
\begin{equation*}
\slim_{t\to \pm \infty}\e^{\i tH_m}\e^{-\i tH_{m'}}=W_{m,m'}^\pm.
\end{equation*}
\end{theoreme}

In the following proof, $C\big([-\infty,\infty]\big)$ denotes the set of
continuous functions on $\R$ having a limit at $+\infty$ and a limit
at $-\infty$.

\begin{proof}
Consider first the equality
\begin{equation*}
\e^{\i tH_m}\e^{-\i tH_{m'}}=
\cF_m\e^{\i t\Q^2}\cF_m^\t \cF_{m'}\e^{-\i t\Q^2}\cF_{m'}^\t
\end{equation*}
and observe that
\begin{equation*}
\cF_m^\t \cF_{m'}=\Xi_m(-A)\Xi_{m'}(A)
= \frac{ \Gamma(\frac{m+1-\i A}{2})\Gamma(\frac{m'+1+\i A}{2})}
{\Gamma(\frac{m+1+\i A}{2})\Gamma(\frac{m'+1-\i A}{2})}.
\end{equation*}
By considering then the asymptotic behavior of the $\Gamma$-function,
as presented for example in \cite[Eq.~6.1.37]{AS} or in Lemma \ref{stirling}, one infers that
the map
$$
t\mapsto \Xi_m(-t)\Xi_{m'}(t)
$$
belongs to $C\big([-\infty,\infty]\big)$ and that
\begin{equation}\label{eq-eq}
\Xi_m(\mp\infty)\Xi_{m'}(\pm\infty)
=\e^{\mp\i \frac{\pi}{2}(m-m')}.
\end{equation}
One infers then by Lemma \ref{lem_limit_A} that
\begin{equation*}
\slim_{t\to \pm \infty}\e^{\i t\Q^2}\cF_m^\t \cF_{m'}\e^{-\i
t\Q^2}=\e^{\pm\i \frac{\pi}{2}(m-m')},
\end{equation*}
from which one deduces that
\begin{align*}
\slim_{t\to \pm \infty}\e^{\i tH_m}\e^{-\i tH_{m'}}
&= \cF_m \e^{\pm\i \frac{\pi}{2}(m-m')}\cF_{m'}^\t \\
& =\cF_m^{\pm } \cF_{m'}^{\mp \t}  .
\qedhere
\end{align*}
\end{proof}

\subsection{Some special cases}

In some special situations the scattering theory for $H_m$ is
very explicit. The Gamma function is not even used in these special cases.
Some of them are provided in this section.

\begin{proposition}
The following identities hold:
\begin{align}
\nonumber \cF_{-m}=&\cF_m\frac{\cos\big(\frac\pi2(m+\i A)\big)}{\cos\big(\frac\pi2(m-\i A)\big)}
= \frac{\cos\big(\frac\pi2(m-\i A)\big)}{\cos\big(\frac\pi2(m+\i A)\big)}\cF_m, \\
\nonumber \cF_{-m}^\mp = &\cF_{m}^\mp\frac{\e^{\pm\pi A}+\e^{\pm\i \pi m}}{\e^{\pm\pi A}+\e^{\mp\i \pi m}}
= \frac{\e^{\mp\pi A}+\e^{\pm\i \pi m}}{\e^{\mp\pi A}+\e^{\mp\i \pi m}} \cF_m^\mp,\\
\label{wave-explicit} W_{-m,m}^\pm =&  \frac{\e^{\pm\pi A}+\e^{\mp\i \pi m}}{\e^{\pm\pi A}+\e^{\pm\i \pi m}} ,\\
\nonumber S_{-m,m} = &\e^{\i2\pi m}\one.
\end{align}
\label{prop_wave_op}\end{proposition}

\begin{proof}
By using the identity
\begin{equation*}
\Gamma\Big(z+\frac12\Big)
\Gamma\Big(-z+\frac12\Big)=\frac{\pi}{\cos(\pi z)},
\end{equation*}
one infers that
\begin{equation*}
\Xi_{-m}(t)=\Xi_m(t)\frac{\cos\big(\frac\pi2(m+\i t)\big)}{\cos\big(\frac\pi2(m-\i t)\big)}.
\end{equation*}
All the mentioned equalities can then be easily deduced.
\end{proof}

\begin{proposition}
The following identities hold:
\begin{align*}
\cF_{m+2} =&\cF_m\frac{m+1+\i A}{m+1-\i A}=\frac{m+1-\i A}{m+1+\i A}\cF_m,\\
\cF_{m+2}^\pm =&\cF_m^\pm\frac{m+1+\i A}{m+1-\i A}=\frac{m+1-\i A}{m+1+\i A}\cF_m^\pm,\\
W_{m+2,m}^\pm =& \frac{m+1-\i A}{m+1+\i A},\\
 S_{m+2,m}  =&\one.
\end{align*}
\end{proposition}

\begin{proof}
By using the identity $\Gamma(z+1) = z\Gamma(z)$
one infers that
\begin{equation*}
\Xi_{m+2}(t)=\Xi_m(t)\frac{m+1+\i t}{m+1-\i t},
\end{equation*}
from which all the equalities can be easily deduced.
\end{proof}

\subsection{Dirichlet and Neumann Laplacians on the half-line}

The simplest cases of operators $H_m$ are obtained for $m=\pm\frac12$.
They correspond to the Neumann and the Dirichlet boundary conditions.
We denote them with the usual notation, {\it i.e.}
\begin{equation*}
H_{\rm N}:=H_{-\frac12},\qquad H_{\rm D}:=H_{\frac12},
\end{equation*}
and recall some of their properties.

On $L^2(\R_+)$ we define the cosine and sine transformation:
\begin{align*}
\big(\cF_{\rm N} f\big)(x)&:=\sqrt{\frac2\pi}\int_0^\infty \cos(xy) f(y)\;\!\d y,\\
\big(\cF_{\rm D} f\big)(x)&:=\sqrt{\frac2\pi}\int_0^\infty \sin(xy) f(y)\;\!\d y.
\end{align*}
Note that the transformations $\cF_{\rm N}$ and $\cF_{\rm D}$ are involutive, real and unitary,
and that they correspond to the Hankel transforms, namely
\begin{equation*}
\cF_{-\frac12}=\cF_{\rm N},\qquad \cF_{\frac12}=\cF_{\rm D}.
\end{equation*}
In addition, it is well-known, and also a consequence of our previous computations, that
these transformations diagonalize $H_{\rm N}$ and $H_{\rm D}$:
\begin{equation*}
\cF_{\rm N}\;\! H_{\rm N}\;\! \cF_{\rm N}=\Q^2,\qquad \cF_{\rm D}\;\! H_{\rm D} \;\! \cF_{\rm D}=\Q^2.
\end{equation*}

Let us also recall the resolvents of these operators, their boundary values and their
spectral densities:
\begin{align*}
R_{\rm N}(-k^2;x,y)& =\frac{1}{2k}\big(\e^{-k|x-y|}+\e^{-k(x+y)}\big),  \\
R_{\rm N}(k^2\pm\i0;x,y) & =
\pm \frac{\i}{k}\left\{\begin{matrix}\! \cos(kx)\;\! \e^{\pm\i ky} & \hbox{ if } \ 0 < x \leq y \\
\cos(ky)\;\! \e^{\pm\i kx} & \hbox{ if } \ 0 < y < x \end{matrix}\right. ,\\
p_{\rm N}(k^2;x,y) & = \frac{1}{\pi k}\cos(kx)\cos(ky),
\end{align*}
and
\begin{align*}
R_{\rm D}(-k^2;x,y) & =\frac{1}{2k}\big(\e^{-k|x-y|}-\e^{-k(x+y)}\big),  \\
R_{\rm D}(k^2\pm\i0;x,y) & =
\frac{1}{k}\left\{\begin{matrix}\! \sin(kx)\;\! \e^{\pm\i ky}& \hbox{ if } \ 0 < x \leq y \\
\sin(ky)\;\! \e^{\pm\i kx} & \hbox{ if } \ 0 < y < x \end{matrix}\right. ,\\
p_{\rm D}(k^2;x,y) & = \frac{1}{\pi k}\sin(kx)\sin(ky).
\end{align*}

The incoming and outgoing Hankel transforms in this case differ from
the regular ones by a phase factor closely related to \emph{the Maslov correction}, famous in
the late 70's (see {\it e.g.}~the introduction to \cite{BPR}).
According to our definition given in \eqref{def_cF_m_pm} one has
\begin{equation}\label{wer7}
\cF_{\rm N}^\mp= \e^{\pm\i\frac{\pi}{4}}\cF_{\rm N},\qquad
\cF_{\rm D}^\mp= \e^{\mp\i\frac{\pi}{4}}\cF_{\rm D}.
\end{equation}

\begin{proposition}\label{wer4}
The wave operators for the pair $\big(H_{\rm N},H_{\rm D}\big)$
are given by
\begin{align}
\label{wer1} W_{\rm ND}^\pm = & \e^{\mp\i\frac{\pi}{2}}\cF_{\rm N}\cF_{\rm D} \\
\label{wer2} = &\pm \tanh(\pi A)\mp \i \cosh(\pi A)^{-1}.
\end{align}
Its kernel is
\begin{equation}\label{wer3}
W_{{\rm N}{\rm D}}^\pm(x,y)=\mp\frac{2\i}{\pi}\int_0^\infty\cos(ky)\sin(kx)\d        k.
\end{equation}
The corresponding scattering operator is simply given by
$S_{{\rm N},{\rm D}}=-\one$.
\end{proposition}

\begin{proof}
The equality \eqref{wer1} directly follows from Proposition \ref{prop_wave_op} together with
the explicit formula \eqref{wer7}.
The kernel \eqref{wer3} is then a direct consequence of \eqref{wer1}.

For \eqref{wer2}, recall that $W_{\rm ND}^\pm = W_{-\frac{1}{2},\frac{1}{2}}^\pm$,
and from \eqref{wave-explicit} one infers that
\begin{equation*}
W_{-\frac{1}{2},\frac{1}{2}}^\pm =
\pm\frac{\e^{\pi A}-\i}{\e^{\pi A}+\i} = \pm \tanh(\pi A)\mp \i \cosh(\pi A)^{-1},
\end{equation*}
where some identities involving the hyperbolic functions have been used.
\end{proof}

Note that with a different approach the expressions of Proposition
\ref{wer4} were already obtained in \cite[Sec.~2]{R16}.

\section{Point spectrum}\label{sec2}
\setcounter{equation}{0}

In this section we return to the study of the operators
$H_{m,\kappa}$ and $H_0^\nu$, and describe their point spectra.

\subsection{Eigenvalues}\label{sec21}

In the first statement, we fix ${z}\in\C$ and then look for
operators from our families which have the eigenvalue ${z}$.
Obviously, since most of the operators are not self-adjoint,
we consider arbitrary complex eigenvalues.
Recall also that the notations $z\mapsto \ln(z)$ and $z \mapsto z^m$  are used for the
principal branch of both functions, with domain $\cc\setminus]-\infty,0]$.
Recall that $\gamma$ denotes Euler's constant.

\begin{proposition}\label{prop_spec}
Let $m\in \C^\times$ with $|\Re(m)|<1$ and let $\kappa, \nu \in \C\cup\{\infty\}$. Then one has
\begin{enumerate}
\item[(i)] ${z}\in\sigma_\p(H_{m,\kappa})$ if and only if
${z}\in\C\setminus [0,\infty[$, and
\begin{equation}\label{tri}
\kappa=\frac{\Gamma(m)}{\Gamma(-m)}(-{z}/4)^{-m}.
\end{equation}
\item[(ii)] ${z}\in\sigma_\p(H_0^\nu)$ if and only if
${z}\in\C\setminus [0,\infty[$ and
\begin{equation*}
\nu=\gamma+\frac12\ln(-{z}/4).
\end{equation*}
\end{enumerate}
\end{proposition}

Before providing the proof of this proposition,
we deduce from it the main result about the point spectrum for our operators.
For that purpose let us also introduce for $m\neq 0$ the new parameter
\begin{equation}\label{short}
\varsigma(m,\kappa)=\varsigma:=
\kappa\frac{\Gamma(-m)}{\Gamma(m)}.
\end{equation}

\begin{theoreme}\label{thm_spec}
Let $m\in \C$ with $|\Re(m)|<1$.
\begin{enumerate}
\item[(i)] For $m\neq 0$ and $\kappa\in \C^\times$, one has
\begin{equation}\label{eq_align}
\sigma_{\p}(H_{m,\kappa})= \Big\{-4\e^{-w} \mid w\in
\frac{1}{m}\Ln(\varsigma) \hbox{ and } -\pi<\Im(w)<\pi\Big\}
\end{equation}
\item[(ii)] For any $\nu \in \C$, $\sigma_\p(H_0^\nu)$ is nonempty if and only if $-\frac\pi2<\Im(\nu)<\frac\pi2$, and then
\begin{equation*}
\sigma_\p(H_0^\nu) =\big\{-4\e^{2(\nu-\gamma)}\big\}.
\end{equation*}
\item[(iii)] $\sigma_\p(H_{m,0}) = \sigma_\p(H_{m,\infty})=\sigma_\p(H_0^\infty)=\emptyset\ $.
\end{enumerate}
\end{theoreme}

\begin{proof}
Only (i) needs a comment because multivalued functions can
be tricky. We can rewrite \eqref{tri} as
\begin{equation*}
(-{z}/4)^{-m}=\varsigma.
\end{equation*}
This is equivalent to
\begin{equation}\label{tri2}
\ln(-{z}/4)\in\frac{-1}{m}\Ln(\varsigma),
\end{equation}
where on the right of \eqref{tri2} we
have the set of values of the multivalued logarithm.
Finally, one deduces from the above inclusion that
\begin{equation*}
-{z}/4=\e^{-w},\ \ \ w\in
\frac{1}{m}\Ln(\varsigma),\ \
-\pi<\Im (w)<\pi,
\end{equation*}
which corresponds to \eqref{eq_align}.
\end{proof}

Let us stress that $\sigma_\p(H_{m,\kappa})$ depends in a complicated way on
the parameters $m$ and $\kappa$.
There exists a complicated pattern of \emph{phase transitions}, when
some eigenvalues ``disappear''. This happens if
\begin{equation}\label{excep}
\pi\in\Im\Big(\frac1m\Ln(\varsigma)\Big),\ \
\hbox{ or }\ \ \
-\pi\in\Im\Big(\frac1m\Ln(\varsigma)\Big).
\end{equation}
A pair $(m,\kappa)$ satisfying \eqref{excep} will be called \emph{exceptional}.

Similarly for the family of operators $H_0^\nu$, we shall say that $\nu$ is exceptional
if
\begin{equation}\label{excep_nu}
\Im(\nu)=-\frac\pi2,\ \hbox{ or }\ \Im(\nu)=\frac\pi2.
\end{equation}

Below we provide a characterization
of $\#\sigma_\p(H_{m,\kappa})$, {\it i.e.}~of the  number of eigenvalues of $H_{m,\kappa}$.

\begin{proposition}
Let $m= m_\r+\i m_\i \in \C^\times$ with $|m_\r|<1$.
\begin{enumerate}
\item[(i)] Let $m_\r=0$.
\begin{enumerate}
\item[(a)] If $\frac{\ln(|\kappa|)}{m_\i}\in ]-\pi,\pi[$.
then $\#\sigma_\p(H_{m,\kappa}) = \infty$,
\item[(b)] If $\frac{\ln(|\kappa|)}{m_\i}\not \in ]-\pi,\pi[$
then $\#\sigma_\p(H_{m,\kappa}) = 0$.
\end{enumerate}
\item[(ii)]  If $m_\r\neq 0$ and if $N\in \{0,1,2,\dots\}$ satisfies
$N<\frac{m_\r^2+m_\i^2}{|m_\r|} \leq N+1$, then
\begin{equation*}
\#\sigma_\p(H_{m,\kappa})\in \{N,N+1\}.
\end{equation*}
\end{enumerate}
\end{proposition}

\begin{proof}
The case (i)  can easily be deduced from \eqref{eq_align}.
For (ii), let  $m=m_\r+\i m_\i\in\cc^\times$ and
$\alpha=\alpha_\r+\i\alpha_\i \in\Ln(\varsigma)$.
Then one has
\begin{align}\label{wer8}
\Im\Big(\frac{1}{m}(\alpha+2\pi \i j)\Big) &
=\frac{\alpha_\i m_\r-\alpha_\r m_\i}{m_\r^2+m_\i^2}+\frac{2\pi
jm_\r}{m_\r^2+m_\i^2} .
\end{align}
If $m_\r>0$, the condition $-\pi<$\eqref{wer8}$<\pi$ is equivalent to
\begin{equation*}
0<j+\frac{\alpha_\i}{2\pi}-\frac{\alpha_\r m_\i}{2\pi
m_\r}+\frac{m_\r^2+m_\i^2}{2m_\r}<
\frac{m_\r^2+m_\i^2}{m_\r}.
\end{equation*}
Thus we can apply Lemma \ref{wer9} below with
$\beta:=\frac{m_\r^2+m_\i^2}{m_\r}$ and
$\gamma:=\frac{\alpha_\i}{2\pi}-\frac{\alpha_\r m_\i}{2\pi
m_\r}+\frac{m_\r^2+m_\i^2}{2m_\r}$ and infer the statement (ii).
If $m_\r<0$ one obtains
\begin{equation*}
0<j+\frac{\alpha_\i}{2\pi}-\frac{\alpha_\r m_\i}{2\pi
m_\r}-\frac{m_\r^2+m_\i^2}{2m_\r}<
-\frac{m_\r^2+m_\i^2}{m_\r}
\end{equation*}
and the same argument leads to the expected result.
\end{proof}

In the following lemma $[\beta]$ denotes the integral part of
$\beta\in\R$ and $\{\beta\}:=\beta-[\beta]$.

\begin{lemma} \label{wer9}
Let $\beta\geq0$ and $\gamma\in \R$.
If $\beta\in\Z$, then
\begin{equation*}
\#\{j\in\Z \mid 0<\gamma+j<\beta\}=\begin{cases}\beta-1& \hbox{ if } \  \gamma\in\Z\\
\beta & \hbox{ if } \ \gamma\not\in\Z,\end{cases}
\end{equation*}
while if $\beta\not\in\Z$, then
\begin{equation*}
\#\{j\in\Z\mid 0<\gamma+j<\beta\}=\begin{cases}[\beta]& \hbox{ if } \  \gamma\in\Z
\ \hbox{or}\ \{\beta\}\leq\{\gamma\}\\
[\beta]+1& \hbox{ if } \   0<\{\gamma\}<\{\beta\}.\end{cases}
\end{equation*}
\end{lemma}

\begin{proof}
For $\beta\in\Z$ the following $j$ are suitable:
\begin{align*}
j= & -\gamma+1,\dots,-\gamma+\beta-1, \quad \hbox{ if } \ \gamma\in\Z,\\
j= & -[\gamma],\dots,-[\gamma]+\beta-1, \quad \hbox{ if } \ \gamma\not\in\Z.
\end{align*}
For $\beta\not\in\Z$ the following $j$ are suitable:
\begin{align*}
j=& -\gamma+1,\dots,-\gamma+[\beta], \quad \hbox{ if } \ \gamma\in\Z,\\
j= & -[\gamma],\dots,-[\gamma]+[\beta]-1, \quad \hbox{ if }\  \{\beta\}\leq\{\gamma\},\\
j= & -[\gamma],\dots,-[\gamma]+[\beta], \quad \hbox{ if } \ 0< \{\gamma\}<\{\beta\}.
\qedhere
\end{align*}
\end{proof}

\subsection{Proof of Proposition \ref{prop_spec}}

In this section, we prove  Proposition \ref{prop_spec} about the
location of possible eigenvalues.
First of all, instead of looking for solution of the equation
$L_{m^2}f={z} f$, it will be convenient to write ${z}=-k^2$ with $k \in \C$ and $\Re(k)\geq 0$.
Then, recall that for ${k} \neq 0$ two linearly independent solutions of the differential equation (acting on distributions)
\begin{equation}\label{Eq_eigen}
L_{m^2} \;\!f = -k^2 f
\end{equation}
are provided by $x\mapsto \Ia_m({k} x)$ and $x\mapsto  \Ka_m({k} x)$.
On the other hand, in the special case $k =0$ the equation
$L_{m^2}\;\! f = 0$ corresponds to Euler's equation,
and its solutions consist in elementary functions.
 In fact, we shall treat separately the three cases $\Re(k)>0$, $\Re({k})=0$ but ${k} \neq 0$,
and the special case ${k} =0$.

\subsubsection{$\Re({k})>0$}

We first concentrate on the case ${k} \in \C$ with $\Re({k})> 0$. Since $|\arg(kx)|<\frac{\pi}{2}$
the function $x\mapsto \Ia_m({k} x)$
 exponentially increases for large $x$. One deduces that for any $m$ it cannot be in $L^2(\R_+)$.
On the other hand, by \eqref{tri4a} and \eqref{eq:besselfunc2_1}
the function $x\mapsto \Ka_m(kx)$ is square integrable.
As a consequence,  it remains to determine for which pairs
$(m,\kappa)$ it belongs to $\Dom(H_{m,\kappa})$,
or for which $\nu$  it belongs to $\Dom(H_0^\nu)$.

For $|\Re(m)|<1$ with $m\neq 0$, consider the equality
\eqref{macdo1} and the power expansion of $\Ia_m$ provided in \eqref{asym2}.
For $x\in ]0,1[$ one obtains
\begin{align*}
\Ka_m({k} x)
= & \frac{\sqrt\pi}{\sin(\pi m)}\frac{1}{\Gamma(1+m)}
\Big(\frac{\Gamma(1+m)}{\Gamma(1-m)}\big(\tfrac{kx}{2}\big)^{1/2-m}
- \big(\tfrac{kx}{2}\big)^{1/2+m}
\Big)  + f_{k}(x) \\
= &\frac{\Gamma(-m)}{\sqrt\pi} \big(\tfrac{k}{2}\big)^{1/2+m}\Big(  \frac{\Gamma(m)}{\Gamma(-m)}(\tfrac{k}{2})^{-2m}x^{1/2-m} + x^{1/2+m}
\Big)  + f_{k}(x),
\end{align*}
with $f_{k} \in \Dom(L_{m^2}^{\min})$ around $0$, by \cite[Prop.~4.12]{BDG}.
Thus, one infers from this computation
that the function $x\mapsto \Ka_m({k} x)$ belongs to $\Dom(H_{m,\kappa})$ if and only if
\begin{equation}\label{Eq_condition_1}
\kappa =\frac{\Gamma(m)}{\Gamma(-m)}({k}/2)^{-2m}.
\end{equation}
Equivalently, $-{k}^2$ is an eigenvalue of $H_{m,\kappa}$ if and only if the equality
\eqref{Eq_condition_1} holds.
Note that $H_{m,\infty}$ has no eigenvalue.

In the special case $m=0$, by \eqref{bess5} observe that
\begin{equation*}
\Ka_0({k} x) = -\frac{\sqrt{2kx}}{\sqrt{\pi}}\Big(\ln(x) +
\ln\big(\tfrac{k}{2}\big)+\gamma
\Big) + f_{k}(x)
\end{equation*}
with $f_{k} \in \Dom(L_0^{\min})$ around $0$, by \cite[Prop.~4.12]{BDG}.
One infers from this computation
that the function $x\mapsto\Ka_0({k} x)$ does not belong to $\Dom(H_{0,\kappa})$, for any $\kappa \in \C\cup\{\infty\}$.
On the other hand, this function belongs to $\Dom(H_0^\nu)$
if and only if
\begin{equation}\label{Eq_condition_2}
\nu = \gamma +\ln({k}/2).
\end{equation}
Equivalently, $-{k}^2$ is never an eigenvalue of $H_{0,\kappa}$, but $-k^2$  is an eigenvalue of
$H_0^\nu$ if and only if the equality \eqref{Eq_condition_2} holds.

\subsubsection{${k} \in \i\R^\times$}

Let us set ${k} = \i\mu$ with $\mu \in \R^\times$.
Our aim is to show that $-{k}^2>0$ can never be an eigenvalue of $H_{m,\kappa}$ or of $H_0^\nu$.
For that purpose, consider the two linearly
independent solutions of $L_{m^2}f = \mu^2 f$ provided by $x\mapsto
\Ha_m^\pm(\mu x)$.
By the asymptotics \eqref{trib} it appears that no non-trivial linear combination of these functions is
square integrable at infinity.
As expected, we have thus shown that if $k\in \i\R^\times$, no solution of
the equation \eqref{Eq_eigen} is in $L^2(\R_+)$.

\subsubsection{${k}=0$}

When ${k} =0$, the problem consists first in finding solutions to the homogeneous equation $L_{m^2} f=0$
with $f\in \Dom(L_{m^2}^{\max})$. Basic solutions for this equations
for $m\neq0$ are the functions $f_\pm$ with
$f_\pm(x)=x^{\pm m+1/2}$.
In the special case $m=0$, a second solution for this equation
is provided by the function $f_0$ with $f_0(x)=x^{1/2}\ln(x)$. However, none of these functions belongs to $L^2(\R_+)$,
which means that $0$ is never an eigenvalue for the operators
$H_{m,\kappa}$ or $H_0^\nu$.

\begin{proof}[Proof of Proposition \ref{prop_spec}]
It has been shown above that $-k^2$ is an eigenvalue of some of the operators $H_{m,\kappa}$ or $H_0^\nu$ if
$\Re(k)>0$ and if \eqref{Eq_condition_1} or \eqref{Eq_condition_2} hold.
The first statement of the proposition corresponds to reformulation of \eqref{Eq_condition_1} with ${z} = -k^2$,
while the second statement corresponds to \eqref{Eq_condition_2} also with ${z} = -k^2$.
\end{proof}

\subsection{The self-adjoint case}\label{sec_self_adj}

Let us summarize the content of the first part of this section for self-adjoint operators $H_{m,\kappa}$ or $H_0^\nu$.
The following statement is a reformulation of Corollary \ref{corol_SA} and of Theorem \ref{thm_spec}.

\begin{theoreme}
\begin{enumerate}
\item[(i)] If $m\in ]-1,1[\setminus\{0\}$, then $H_{m,\kappa}$ is
self-adjoint if and only if $\kappa\in\rr\cup\{\infty\}$, and then
\begin{align*}
\sigma_\p(H_{m,\kappa})& =  \Big\{-4\Big(\kappa
\frac{\Gamma(-m)}{\Gamma(m)}\Big)^{ -1/m}\Big\}
\quad \hbox{ for } \kappa\in]-\infty,0[,\\
\sigma_\p(H_{m,\kappa})& =\emptyset\quad \hbox{ for } \kappa\in[0,\infty].
\end{align*}
\item[(ii)] If $m=\i m_\i\in\i\R\setminus\{0\}$, then $H_{\i m_\i,\kappa}$ is
self-adjoint if and only if $|\kappa|=1$, and then
\begin{equation*}
\sigma_\p(H_{\i m_\i,\kappa}) = \bigg\{-4 \exp\bigg(-\frac{\arg\big(\kappa
\frac{\Gamma(-\i m_\i)}{\Gamma(\i m_\i)}\big)+2\pi j}{m_\i}\bigg)\mid j\in \Z\bigg\}.
\end{equation*}
\item[(iii)]  $H_0^\nu$ is self-adjoint if and only if
$\nu\in\rr\cup\{\infty\}$, and then
\begin{align*}
\sigma_\p(H_0^\nu)& = \big\{-4\e^{2(\nu-\gamma)}\big\} \quad \hbox{ for } \nu\in\R,\\
\sigma_\p(H_0^\infty)& = \emptyset.
\end{align*}
\end{enumerate}
\end{theoreme}

\begin{remark}
Let us emphasize that none of the pairs $(m,\kappa)$ corresponding to a self-adjoint operator $H_{m,\kappa}$
is an exceptional pair.
Similarly, the parameter $\nu$ corresponding to a self-adjoint operator $H_0^\nu$ is never exceptional.
\end{remark}

\section{Continuous spectrum of $H_{m,\kappa}$}
\setcounter{equation}{0}

In this section we extend the results obtained in Section \ref{sec_homogeneous}
to the  families of operators $H_{m,\kappa}$.

\subsection{Resolvent}

In this section we consider $m\in \C^\times$ with $|\Re(m)|<1$, $\kappa\in \C\cup \{\infty\}$,
and $x\in\R_+$.
Let us also fix $k\in \C$ with $-k^2\not\in\sigma(H_{m,\kappa})$, and as before we impose $\Re(k)> 0$.
Note that we have directly imposed that $k\not \in \i \R$ since later
on we shall show that $[0,\infty[\;\!\subset \sigma(H_{m,\kappa})$.
Our aim is to compute the integral kernel of the resolvent of $H_{m,\kappa}$
\begin{equation*}
R_{m,\kappa}(-k^2):= (H_{m,\kappa}+k^2)^{-1}.
\end{equation*}

First of all, recall from Section \ref{sec2} that the map $x\mapsto \Ka_m(kx)$ is a solution of the equation
\begin{equation}\label{eq_kernel}
(L_{m^2}+k^2)f=0
\end{equation}
and belongs to $L^2(\R_+)$.
Similarly, both functions $x\mapsto \Ia_{m}(kx)$ and
$x\mapsto \Ia_{-m}(kx)$ satisfy the equation \eqref{eq_kernel},
but only a certain linear combination belongs to $\Dom(H_{m,\kappa})$
around $0$.

Recall now the parameter that we have introduced in \eqref{short}, namely
\begin{equation*}
\varsigma=\kappa\frac{\Gamma(-m)}{\Gamma(m)}.
\end{equation*}
For $(\frac{k}{2})^{2m}\varsigma\neq1$,
by taking \eqref{eq_Im_asym} into account,
one infers that the map
\begin{equation}\label{def_v_m_kappa}
x\mapsto \frac{\Ia_{m}(kx)-\varsigma (\frac{k}{2})^{2m} \Ia_{-m}(kx)}
{1-\varsigma (\frac{k}{2})^{2m}},
\end{equation}
satisfies \eqref{eq_do1} or \eqref{eq_do2} around $0$,
and hence belongs to $\Dom(H_{m,\kappa})$ around $0$.
Obviously, it also solves \eqref{eq_kernel}.
Furthermore, by using the formulas of Section \ref{subsec_Wronsk},
the Wronskian of \eqref{def_v_m_kappa} and $x\mapsto\Ka_m(kx)$ equals $k$

From the general theory of Sturm-Liouville operators,
as recalled for example in \cite[Prop.~A.1]{BDG}, the kernel of the
operator $(H_{m,\kappa}+k^2)^{-1}$
is provided for  $\varsigma (\frac{k}{2})^{2m}\neq 1$ by the expression
\begin{eqnarray*}
&&R_{m,\kappa}(-k^2;x,y)\\&=& \frac1{k\big(1-\varsigma (\frac{k}{2})^{2m}\big)}
\left\{\begin{matrix}
\big(\Ia_{m}(kx)-\varsigma (\frac{k}{2})^{2m} \Ia_{-m}(kx)\big)\;\!\Ka_m(ky) & \hbox{ if } 0 < x < y, \\[2ex]
\big(\Ia_{m}(ky)-\varsigma (\frac{k}{2})^{2m} \Ia_{-m}(ky)\big)\;\! \Ka_m(kx) & \hbox{ if } 0 < y < x.
\end{matrix}\right.
\end{eqnarray*}

Let us still provide a relation between $R_{m,\kappa}(-k^2)$ and
$R_{m}(-k^2)$. Indeed, by the definition \eqref{def_v_m_kappa}, and by taking
the equalities $H_{m,0}=H_{m}$ and $H_{m,\infty}=H_{-m}$ into account, one gets
\begin{align}
\label{eq_dif_res_1} R_{m,\kappa}(-k^2)
& = \frac{1}{1-\varsigma (\tfrac{k}{2})^{2m}}
\Big(R_m(-k^2)-\varsigma (\tfrac{k}{2})^{2m} R_{-m}(-k^2)\Big) \\
\label{eq_dif_res_1bis} & = R_m(-k^2) - \frac{ \varsigma(\tfrac{k}{2})^{2m}}
{1- \varsigma(\tfrac{k}{2})^{2m}}  \tfrac{m}{k^2} \;\!P_m(-k^2),
\end{align}
where $P_m(-k^2)$ is the projection defined in \eqref{proj1}, see also \eqref{eq_dif_res_m}.

Let us finally observe that for fixed $\kappa$ the condition $\varsigma(\tfrac{k}{2})^{2m}=1$
defines a discrete set which corresponds to the eigenvalues of $H_{m,\kappa}$ by Proposition \ref{prop_spec}.
On the other hand, since $R_{m,\kappa}(-k^2)$ is a rank one perturbation of $R_{m}(-k^2)$, one infers
that $[0,\infty[$ belongs to the spectrum of $H_{m,\kappa}$, as already mentioned at the beginning of this section.

\subsection{Boundary value of the resolvent and spectral density}\label{sec_spec_2}

Since $[0,\infty[$ belongs to the spectrum of $H_{m,\kappa}$,
it is natural to mimic the computations performed in Section \ref{sec_spec_1} for $H_m$.
Recall that an exceptional pair $(m,\kappa)$ has been defined in \eqref{excep}.

\begin{theoreme}
Let $m\in \C^\times$ with $|\Re(m)|<1$, let $\kappa\in \C\cup
\{\infty\}$, and let $k>0$.
\begin{enumerate}
\item[(i)] If $(m,\kappa)$ is not an exceptional pair, then the boundary values of the resolvent
\begin{equation}\label{eq_lim_abs_2}
R_{m,\kappa}(k^2\pm \i 0):=\lim_{\epsilon \searrow 0}  R_{m,\kappa}(k^2\pm \i \epsilon)
\end{equation}
exist in the sense of operators from  $\langle\Q\rangle^{-s}L^2(\R_+)$
to $\langle\Q\rangle^{s}L^2(\R_+)$ for any $s>\frac{1}{2}$,
uniformly in $k$ on each compact subset of $\R_+$.
The kernel of $R_{m,\kappa}(k^2\pm \i 0)$ is given for $0<x\leq y$ by
\begin{align*}
& R_{m,\kappa}(k^2\pm\i0;x,y)\\
&= \frac{\pm \i}{k\big(1-\varsigma \e^{\mp \i \pi m}\big(\tfrac{k}{2}\big)^{2m}\big)}
\Big(\Ja_{m}(kx)- \varsigma\big(\tfrac{k}{2}\big)^{2m}\Ja_{-m}(kx)\Big)\Ha_m^\pm(ky)
\end{align*}
and the same expression with the role of $x$ and $y$ exchanged for $0<y<x$.
\item[(ii)] If $(m,\kappa)$ is an exceptional pair, the same statement holds
for $k$ uniformly  on each compact subset of
\begin{equation}\label{exco}
\Big\{k\in\rr_+ \mid \big(\tfrac{k}{2}\big)^{2m}\varsigma\e^{\mp\i \pi
m}\neq1\Big\}.
\end{equation}
\end{enumerate}
\end{theoreme}

\begin{proof}
By taking the equality \eqref{eq_dif_res_1} into account one infers that
the limiting absorption principle \eqref{eq_lim_abs_2} can be deduced from Theorem \ref{thm_boundary} for $R_m$
and for $R_{-m}$.
The local uniformity in $k$ is also consequence of these estimates and of the explicit
formula for the pre-factors appearing in \eqref{eq_dif_res_1},
as long as the first factor has no singularity.

For the kernel one directly gets for $0<x\leq y$ that
\begin{align*}
& R_{m,\kappa}(k^2\pm \i 0;x,y) \\
&=\frac{1}{1-\varsigma \e^{\mp \i \pi m}\big(\tfrac{k}{2}\big)^{2m}}R_{m}(k^2\pm \i 0;x,y)\\
& \quad -\frac{\varsigma \e^{\mp \i \pi m}\big(\tfrac{k}{2}\big)^{2m}}
{1-\varsigma \e^{\mp \i \pi m}\big(\tfrac{k}{2}\big)^{2m}}R_{-m}(k^2\pm \i 0;x,y) \\
& = \frac{\pm \i}{k\Big(1-\varsigma \e^{\mp \i \pi m}\big(\tfrac{k}{2}\big)^{2m}\Big)}
\Big(\Ja_{m}(kx)- \varsigma\big(\tfrac{k}{2}\big)^{2m} \Ja_{-m}(kx)\Big)\Ha_m^\pm(ky),
\end{align*}
where \eqref{eq_H_pm_m} and \eqref{eq_boundary} have been taken into account.
Note that for $0<y<x$ one gets the same expression with the role of $x$ and $y$ exchanged.
\end{proof}

Based on the previous result, the corresponding \emph{spectral density} can now be defined for any $m\in \C^\times$ with $|\Re(m)|<1$.
More precisely, for any $k>0$
if $(m,\kappa)$ is not an exceptional pair, or for $k$ belonging to \eqref{exco} if $(m,\kappa)$ is an exceptional pair,
let us set
\begin{equation*}
p_{m,\kappa}(k^2):= \frac{1}{2\pi \i}\Big(
R_{m,\kappa}\big(k^2+\i0\big) -R_{m,\kappa}\big(k^2-\i0\big)\Big),
\end{equation*}
which is bounded from  $\langle\Q\rangle^{-s}L^2(\R_+)$
to $\langle\Q\rangle^{s}L^2(\R_+)$ for any $s>\frac{1}{2}$.

\begin{proposition}
The kernel of the spectral density is given by the following formula:
\begin{align*}
& p_{m,\kappa}(k^2;x,y) \\
&=\frac{\Big(\Ja_{m}(kx)- \varsigma\big(\tfrac{k}{2}\big)^{2m}\Ja_{-m}(kx)\Big)
\Big(\Ja_{m}(ky)- \varsigma\big(\tfrac{k}{2}\big)^{2m}\Ja_{-m}(ky)\Big)}{\pi k
\Big( \sin^2(\pi m) + \big(\cos(\pi m)-\varsigma\big(\tfrac{k}{2}\big)^{2m}\big)^2\Big)}.
\end{align*}
\end{proposition}

\begin{proof}
For the kernel of this operator, observe that for $0<x\leq y$ one has
\begin{align*}
& 2\pi k \;\!p_{m,\kappa}(k^2;x,y) \\
& = \frac{1}{1-\varsigma \e^{-\i \pi m}\big(\tfrac{k}{2}\big)^{2m}}
\Big(\Ja_{m}(kx)- \varsigma\big(\tfrac{k}{2}\big)^{2m} \Ja_{-m}(kx)\Big)\Ha_m^+(ky) \\
& \quad + \frac{1}{1-\varsigma \e^{\i \pi m}\big(\tfrac{k}{2}\big)^{2m}}
\Big(\Ja_{m}(kx)- \varsigma\big(\tfrac{k}{2}\big)^{2m}\Ja_{-m}(kx)\Big)\Ha_m^-(ky) \\
& = \frac{\Ja_{m}(kx)- \varsigma\big(\tfrac{k}{2}\big)^{2m} \Ja_{-m}(kx)}{
\big(1-\varsigma \e^{-\i \pi m}\big(\tfrac{k}{2}\big)^{2m}\big)
\big(1-\varsigma \e^{\i \pi m}\big(\tfrac{k}{2}\big)^{2m}\big)} \\
&\quad \times\left\{\Big(1-\varsigma \e^{\i \pi m}\big(\tfrac{k}{2}\big)^{2m}\Big)\Ha_m^+(ky) +
\Big(1-\varsigma \e^{-\i \pi m}\big(\tfrac{k}{2}\big)^{2m}\Big)\Ha_m^-(ky)\right\}  \\
& = \frac{\Big(\Ja_{m}(kx)- \varsigma\big(\tfrac{k}{2}\big)^{2m}\Ja_{-m}(kx)\Big)
\Big(\Ja_{m}(ky)- \varsigma \big(\tfrac{k}{2}\big)^{2m} \Ja_{-m}(ky)\Big)}{\frac{1}{2}
\Big(1-2\varsigma \cos(\pi m) \big(\tfrac{k}{2}\big)^{2m} +
\varsigma^2\big(\tfrac{k}{2}\big)^{4m}\Big)}.
\end{align*}
Then, by using a simple trigonometric equality for the denominator, and since the role of $x$ and $y$ can be exchanged,
one directly obtains the desired expression for any $x,y\in \R_+$.
\end{proof}

\subsection{Generalized Hankel transform}

We would like to generalize the definition of the Hankel transformations
of Section \ref{sec_Hankel} and
to show their relations with the operators $H_{m,\kappa}$
The main idea is to factorize the spectral density,
but let us stress that this factorization is not unique.

One possibility, which works well for $m$ and $\kappa$ real, is as follows.
For any $m\in \R^\times$ with $|m|<1$ and any $\kappa \in \R\cup \{\infty\}$ one could set
\begin{equation*}
\cF_{m,\kappa}:C_{\rm c}(\R_+)\to L^2(\R_+)
\end{equation*}
with
\begin{equation*}
\big(\cF_{m,\kappa} f\big)(x):=\int_0^\infty\cF_{m,\kappa} (x,y)f(y)\d y,
\end{equation*}
and
\begin{equation*}
\cF_{m,\kappa} (x,y) := \sqrt{\frac2\pi}
\frac{\Ja_{m}(xy)- \varsigma \Ja_{-m}(xy)\big(\tfrac{y}{2}\big)^{2m}}{
\sqrt{\sin^2(\pi m) + \big(\cos(\pi m)-\varsigma
\big(\tfrac{y}{2}\big)^{2m}\big)^2}}.
\end{equation*}
For real $m,\kappa$, the denominator is the square root of a positive
number, which has a clear meaning. For  general $m,\kappa$ we need to
choose the branch of the square root, which is ambiguous.

Another option, which we will prefer since it has always a unique
definition, is to define the \emph{incoming} and
\emph{outgoing Hankel transformations} $\cF_{m,\kappa}^\mp$ given by
the kernels
\begin{equation*}
\cF_{m,\kappa}^\mp(x,y) :=\e^{\mp \i \frac{\pi}{2}m}\sqrt{\frac2\pi}
\frac{\Ja_{m}(xy)- \varsigma \Ja_{-m}(xy)\big(\tfrac{y}{2}\big)^{2m}}{
1-\varsigma \e^{\mp\i \pi m}\big(\tfrac{y}{2}\big)^{2m}}.
\end{equation*}
Note that the denominator of this kernel has a singularity if $(m,\kappa)$ is an exceptional pair.
For simplicity, we shall not consider this situation anymore in the sequel.
Thus, if $(m,\kappa)$ is not an exceptional pair, the following equalities can easily be obtained:
\begin{equation}\label{eq_m_kappa_pm}
\cF_{m,\kappa}^\mp =
\Big(\cF_m - \varsigma  \cF_{-m}\big(\tfrac{\Q}{2}\big)^{2m}\Big)
\frac{\e^{\mp \i \frac{\pi}{2}m}}{1-\varsigma\e^{\mp\i \pi m}\big(\tfrac{\Q}{2}\big)^{2m}}.
\end{equation}

Let us now show that these transformations provide a kind of diagonalization
of the operator $H_{m,\kappa}$. The statements and the proofs are divided into several shorter statements.

\begin{lemma}
For any $m\in \C^\times$ with $|\Re(m)|<1$ and any $\kappa\in \C\cup
\{\infty\}$ with $(m,\kappa)$ not an exceptional pair, the operators $\cF_{m,\kappa}^\mp $
extend continuously to the following operators:
\begin{equation}\label{eq_a_gauche}
\begin{split}
& \Big(\Xi_m(-A)-\varsigma\Xi_{-m}(-A) (2\Q)^{-2m} \Big)
\frac{\e^{\mp \i \frac{\pi}{2}m}}{1-\varsigma \e^{\mp \i \pi m} (2\Q)^{-2m}}\Ju \\
&=  \Ju \Big(\Xi_m(A)-\varsigma\Xi_{-m}(A) \big(\tfrac{\Q}{2}\big)^{2m}\Big)
\frac{\e^{\mp \i \frac{\pi}{2}m}}{1-\varsigma \e^{\mp \i \pi m} \big(\tfrac{\Q}{2}\big)^{2m}}
\end{split}
\end{equation}
where $\Xi_m$ and $\Xi_{-m}$ have been defined in \eqref{eq_def_Xi}.
\end{lemma}

\begin{proof}
The proof simply consists in using the equalities obtained in Proposition \ref{prop_sur_A}
and in a short computation based on the definition of $\Ju$, see equation \eqref{def_de_J}.
\end{proof}

Let us still provide the expressions for the adjoint of these operators, namely $\cF_{m,\kappa}^{\mp \t}$
are equal to the operators
\begin{equation}\label{eq_a_droite}
\begin{split}
& \Ju \frac{\e^{\mp \i \frac{\pi}{2}m}}{1-\varsigma \e^{\mp \i \pi m} (2\Q)^{-2m}}
\Big(\Xi_m(A)-\varsigma  (2\Q)^{-2m}\Xi_{-m}(A)\Big) \\
& =  \frac{\e^{\mp \i \frac{\pi}{2}m}}{1-\varsigma \e^{\mp \i \pi m}  \big(\tfrac{\Q}{2}\big)^{2m}}
\Big(\Xi_m(-A)-\varsigma   \big(\tfrac{\Q}{2}\big)^{2m}\Xi_{-m}(-A)\Big) \Ju.
\end{split}
\end{equation}
Additional information on the operators $\cF_{m,\kappa}^{\pm\t}$ are provided in the next statement.

\begin{lemma}\label{lem_proj_m_kappa}
For any $m\in \C^\times$ with $|\Re(m)|<1$ and  $\kappa\in \C\cup
\{\infty\}$ with $(m,\kappa)$ not an exceptional pair, the following equalities hold:
\begin{equation*}
\cF_{m,\kappa}^{\pm \t} \cF_{m,\kappa}^{\mp}=\one.
\end{equation*}
\end{lemma}

\begin{proof}
The proof consists simply in considering the terms \eqref{eq_a_gauche} and \eqref{eq_a_droite} and in checking that
their product (for the correct sign) is equal to $\one$. Indeed, observe first that on $C_{\rm c}(\R_+)$ the following equalities hold:
\begin{align*}
& \Big(\Xi_m(A)-\varsigma (2\Q)^{-2m}\Xi_{-m}(A)  \Big)
\Big(\Xi_m(-A)-\varsigma\Xi_{-m}(-A) (2\Q)^{-2m} \Big) \\
& =1 + \varsigma^2 (2\Q)^{-4m} - \varsigma \Big(  (2\Q)^{-2m}\Xi_{-m}(A) \Xi_m(-A) + \Xi_m(A)\Xi_{-m}(-A) (2\Q)^{-2m} \Big).
\end{align*}
If one then shows that
\begin{equation}\label{eq_to_prove}
(2\Q)^{-2m}\Xi_{-m}(A) \Xi_m(-A) + \Xi_m(A)\Xi_{-m}(-A) (2\Q)^{-2m} = 2\cos(\pi m) (2\Q)^{-2m}
\end{equation}
the statement of the lemma directly follow by using the full expressions provided in
\eqref{eq_a_gauche} and \eqref{eq_a_droite}.

For the proof of \eqref{eq_to_prove} that recall $\{U_\tau\}_{\tau\in \R}$ corresponds to the dilation
group, namely $U_\tau=\e^{\i\tau A}$ as introduced in Section \ref{sec_notation}.
Then for any $f\in C_{\rm c}(\R_+)$ and $x\in \R_+$ one has
\begin{align*}
\big((2\Q)^{-2m} U_\tau (2\Q)^{2m}f\big)(x) & = \e^{-2m\ln(x)} \big(U_\tau \Q^{2m}f\big)(x) \\
& = \e^{-2m\ln(x)} \e^{\tau/2} \big( \Q^{2m}f\big)(\e^\tau x) \\
& = \e^{-2m\ln(x)} \e^{\tau/2} \e^{2m\ln(e^\tau x)} f(\e^\tau x) \\
& = \e^{2m\tau}\big(\e^{\i\tau A}f\big)(x) \\
& = \big(\e^{\i \tau (A-\i 2m)}f\big)(x).
\end{align*}
One infers from this computation that the l.h.s.~of \eqref{eq_to_prove} is equal to
\begin{equation*}
\Big(\Xi_{-m}(A-\i 2 m) \Xi_m(-A+\i 2 m) + \Xi_m(A)\Xi_{-m}(-A)\Big) (2\Q)^{-2m}.
\end{equation*}
Finally, by taking into account the explicit formula \eqref{eq_def_Xi} for $\Xi_m$
and the equality $\Gamma(z)\Gamma(1-z)=\frac{\pi}{\sin(\pi z)}$ one gets (with $t$ instead of $A$)
\begin{align*}
& \Xi_{-m}(t-\i 2 m) \Xi_m(-t+\i 2 m) + \Xi_m(t)\Xi_{-m}(-t) \\
& = \Gamma\big(\tfrac{m+1+\i t}{2}\big) \Gamma\big(\tfrac{-m+1-\i t}{2}\big)
\Big(\tfrac{1}{\Gamma\big(\frac{-3m+1-\i t}{2}\big)\Gamma\big(\frac{3m+1+\i t}{2}\big)}
+\tfrac{1}{\Gamma\big(\frac{m+1-\i t}{2}\big)\Gamma\big(\frac{-m+1+\i t}{2}\big)}
\Big) \\
& = \frac{1}{\sin(\alpha)}\big(\sin(\alpha-2\beta)+\sin(\alpha+2\beta)\big)
\end{align*}
with $\alpha := \pi\big(\tfrac{-m+1-\i t}{2}\big)$ and $\beta:=\pi\tfrac{m}{2}$.
Some trigonometric identities lead then directly to the equality
\begin{equation*}
 \frac{1}{\sin(\alpha)}\big(\sin(\alpha-2\beta)+\sin(\alpha+2\beta)\big) = 2\cos(2\beta) = 2\cos(\pi m),
\end{equation*}
as expected.
\end{proof}

Let us now set
\begin{equation}\label{def_1}
\one_{\R_+}(H_{m,\kappa}) := \cF_{m,\kappa}^{\pm}\;\!\cF_{m,\kappa}^{\mp \t}
\end{equation}
and observe that this is again a projection.
In the self-adjoint case this operator corresponds to the usual projection onto the continuous spectrum of $H_{m,\kappa}$.
One can then prove the analogue of Proposition \ref{priop1}.

\begin{proposition}\label{prop_inter}
Let $m\in \C^\times$ with $|\Re(m)|<1$ and $\kappa\in \C\cup
\{\infty\}$ with $(m,\kappa)$ not an exceptional pair. Then for
any $k\in \C$ with $\Re(k)>0$ and $-k^2\not \in \sigma_{\rm p}(H_{m,\kappa})$ the following equalities hold:
\begin{equation*}
R_{m,\kappa}(-k^2) \one_{\R_+}(H_{m,\kappa})
=\cF_{m,\kappa}^{\pm}(\Q^2+k^2)^{-1} \cF_{m,\kappa}^{\mp \t}
=\one_{\R_+}(H_{m,\kappa})R_{m,\kappa}(-k^2).
\end{equation*}
\end{proposition}

\begin{proof}
We will use the following convenient expression for
$\cF_{m,\kappa}^{\mp\t}$ and two
formulas for the resolvent:
\begin{align}
\label{reso0}  \cF_{m,\kappa}^{\mp \t} &= \frac{\e^{\mp \i \frac{\pi}{2}m}}{1-\varsigma
\e^{\mp\i \pi m}\big(\tfrac{\Q}{2}\big)^{2m}} \Big(\cF_m - \varsigma \big(\tfrac{\Q}{2}\big)^{2m} \cF_{-m}\Big),\\
\label{reso1}  R_{m,\kappa}(-k^2)&= R_m(-k^2)-\frac{\varsigma(\frac{k}{2})^{2m}}
{1-\varsigma(\frac{k}{2})^{2m}}\frac{m}{k^2}P_m(-k^2),\\
\label{reso2}  R_{m,\kappa}(-k^2)&= R_{-m}(-k^2)-\frac{1}
{1-\varsigma(\frac{k}{2})^{2m}}\frac{m}{k^2}P_m(-k^2).
\end{align}
By multiplying \eqref{reso1} and
\eqref{reso2} from the left by $\cF_m$ and $\cF_{-m}$ respectively we obtain
\begin{align*}
\cF_mR_{m,\kappa}(-k^2) =& (X^2+k^2)^{-1}\cF_m-\frac{\varsigma(\frac{k}{2})^{2m}}
{1-\varsigma(\frac{k}{2})^{2m}}\frac{m}{k^2}\cF_mP_m(-k^2),\\
\cF_{-m}R_{m,\kappa}(-k^2) =& (X^2+k^2)^{-1}\cF_{-m}-\frac{1}
{1-\varsigma(\frac{k}{2})^{2m}}\frac{m}{k^2}\cF_{-m}P_m(-k^2).
\end{align*}
By combining then the above two identities, and by taking \eqref{reso0} into account,
we get
\begin{align*}
&\cF_{m,\kappa}^{\mp \t} R_{m,\kappa}(-k^2)\ = \ (\Q^2+k^2)^{-1}\cF_{m,\kappa}^{\mp \t}\\
&-\frac{\e^{\mp \i \frac{\pi}{2}m}}{\Big(1-\varsigma
\e^{\mp\i \pi m}\big(\tfrac{\Q}{2}\big)^{2m}\Big)\Big(1-
\varsigma(\frac{k}{2})^{2m}\Big)}\frac{m\varsigma}{k^2}\Big(\big(\tfrac{k}{2}\big)^{2m}\cF_m -
\big(\tfrac{\Q}{2}\big)^{2m}\cF_{-m}\Big)P_m(-k^2).
\end{align*}

By taking the equality \eqref{beqo} into account, one can deduce
that for any $z>0$ and $m\in \C$ with $|\Re(m)|<1$ one has
\begin{equation}\label{eq_almost_done_3}
\int_0^\infty  \Ka_m(z^{-1}x) \Ja_m(x)\d x = z^{2m} \int_0^\infty \Ka_m(z^{-1}x) \Ja_{-m}(x)  \d x.
\end{equation}
We infer then from this equality that
\begin{equation}\label{eq_useful}
\Big(\frac{k}{2}\Big)^{2m}\cF_mP_m(-k^2)=
\Big(\frac{\Q}{2}\Big)^{2m}\cF_{-m}P_m(-k^2),
\end{equation}
and as a direct consequence,
\begin{equation}\label{eq_equality}
\cF_{m,\kappa}^{\mp \t} R_{m,\kappa}(-k^2)=(\Q^2+k^2)^{-1}  \cF_{m,\kappa}^{\mp \t}.
\end{equation}
This equality corresponds to one of the identities of our theorem, the proof of the
other identity being analogous.
\end{proof}

\subsection{Spectral projections}\label{sec_proj_2}

We first describe the spectral projections corresponding to
eigenvalues of $H_{m,\kappa}$.

\begin{proposition}\label{prop_aa}
For any $-k^2\in\sigma_\p(H_{m,\kappa})$ one has
\begin{equation*}
\one_{\{-k^2\}}(H_{m,\kappa})=P_m(-k^2).
\end{equation*}
\end{proposition}

\begin{proof}
Let $\gamma$ be a contour encircling $-k^2$ in the complex plane, with no other eigenvalue inside $\gamma$
and with no intersection with $[0,\infty[$.
We use \eqref{eq_dif_res_1bis}, and then compute the residue of
the resolvent at $z=-k^2$:
\begin{align*}
\one_{\{-k^2\}}(H_{m,\kappa})
& =-\frac{1}{2\pi \i}\int_\gamma R_{m,\kappa}(z)\d z\\
&=-\frac{1}{2\pi \i}\int_\gamma \frac{\varsigma\frac{(-z)^m}{2^{2m}}} {1-\varsigma\frac{(-z)^m}{2^{2m}}}\frac{m}{z}P_m(z) \\
&= - \frac{\varsigma\frac{(-z)^m}{2^{2m}}}{\frac{\d}{\d z}\big(1-\varsigma\frac{(-z)^m}{2^{2m}}\big)} \frac{m}{z}P_m(z)\Big|_{z=-k^2} \\
& =\ P_m(-k^2).
\qedhere
\end{align*}
\end{proof}

Let us now assume that $(m,\kappa)$ is not an exceptional pair.
As in Section \ref{sec_proj_1} we can consider for any $0<
  a<b<\infty $ the operator
\begin{equation}\label{eq_proj_m_kappa_1}
\one_{[a,b]}(H_{m,\kappa}) := 2\int_{\sqrt{a}}^{\sqrt{b}}p_{m,\kappa}(k^2) k\;\!\d k
\end{equation}
which is bounded from  $\langle\Q\rangle^{-s}L^2(\R_+)$
to $\langle\Q\rangle^{s}L^2(\R_+)$ for any $s>\frac{1}{2}$.
Its kernel is given for $x,y\in \R_+$ by the expression
\begin{align}\label{eq_proj_m_kappa_2}
& \one_{[a,b]}(H_{m,\kappa})(x,y) \\
&= 2\int_{\sqrt{a}}^{\sqrt{b}}\frac{\Big(\Ja_{m}(kx)- \varsigma\big(\tfrac{k}{2}\big)^{2m}\Ja_{-m}(kx)\Big)
\Big(\Ja_{m}(ky)- \varsigma\big(\tfrac{k}{2}\big)^{2m}\Ja_{-m}(ky)\Big)}{\pi
\Big( \sin^2(\pi m) + \big(\cos(\pi m)-\varsigma\big(\tfrac{k}{2}\big)^{2m}\big)^2\Big)}\;\!\d k.\nonumber
\end{align}

\begin{proposition}\label{prop_diag}
For any $0< a <b$, any $m\in \C^\times$ with $|\Re(m)|<1$, and any $\kappa\in \C\cup \{\infty\}$ with $(m,\kappa)$
not an exceptional pair one has
\begin{equation}\label{eq_diag_m_kappa}
\one_{[a,b]}(H_{m,\kappa})=\cF_{m,\kappa}^{\pm}\;\!\one_{[a,b]}(\Q^2)\;\!\cF_{m,\kappa}^{\mp \t}
\end{equation}
in $\B\big(L^2(\R_+)\big)$. In addition, $\one_{[a,b]}(H_{m,\kappa})$ is a projection.
\end{proposition}

\begin{proof}
Let us recall that the l.h.s.~has been defined in \eqref{eq_proj_m_kappa_1}, and check that
the r.h.s.~corresponds to this expression. Indeed for any $f\in C_{\rm c}(\R_+)$
and $x\in \R_+$ one has
\begin{align*}
&\big(\cF_{m,\kappa}^{\pm}\;\!\one_{[a,b]}(\Q^2)\;\!\cF_{m,\kappa}^{\mp \t} f\big)(x) \\
& = \e^{\pm \i \frac{\pi}{2}m}\sqrt{\frac2\pi} \int_{\sqrt{a}}^{\sqrt{b}}
\frac{\Ja_{m}(xk)- \varsigma \big(\tfrac{k}{2}\big)^{2m} \Ja_{-m}(xk)}{
1-\varsigma \e^{\pm\i \pi m}\big(\tfrac{k}{2}\big)^{2m}}\;\!\big(\cF_{m,\kappa}^{\mp \t} f\big)(k)\;\!\d k \\
& = \frac{2}{\pi}\int_{\sqrt{a}}^{\sqrt{b}} \bigg(
\frac{\Ja_{m}(xk)- \varsigma \big(\tfrac{k}{2}\big)^{2m} \Ja_{-m}(xk)}{
1-\varsigma \e^{\pm\i \pi m}\big(\tfrac{k}{2}\big)^{2m}}\   \times \\
& \qquad \times \int_0^\infty
\frac{\Ja_{m}(ky)- \varsigma \big(\tfrac{k}{2}\big)^{2m} \Ja_{-m}(ky)}{
1-\varsigma \e^{\mp \i \pi m}\big(\tfrac{k}{2}\big)^{2m}}\;\!f(y)\;\!\d y\bigg)\d k \\
& = \int_0^\infty\one_{[a,b]}(H_{m,\kappa})(x,y)\;\! f(y)\;\!\d y,
\end{align*}
where the kernel $\one_{[a,b]}(H_{m,\kappa})(x,y)$ has been defined in \eqref{eq_proj_m_kappa_2}.
Note that since the r.h.s.~of \eqref{eq_diag_m_kappa} defines a bounded operator on $L^2(\R_+)$, this
equality provides a natural continuous extension of $\one_{[a,b]}(H_{m,\kappa})$ as a bounded operator on $L^2(\R_+)$.
Finally, by Lemma \ref{lem_proj_m_kappa} one readily infers that $\one_{[a,b]}(H_{m,\kappa})$ is a projection.
\end{proof}

Note that the equality
\begin{equation}\label{eq_diag_m_kappa1}
\one_{\Xi}(H_{m,\kappa})=\cF_{m,\kappa}^{\pm}\;\!\one_{\Xi}(\Q^2)\;\!\cF_{m,\kappa}^{\mp \t}
\end{equation}
extends \eqref{eq_diag_m_kappa} to any Borel subset $\Xi$ of $\R_+$.
In particular, $\one_{\R_+}(H_{m,\kappa})$ obtained with this definition corresponds to
the one already introduced in \eqref{def_1}.

\subsection{M{\o}ller  operators and scattering operator}

In this section we extend the results obtained for the M{\o}ller
and the scattering operators
to the larger family of operators $H_{m,\kappa}$.
For that purpose, we consider pairs $(m,\kappa)$ and $(m',\kappa')$ which are not exceptional.
The easiest way to introduce the wave operators for
the pair $(H_{m,\kappa},H_{m',\kappa'})$ is to define them using the Hankel
transformations:
\begin{equation}
W_{m,\kappa;m',\kappa'}^\pm=  \cF_{m,\kappa}^{\pm}\;\! \cF_{m',\kappa'}^{\mp \t}.
\label{moller1}
\end{equation}

These definitions immediately imply the following identities:

\begin{proposition}
The following identities hold:
\begin{align}
\label{inve_3} W_{m,\kappa;m',\kappa'}^{\mp \t} \;\!W_{m,\kappa;m',\kappa'}^\pm & = \one_{\R_+}(H_{m',\kappa'}), \\
\nonumber W_{m,\kappa;m',\kappa'}^\pm \;\!W_{m,\kappa;m',\kappa'}^{\mp \t}&= \one_{\R_+}(H_{m,\kappa}), \\
\nonumber W_{m,\kappa;m',\kappa'}^{\pm\t} & = W_{m',\kappa';m,\kappa}^{\mp},\\
\nonumber W_{m,\kappa;m',\kappa'}^\pm H_{m',\kappa'} & = H_{m,\kappa}W_{m,\kappa;m',\kappa'}^\pm .
\end{align}
\end{proposition}

By \eqref{inve_3},
$W_{m,\kappa;m',\kappa'}^{-\t}$ can be treated as an inverse of
$W_{m,\kappa;m',\kappa'}^+$. Therefore, we define
the scattering operator as
\begin{equation*}
S_{m,\kappa;m',\kappa'}:=W_{m,\kappa;m',\kappa'}^{-\t}W_{m,\kappa;m',\kappa'}^-.
\end{equation*}
Clearly, the scattering operator can be expressed in terms of the
Hankel transform:
\begin{equation*}
S_{m,\kappa;m',\kappa'}=\cF_{m',\kappa'}^{+}\;\! \cF_{m,\kappa}^{- \t}\;\!\cF_{m,\kappa}^{-} \;\!\cF_{m',\kappa'}^{+\t}.
\end{equation*}

In order to  analyze the scattering operator it is convenient to
introduce the operators
\begin{equation}\label{reso3}
\cG_{m,\kappa}^\mp:=\cF_{m,\kappa}^{\mp\t} \cF_{m,\kappa}^{\mp}.
\end{equation}

\begin{proposition}
The following equalities hold:
\begin{equation*}
\cG_{m,\kappa}^\mp
= \e^{\mp \i \pi m}\;\!\frac{\one-\varsigma\e^{\pm \i\pi m}(\frac{\Q}{2})^{2m}}{\one-\varsigma\e^{\mp \i\pi m}(\frac{\Q}{2})^{2m}}.
\end{equation*}
\end{proposition}

\begin{proof}
The proof consists in a simple computation, starting from the expressions \eqref{eq_a_gauche} and \eqref{eq_a_droite}
and taking the equality \eqref{eq_to_prove} into account.
\end{proof}

Let us stress that $\cG_{m,\kappa}^\mp$ are simply functions of $\Q$.
Finally, by using the Hankel transformations, one can obtain a diagonal
representation of scattering operators. These operators are expressed in terms
of the operators \eqref{reso3}, namely
\begin{equation*}
\cF_{m',\kappa'}^{- \t} \;\!S_{m,\kappa;m',\kappa'} \;\!\cF_{m',\kappa'}^{+}
= \cG_{m,\kappa}^-\cG_{m',\kappa'}^+
= \cF_{m',\kappa'}^{+\t} \;\!S_{m,\kappa;m',\kappa'}\;\!\cF_{m',\kappa'}^{-}.
\end{equation*}

In the non-exceptional case, the operator $H_{m,\kappa}$ generates a bounded
one-parameter group, at least on the range of the projection
$\one_{\R_+}(H_{m,\kappa})$, by the formula
\begin{equation}\label{def_eq_group}
\e^{\i tH_{m,\kappa}}\one_{\R_+}(H_{m,\kappa})=\one_{\R_+}(H_{m,\kappa})\e^{\i tH_{m,\kappa}}
=\cF_{m,\kappa}^{\pm}\e^{\i t\Q^2}\cF_{m,\kappa}^{\mp \t}.
\end{equation}
In this context we can then show that
$W_{m,\kappa;m',\kappa'}^\pm$ correspond to the M{\o}ller operators
as usually defined in the time-dependent scattering theory.

\begin{proposition}\label{prop_wave1}
For any $m,m'\in \C$ with $|\Re(m)|<1$ and $|\Re(m')|<1$, and for any $\kappa,\kappa'\in \C\cup \{\infty\}$
such that $(m,\kappa)$ and $(m',\kappa')$ are not exceptional pairs,
the M{\o}ller  operators exist and coincide with the operators
defined in \eqref{moller1}:
\begin{equation*}
\slim_{t\to \pm \infty}
\one_{\R_+}(H_{m,\kappa})\e^{\i tH_{m,\kappa}}\e^{-\i tH_{m',\kappa'}}\one_{\R_+}(H_{m',\kappa'})
= W_{m,\kappa;m',\kappa'}^\pm.
\end{equation*}
\end{proposition}

\begin{proof}
By \eqref{def_eq_group} we have
\begin{equation*}
\one_{\R_+}(H_{m,\kappa}) \e^{\i tH_{m,\kappa}}\e^{-\i tH_{m',\kappa'}}\one_{\R_+}(H_{m',\kappa})
=  \cF_{m,\kappa}^{\pm}\e^{\i t\Q^2}\cF_{m,\kappa}^{\mp \t}\cF_{m',\kappa'}^{\pm}\e^{-\i t\Q^2}\cF_{m',\kappa'}^{\mp \t}.
\end{equation*}
Let us then observe that
\begin{align}\label{eq_product}
& \cF_{m,\kappa}^{\mp \t} \cF_{m',\kappa'}^{\pm} \\
\nonumber & = \frac{\e^{\mp \i \frac{\pi}{2}m}}{1-\varsigma \e^{\mp \i \pi m} \big(\tfrac{\Q}{2}\big)^{2m}}
\Big(\Xi_m(-A)-\varsigma \big(\tfrac{\Q}{2}\big)^{2m}\Xi_{-m}(-A)  \Big) \times \\
\nonumber & \quad \times
\Big(\Xi_{m'}(A)-\varsigma'\Xi_{-m'}(A) \big(\tfrac{\Q}{2}\big)^{2m'} \Big)
\frac{\e^{\pm \i \frac{\pi}{2}m'}}{1-\varsigma' \e^{\pm \i \pi m'} \big(\tfrac{\Q}{2}\big)^{2m'}}.
\end{align}
By using \eqref{eq-eq} and Lemma \ref{lem_limit_A} one infers that
\begin{align*}
&\slim_{t\to \pm \infty} \e^{\i t\Q^2}
\Big(\Xi_m(-A)-\varsigma \big(\tfrac{\Q}{2}\big)^{2m}\Xi_{-m}(-A)  \Big) \times \\
& \qquad \qquad \times \Big(\Xi_{m'}(A)-\varsigma'\Xi_{-m'}(A) \big(\tfrac{\Q}{2}\big)^{2m'} \Big) \e^{-\i t \Q^2} \\
&= \Big(\e^{\pm \i \frac{\pi}{2}m}-\varsigma \big(\tfrac{\Q}{2}\big)^{2m}\e^{\mp \i \frac{\pi}{2}m}  \Big)
\Big(\e^{\mp \i \frac{\pi}{2}m'}-\varsigma'\e^{\pm \i \frac{\pi}{2}m'} \big(\tfrac{\Q}{2}\big)^{2m'} \Big) \\
&= \e^{\pm \i \frac{\pi}{2}(m-m')} \Big(1-\varsigma \e^{\mp \i \pi m} \big(\tfrac{\Q}{2}\big)^{2m}  \Big)
\Big(1-\varsigma' \e^{\pm \i \pi m'} \big(\tfrac{\Q}{2}\big)^{2m'} \Big).
\end{align*}
This together with \eqref{eq_product} yields
\begin{equation}\label{eq_exis}
\slim_{t\to \pm \infty} \e^{\i t\Q^2}  \cF_{m,\kappa}^{\mp \t} \cF_{m',\kappa'}^{\pm}\e^{-\i t \Q^2} =\one
\end{equation}
which directly implies the statement.
\end{proof}

\section{Continuous spectrum of $H_0^\nu$}
\setcounter{equation}{0}

In this section we mimic the computations and results of the previous
section analyzing  the family of operators $H_0^\nu$.

\subsection{Resolvent}

From now on we consider $\nu\in \C\cup\{\infty\}$ and $x\in\R_+$.
Let us also fix $k\in \C$ with $\Re(k)>0$, and $-k^2\not\in \sigma(H_0^\nu)$.
Our first aim is to compute the integral kernel of the resolvent of $H_0^\nu$
\begin{equation*}
R_{0}^\nu(-k^2):= (H_0^\nu+k^2)^{-1}.
\end{equation*}

For that purpose, recall first that the map $x\mapsto \Ka_0(kx)$ satisfies \eqref{eq_kernel} and belongs to $L^2(\R_+)$.
We then consider the map
\begin{equation*}
x\mapsto \Ia_0(kx)+\frac{\pi}{2(\gamma+\ln(\tfrac{k}{2})-\nu)} \Ka_0(kx),
\end{equation*}
and infer from \eqref{eq_Im_asym} and \eqref{eq:besselfunc2_1}
that this map and the map $x\mapsto \Ia_0(kx)$
satisfy the equation \eqref{eq_kernel} as well as the equation \eqref{eq_do3} or \eqref{eq_do4} around $0$.
In addition, their Wronskian is equal to $k$.

From the general theory of Sturm-Liouville operators,
as recalled for example in \cite[Prop.~A.1]{BDG},
one deduces that the kernel of $R_0^\nu(-k^2)$ is given by the expression
\begin{equation*}
R_0^\nu(-k^2;x,y)= \frac{1}k
\left\{\begin{matrix}
\Big(\Ia_0(kx)+\frac{\pi}{2\big(\gamma +\ln(\tfrac{k}{2})-\nu\big)} \Ka_0(kx)\Big)\;\! \Ka_0(ky) & \hbox{ if } 0 < x < y, \\
\Big(\Ia_0(ky)+\frac{\pi}{2\big(\gamma+\ln(\tfrac{k}{2})-\nu\big)} \Ka_0(ky)\Big)\;\! \Ka_0(kx) & \hbox{ if } 0 < y < x.
\end{matrix}\right.
\end{equation*}

Let us also observe that the following relation holds:
\begin{equation}\label{eq_resolv_0}
R_0^{\nu}(-k^2)= R_0(-k^2)+\frac{1}{2k^2(\gamma + \ln(\frac{k}{2})-\nu)} P_0(-k^2),
\end{equation}
where $P_0(-k^2)$ is the projection defined in \eqref{proj2}.
Hence $R_0^{\nu}(-k^2)$ is well-defined except for
$\nu =\gamma + \ln(\frac{k}{2})$. This restriction
corresponds to the eigenvalue of $H_0^\nu$, as already mentioned in Proposition \ref{prop_spec}.
Note also that since $R_{0}^\nu(-k^2)$ is a rank one perturbation of $R_{0}(-k^2)$, one again infers
that $[0,\infty[$ belongs to the spectrum of $H_{0}^\nu$.
This justifies {\it a posteriori} our choice for $\Re(k)>0$.

\subsection{Boundary value of the resolvent and spectral density}\label{sec_spec_2-}

Since $[0,\infty[$ belongs to the spectrum of $H_0^\nu$
it is natural to mimic the computations performed in Section \ref{sec_spec_1}.
Note that it will be convenient to use the function $\Ya_0$, as introduced in Section \ref{sec_Neumann}.
Note also that since the special case $\nu=\infty$ has already been treated in Section
\ref{sec_homogeneous}, when considering the operator $H_0$, we shall
not consider it again.

\begin{proposition}
Let $\nu \in \C$, and let $k>0$.
\begin{enumerate}
\item[(i)] If $\nu$ is not an exceptional value, then the boundary values of the resolvent
\begin{equation*}
R_0^\nu(k^2\pm \i 0):=\lim_{\epsilon \searrow 0}  R_{0}^\nu(k^2\pm \i \epsilon)
\end{equation*}
exist in the sense of operators from  $\langle\Q\rangle^{-s}L^2(\R_+)$
to $\langle\Q\rangle^{s}L^2(\R_+)$ for any $s>\frac{1}{2}$,
uniformly in $k$ on each compact subset of $\R_+$.
The kernel of $R_{0}^\nu(k^2\pm \i 0)$ is given for $0<x\leq y$ by
\begin{align*}
& R_{0}^\nu(k^2\pm\i0; x,y) \\
& = \frac{\pm \i}{k\big(\gamma +\ln\big(\frac{k}{2}\big)-\nu\mp \i \frac{\pi}{2}\big)}
\Big(\big(\gamma+\ln\big(\tfrac{k}{2}\big)- \nu \big)\Ja_0(kx)- \tfrac{\pi}{2} \Ya_0(kx)\Big)
\Ha_0^\pm(ky).
\end{align*}
and the same expression with the role of $x$ and $y$ exchanged for $0<y<x$.
\item[(ii)] If $\nu$ is an exceptional value, then the same statement holds
for $k$ uniformly on each compact subset of $\R_+\setminus \{2\e^{\Re(\nu)-\gamma}\}$.
\end{enumerate}
\end{proposition}

\begin{proof}
The starting point for the proof of this statement is formula \eqref{eq_resolv_0}.
In addition, since the first term on the r.h.s.~of this equality has already been treated in Section
\ref{sec_spec_1}, we shall concentrate on the second term only.
For that purpose and as in the proof of Theorem \ref{thm_boundary}
we consider for $s>\frac{1}{2}$ and $x,y>0$ the expression
\begin{equation}\label{eq_ker1}
\langle x\rangle^{-s}\Ka_0\big(\sqrt{-k^2\mp\i \epsilon}\;\!x\big)\Ka_0\big(\sqrt{-k^2\mp\i \epsilon}\;\!y\big) \langle y\rangle^{-s}.
\end{equation}
By the estimate provided in \eqref{eq_K_0} one easily infers that
this kernel belongs to $L^2(\R_+\times\R_+)$. In addition, since
this kernel converges pointwise to
\begin{equation*}
\pm \i \langle x\rangle^{-s}\Ha_0^\pm(kx)\Ha_0^\pm(ky) \langle y\rangle^{-s}
\end{equation*}
one concludes by the Lebesgue Dominated Convergence Theorem that
this convergence also holds in $L^2(\R_+\times\R_+)$, which is equivalent to a convergence in the Hilbert-Schmidt norm.
Note that the uniform convergence in $k$ on compact subsets of $\R_+$ can be directly checked,
as well as the convergence of the prefactors,
as long as  this factor has no singularity.

For the computation of the kernel of $R_{0}^\nu(k^2\pm\i0)$ one has for $0 < x \leq y$
\begin{align*}
& R_{0}^\nu(k^2\pm\i0;x,y) \\
& = R_{0}(k^2\pm\i0;x,y)
- \frac{\pi }{2 k\big(\gamma + \ln\big(\frac{\mp \i k}{2}\big)-\nu\big)}
\;\!\Ha_0^\pm(kx) \Ha_0^\pm(ky) \\
& = \pm \frac{\i}{k}\Big(\Ja_0(kx) \pm \i \frac{\frac{\pi}{2}}{\big(\gamma + \ln\big(\frac{k}{2}\big)-\nu \mp \i \frac{\pi}{2}\big)}
\;\!\Ha_0^\pm(kx)\Big) \Ha_0^\pm(ky)  \\
& = \frac{\pm \i}{k\big(\gamma+ \ln\big(\frac{k}{2}\big)-\nu \mp \i \frac{\pi}{2}\big)}
\Big(\big(\gamma+ \ln\big(\tfrac{k}{2}\big)-\nu\big)\Ja_0(kx) -  \tfrac{\pi}{2} \Ya_0(kx)\Big)
\Ha_0^\pm(ky),
\end{align*}
as expected. For $0<y\leq x$ the same expression can be obtained, with the role of $x$ and $y$ exchanged.
\end{proof}

Based on the previous result, the corresponding spectral density can now be computed,
namely if $\nu$ is not exceptional for any $k>0$ and for any $s>\frac{1}{2}$ one has
\begin{equation*}
p_{0}^\nu(k^2):= \frac{1}{2\pi \i}\Big(R_{0}^\nu\big(k^2+\i0\big) -R_{0}^\nu\big(k^2-\i0\big)\Big) \in \B\big( \langle\Q\rangle^{-s}L^2(\R_+),\langle\Q\rangle^{s}L^2(\R_+)\big).
\end{equation*}
If $\nu$ is exceptional, the same formulas hold once a suitable restriction on $k$
has been imposed. In the sequel, this restriction will be made tacitly.

\begin{proposition}
The kernel of the spectral density is given by the following formula:
\begin{align*}
& p_{0}^\nu(k^2;x,y) \\
& = \frac{\big((\gamma + \ln(\frac{k}{2})-\nu )\Ja_0(kx) - \frac{\pi}{2} \Ya_0(kx)\big)
\big((\gamma + \ln(\frac{k}{2})-\nu )\Ja_0(ky)- \frac{\pi}{2} \Ya_0(ky)\big)}
{\pi k \big((\gamma+ \ln(\frac{k}{2})-\nu)^2+ (\frac{\pi}{2})^2\big)}.
\end{align*}
\end{proposition}

\begin{proof}
For $0 < x \leq y$ one has
\begin{align*}
&2\pi k p_{0}^\nu(k^2;x,y) \\
& =  \frac{1}{\gamma + \ln\big(\frac{k}{2}\big)-\nu - \i \frac{\pi}{2}}
\Big(\big(\gamma + \ln\big(\tfrac{k}{2}\big)-\nu\big)\Ja_0(kx)- \tfrac{\pi}{2} \Ya_0(kx)\Big)
\Ha_0^+(ky) \\
&\ \ + \frac{1}{\gamma + \ln\big(\frac{k}{2}\big)-\nu + \i \frac{\pi}{2}}
\Big(\big(\gamma + \ln\big(\tfrac{k}{2}\big) -\nu \big)\Ja_0(kx)- \tfrac{\pi}{2} \Ya_0(kx)\Big)
\Ha_0^-(ky) \\
& = \frac{\big((\gamma + \ln(\frac{k}{2})-\nu )\Ja_0(kx) - \frac{\pi}{2} \Ya_0(kx)\big)
\big((\gamma + \ln(\frac{k}{2})-\nu )\Ja_0(ky)- \frac{\pi}{2} \Ya_0(ky)\big)}
{\frac{1}{2}\big(\gamma + \ln(\frac{k}{2})-\nu - \i \frac{\pi}{2}\big)
\big(\gamma + \ln(\frac{k}{2})-\nu + \i \frac{\pi}{2}\big)}.
\end{align*}
Since the role of $x$ and $y$ can be exchanged, one directly gets the statement.
\end{proof}

\subsection{Generalized Hankel transform}

Let us now define the \emph{incoming} and \emph{outgoing Hankel transformations} $\cF_0^{\nu \mp}$ given by
the kernels
\begin{align}
\nonumber \cF_0^{\nu \mp}(x,y) & :=\sqrt{\frac2\pi}
\Big(\Ja_0(xy) \pm \frac{\i \frac{\pi}{2}}{\gamma+ \ln\big(\frac{y}{2}\big)-\nu \mp \i \frac{\pi}{2}}
\;\!\Ha_0^\pm(xy)\Big)\\
\label{eq_to_be_used}& = \sqrt{\frac2\pi}
\bigg(\frac{\big(\gamma+ \ln\big(\frac{y}{2}\big)-\nu \big)\Ja_0(xy)- \frac{\pi}{2} \Ya_0(xy)}
{\gamma + \ln\big(\frac{y}{2}\big)-\nu \mp \i \frac{\pi}{2}}\bigg).
\end{align}
Note that we have written two expressions since they will be both useful
later on. Note also since the denominator has a singularity if $\nu$ is exceptional, we shall
ignore this special case in the sequel.

In order to have a better picture of the maps $\cF_0^{\nu \mp}$ we recall that
Mellin-Barnes representation of $\Ha_0^\pm$ has been provided in \eqref{eq_H_0}.
Based on this formula, the following statement can be proved. Note
that $C_{\rm b}(\R)$ denotes the set of continuous and bounded functions on $\R$.

\begin{lemma}
The map $\H_0^\pm : C_{\rm c}(\R_+)\to L^2(\R_+)$ with kernel
$$
\H^\pm_0(x,y):=\sqrt{\frac2\pi} \Ha^\pm_0(xy)
$$
continuously extends to a bounded operator of the form $\Ju\Xi_0^\pm(A)= \Xi_0^\pm(-A)\Ju$ with
\begin{equation}\label{eq_Xi_pm}
\Xi_0^\pm(t):=\frac{1}{\pi} \Gamma\big(\tfrac{1+\i t}{2}\big)^2\e^{\i \ln(2)t}\e^{\mp\frac{\pi}{2}t}
\end{equation}
and $\Xi_0^\pm\in C_{\rm b}(\R)$.
\end{lemma}

\begin{proof}
The operator $\Ju \H_0^\pm : C_{\rm c}(\R_+)\to L^2(\R_+)$ has kernel
\begin{equation*}
\sqrt{\frac2\pi} \frac{1}{x} \Ha^\pm_0\Big(\frac{y}{x}\Big)
= \frac{1}{2\pi^{2}}\frac{1}{\sqrt{xy}}
\int_{-\infty}^{+\infty}\Gamma\big(\tfrac{1+\i t}{2}\big)^2\e^{\i \ln(2)t}\e^{\mp\frac{\pi}{2}t} \frac{y^{-\i t}}{x^{-\i t}}\d t.
\end{equation*}
By using \cite[Lem.~6.4]{BDG} one infers that this kernel corresponds to the kernel of the operator
defined by $\Xi^\pm_0(A)$, see also the proof of Proposition \ref{prop_sur_A}.

Clearly, the map $t\mapsto \frac{1}{\pi} \Gamma\big(\tfrac{1+\i t}{2}\big)^2\e^{\i \ln(2)t}\e^{\mp\frac{\pi}{2}t}$
is continuous and locally bounded. In order to show that it is bounded, let us estimate its asymptotic values as $t\to \pm \infty$.
By a consequence of Lemma \ref{stirling} one gets
\begin{align*}
\Big|\frac{1}{\pi} \Gamma\big(\tfrac{1+\i t}{2}\big)^2\e^{\i \ln(2)t}\e^{\mp\frac{\pi}{2}t}\Big|
&=2\e^{-\frac{\pi}{2}|t|}\e^{\mp\frac{\pi}{2}t}\Big(1+O(t^{-1})\Big).
\end{align*}
One then infers that $\Xi^\pm_0\in C_{\rm b}(\R)$, as expected.
It also means that $\Xi^{\pm}_0(A)$ extends continuously to a bounded operator in $L^2(\R_+)$.
\end{proof}

Based on the previous result one directly infers the following statement:

\begin{lemma}
For any $\nu\in \C$  with $\nu$ not exceptional the operator $\cF_0^{\nu \mp}$ extends continuously to the following operator:
\begin{equation}\label{eq_hp1}
\begin{split}
&\Ju \Big(\Xi_0(A) \pm
\Xi_0^\pm(A) \frac{\i \frac{\pi}{2}}{\gamma+\ln(\frac{\Q}{2})-\nu\mp \i \frac{\pi}{2}}\Big) \\
& =  \Big(\Xi_0(-A) \pm
\Xi_0^\pm(-A)  \frac{\i \frac{\pi}{2}}{\gamma-\ln(2\Q)-\nu\mp \i \frac{\pi}{2}}\Big)\Ju
\end{split}
\end{equation}
with $\Xi_0$ defined in \eqref{eq_def_Xi} and $\Xi_0^\pm$ defined in \eqref{eq_Xi_pm}.
\end{lemma}

Let us still provide the expression for the adjoint of this operator, namely
\begin{equation}\label{eq_hp2}
\begin{split}
(\cF_0^{\nu \mp})^\t & = \Big(\Xi_0(-A) \pm \frac{\i \frac{\pi}{2}}{\gamma +\ln(\frac{\Q}{2})-\nu \mp \i \frac{\pi}{2}}\Xi_0^\pm(-A)
\Big)\Ju \\
& = \Ju \Big(\Xi_0(A) \pm \frac{\i \frac{\pi}{2}}{\gamma - \ln(2\Q)-\nu \mp \i \frac{\pi}{2}} \Xi_0^\pm(A)
\Big)
\end{split}
\end{equation}
In order to derive alternative formulas for these operators, let us first recall the equality
\begin{equation}\label{eq_Gamma}
\Gamma\big(\tfrac{1+\i t}{2}\big)\Gamma\big(\tfrac{1-\i t}{2}\big) = \frac{\pi}{\cosh\big(\frac{\pi}{2}t\big)} \qquad \forall t\in \R,
\end{equation}
and prove the following statement.

\begin{lemma}
The map $\cY_0 : C_{\rm c}(\R_+)\to L^2(\R_+)$ with kernel
\begin{equation*}
\cY_0(x,y):=\sqrt{\frac2\pi} \Ya_0(xy)
\end{equation*}
continuously extends to the bounded operator $\i \Ju \Xi_0(A)\tanh\big(\frac{\pi}{2}A\big)$.
\end{lemma}

\begin{proof}
From the equality $\Ya_0=\mp \i (\Ha_0^\pm-\Ja_0)$ one infers that
\begin{equation*}
\cY_0=\mp \i \big(\H_0^\pm-\cF_0\big) =  \mp \i \Ju \big(\Xi_0^\pm(A)-\Xi_0(A)\big)
\end{equation*}
with
\begin{align*}
& \Xi_0^\pm(t)-\Xi_0(t) =\e^{i\ln(2)t}\frac{\Gamma\big(\frac{1+\i t}{2}\big)}{\Gamma\big(\frac{1-\i t}{2}\big)}
\Big(\tfrac{1}{\pi} \Gamma\big(\tfrac{1+\i t}{2}\big)\Gamma\big(\tfrac{1-\i t}{2}\big)\e^{\mp \frac{\pi}{2}t}-1\Big) \\
& = \Xi_0(t)\Big(\frac{1}{\cosh\big(\frac{\pi}{2}t\big)}\e^{\mp \frac{\pi}{2}t}-1\Big)
= \mp \Xi_0(t) \frac{\sinh\big(\frac{\pi}{2}t\big)}{\cosh\big(\frac{\pi}{2}t\big)},
\end{align*}
which directly leads to the statement.
\end{proof}

\begin{corollary}
For any $\nu \in \C$ not exceptional the following alternative description of $\cF_0^{\nu \mp}$ and $(\cF_0^{\nu \mp})^\t$ hold:
\begin{equation}\label{eq_alter_1}
\cF_0^{\nu \mp} = \cF_0 \Big(\gamma +\ln\big(\tfrac{\Q}{2}\big)-\nu -\i \tfrac{\pi}{2} \tanh\big(\tfrac{\pi}{2}A\big)\Big)
\frac{1}{\gamma + \ln\big(\frac{\Q}{2}\big)-\nu \mp \i \frac{\pi}{2}}
\end{equation}
and
\begin{equation}\label{eq_alter_2}
(\cF_0^{\nu \mp})^\t
= \frac{1}{\gamma+ \ln\big(\frac{\Q}{2}\big)- \nu \mp \i \frac{\pi}{2}}
\Big(\gamma +\ln\big(\tfrac{\Q}{2}\big)-\nu +\i \tfrac{\pi}{2} \tanh\big(\tfrac{\pi}{2}A\big)
\Big)\cF_0.
\end{equation}
\end{corollary}

\begin{proof}
The statement follows directly from the definition \eqref{eq_to_be_used}, the previous result and
the equality $\tanh(-t)=-\tanh(t)$ for any $t\in \R$.
\end{proof}

We can now derive additional properties of the operator $\cF_0^{\nu \pm}$.

\begin{lemma}
For any $\nu\in \C$ not exceptional the equalities $(\cF_0^{\nu \pm})^\t\cF_0^{\nu \mp}=\one$ hold.
\end{lemma}

\begin{proof}
The proof consists in computing the product of the terms \eqref{eq_hp1} and \eqref{eq_hp2} (with the correct sign)
and in checking that this product is equal to $\one$. For that purpose, one first observes that for any $t\in \R$ one has
\begin{equation}\label{eq_Xi_0}
\Xi_0(-t)\Xi_0(t)=1.
\end{equation}
By taking \eqref{eq_Gamma} into account, one also observes that
\begin{equation}\label{eq_Xi_1}
\Xi_0(-t)\Xi_0^\pm(t) =  \Xi_0^\mp(-t)\Xi_0(t) = \frac{\e^{\mp \frac{\pi}{2}t}}{\cosh\big(\frac{\pi}{2}t\big)}
\end{equation}
and that
\begin{equation}\label{eq_Xi_2}
\Xi_0^\mp(-t) \Xi_0^\pm(t)= \left(\frac{\e^{\mp \frac{\pi}{2}t}}{\cosh\big(\frac{\pi}{2}t\big)}\right)^2.
\end{equation}

Then, by a few algebraic manipulations, one easily reduces the statement to the following equality in the form sense on $C_{\rm c}(\R_+)$:
\begin{equation}\label{eq_to_check}
\mp \frac{1}{\pi}\left[\i \frac{\e^{\mp \frac{\pi}{2}A}}{\cosh\big(\frac{\pi}{2}A\big)}, \ln(\Q) \right]
- \frac{\e^{\mp \frac{\pi}{2}A}}{\cosh\big(\frac{\pi}{2}A\big)}
+ \tfrac{1}{2}  \left(\frac{\e^{\mp \frac{\pi}{2}A}}{\cosh\big(\frac{\pi}{2}A\big)}\right)^2=0.
\end{equation}
In order to check this equality, let us recall that the equality $[\i A,\ln(\Q)]=\one$ holds, once suitably defined.
One then infers that
\begin{align}\label{eq_comput}
\nonumber & \mp \frac{1}{\pi}\left[\i \frac{\e^{\mp \frac{\pi}{2}A}}{\cosh\big(\frac{\pi}{2}A\big)}, \ln(\Q) \right] \\
\nonumber & = \frac{1}{\cosh\big(\frac{\pi}{2}A\big)^2}\Big(\tfrac{1}{2}\e^{\mp \frac{\pi}{2}A}\cosh\big(\tfrac{\pi}{2}A\big)
\pm \tfrac{1}{2} \e^{\mp \frac{\pi}{2}A}\sinh\big(\tfrac{\pi}{2}A\big)\Big) \\
& = \frac{\e^{\mp \frac{\pi}{2}A}}{\cosh\big(\frac{\pi}{2}A\big)}
- \tfrac{1}{2}  \frac{\big(\e^{\mp \frac{\pi}{2}A}\big)^2}{\cosh\big(\frac{\pi}{2}A\big)^2},
\end{align}
where the equalities $\cosh(y)=\frac{\e^y+\e^{-y}}{2}$ and $\sinh(y) = \frac{\e^y-\e^{-y}}{2}$ have been used for the last equality.
This final expression leads directly to the equality \eqref{eq_to_check}.
\end{proof}

Once again it is natural to set
\begin{equation*}
\one_{\R_+}(H_{0}^\nu) :=  \cF_0^{\nu \pm }\;\!(\cF_0^{\nu\mp})^\t
\end{equation*}
which is again a projection.

\begin{proposition}\label{prop_nightmare}
Let $\nu\in \C$ be not exceptional. Then, for any $k\in \C$ with $\Re(k)>0$ and $-k^2\not \in \sigma_{\rm p}(H_{0}^\nu)$
the following equalities hold:
\begin{equation*}
R_0^\nu(-k^2) \one_{\R_+}(H_0^\nu)
=\cF_0^{\nu\pm}(\Q^2+k^2)^{-1} (\cF_0^{\nu \mp})^\t
= \one_{\R_+}(H_0^\nu)R_0^\nu(-k^2).
\end{equation*}
\end{proposition}

\begin{proof}
Let us first prove that
\begin{equation}\label{eq_perm_4}
(\cF_0^{\nu \mp})^\t  R_0^\nu(-k^2) - (\Q^2+k^2)^{-1} (\cF_0^{\nu \mp})^\t =0
\end{equation}
which implies the second equality of the statement.
By taking into account the expression for $(\cF_0^{\nu \mp})^\t$ provided in \eqref{eq_alter_2}
together with the equality \eqref{eq_resolv_0} one observes that the l.h.s.~of \eqref{eq_perm_4} is equal to
\begin{align*}
& \frac{1}{\gamma+ \ln\big(\frac{\Q}{2}\big)- \nu \mp \i \frac{\pi}{2}}\bigg\{
\Big(\gamma +\ln\big(\tfrac{\Q}{2}\big)-\nu +\i \tfrac{\pi}{2} \tanh\big(\tfrac{\pi}{2}A\big)
\Big)\cF_0 R_0^\nu(-k^2) \\
& \ \ - (\Q^2+k^2)^{-1}
\Big(\gamma +\ln\big(\tfrac{\Q}{2}\big)-\nu +\i \tfrac{\pi}{2} \tanh\big(\tfrac{\pi}{2}A\big)
\Big)\cF_0  \bigg\}  \\
& = \frac{1}{\gamma+ \ln\big(\frac{\Q}{2}\big)- \nu \mp \i \frac{\pi}{2}}\bigg\{
\i \tfrac{\pi}{2} \big[\tanh\big(\tfrac{\pi}{2}A\big), (\Q^2+k^2)^{-1}\big] \cF_0 \\
& +\frac{1}{2k^2(\gamma + \ln(\frac{k}{2})-\nu)}\Big(\gamma +\ln\big(\tfrac{\Q}{2}\big)-\nu +\i \tfrac{\pi}{2} \tanh\big(\tfrac{\pi}{2}A\big)
\Big)\cF_0  P_0(-k^2)\bigg\}.
\end{align*}
By a few algebraic computations, one then observes that the term inside the curly bracket would be equal to $0$
if
\begin{equation}\label{eq_horror_1}
\pi k^2 \big[\i\tanh\big(\tfrac{\pi}{2}A\big), (\Q^2+k^2)^{-1}\big] = - \cF_0  P_0(-k^2)\cF_0
\end{equation}
and
\begin{equation}\label{eq_horror_2}
\Big(\ln\big(\tfrac{\Q}{k}\big) +\i \tfrac{\pi}{2} \tanh\big(\tfrac{\pi}{2}A\big)
\Big)\cF_0  P_0(-k^2)\cF_0=0.
\end{equation}

From now on, let us compute some kernels, by always considering $x,y\in \R_+$.
Since by \eqref{beqo}
$\int_0^\infty \Ja_0(xy)\Ka_0(ky)\d y = \frac{\sqrt{kx}}{x^2+k^2}$,
one first infers that the kernel of the operator
$\cF_0  P_0(-k^2)\cF_0$ is given by
\begin{equation}\label{eq_yes_0}
\cF_0  P_0(-k^2)\cF_0(x,y)= 2k^2 (xy)^{1/2}\;\!\frac{1}{x^2+k^2}\;\!\frac{1}{y^2+k^2}.
\end{equation}
On the other hand, the kernel of $\tanh\big(\tfrac{\pi}{2}A\big)$ is given by (see for example the proof of \cite[Lem.~9.2]{R16})
\begin{equation*}
\i\tanh\big(\tfrac{\pi}{2}A\big)(x,y)=-\frac{2}{\pi} {\mathrm{Pv}}\bigg(\frac{1}{\frac{x}{y}-\frac{y}{x}}\bigg)(xy)^{-1/2},
\end{equation*}
where ${\mathrm{Pv}}$ denotes the principal value distribution. One then infers the following kernel
\begin{align}\label{eq_yes_1}
\nonumber &\pi k^2 \big[\i\tanh\big(\tfrac{\pi}{2}A\big), (\Q^2+k^2)^{-1}\big](x,y) \\
\nonumber &=-2k^2{\mathrm{Pv}}\bigg(\frac{1}{\frac{x}{y}-\frac{y}{x}}\bigg)(xy)^{-1/2}\Big(\frac{1}{y^2+k^2}-\frac{1}{x^2+k^2}\Big) \\
&= -2k^2 (xy)^{1/2}\;\! \frac{1}{y^2+k^2}\;\!\frac{1}{x^2+k^2}.
\end{align}
By comparing \eqref{eq_yes_0} with \eqref{eq_yes_1} one directly gets \eqref{eq_horror_1}.

In order to check \eqref{eq_horror_2} let us first deduce from \eqref{eq_yes_0} that
\begin{equation}\label{eq_yes_2}
\Big(\ln\big(\tfrac{\Q}{k}\big) \cF_0  P_0(-k^2)\cF_0\Big)(x,y)
= 2k^2 \ln\big(\tfrac{x}{k}\big)(xy)^{1/2}\;\!\frac{1}{x^2+k^2}\;\!\frac{1}{y^2+k^2}.
\end{equation}
On the other hand, since
\begin{equation*}
\int \Big(\frac{1}{x^2-y^2}\Big)\frac{y}{y^2+1}\d y
= \frac{1}{2(x^2+1)}\;\!\ln\bigg(\frac{x^2+1}{|x^2-y^2|}\bigg),
\end{equation*}
one can easily compute the kernel
$-\i \tfrac{\pi}{2} \tanh\big(\tfrac{\pi}{2}A\big)\cF_0  P_0(-k^2)\cF_0$
and observe that it corresponds to the expression obtained in \eqref{eq_yes_2}.
This finishes the proof of the equality \eqref{eq_horror_2}.
\end{proof}

\subsection{Spectral projections}

Let us start again with a result about the spectral projection corresponding to the eigenvalues of $H_0^\nu$.
The proof of the following statement can be mimicked from the proof of Proposition \ref{prop_aa}.

\begin{proposition}\label{prop_aa2}
For any $-k^2\in\sigma_\p(H_0^\nu)$ one has
\begin{equation*}
\one_{\{-k^2\}}(H_0^\nu)=P_0(-k^2).
\end{equation*}
\end{proposition}

As in Section \ref{sec_proj_1} or \ref{sec_proj_2}  we can also consider for any $0<a<b$ the operator
\begin{equation*}
\one_{[a,b]}(H_0^\nu) := 2\int_{\sqrt{a}}^{\sqrt{b}}p_0^\nu(k^2) k\;\!\d k,
\end{equation*}
which is bounded from  $\langle\Q\rangle^{-s}L^2(\R_+)$
to $\langle\Q\rangle^{s}L^2(\R_+)$ for any $s>\frac{1}{2}$.
Its kernel is given for $x,y\in \R_+$ by the expression
\begin{align}\label{def_H_0_nu_ker}
& \one_{[a,b]}(H_0^\nu)(x,y) \\
&= 2\int_{\sqrt{a}}^{\sqrt{b}}
\tfrac{\big((\gamma + \ln(\frac{k}{2})-\nu )\Ja_0(kx) - \frac{\pi}{2} \Ya_0(kx)\big)
\big((\gamma + \ln(\frac{k}{2})-\nu )\Ja_0(ky)- \frac{\pi}{2} \Ya_0(ky)\big)}
{\pi \big((\gamma+ \ln(\frac{k}{2})-\nu)^2+ (\frac{\pi}{2})^2\big)}
\;\!\d k.\nonumber
\end{align}

One can now obtain a result similar to the one contained in Proposition \ref{prop_diag}.

\begin{proposition}
For any $0< a <b$ and any $\nu\in \C \cup\{\infty\}$ not exceptional one has
\begin{equation}\label{to_be_extended}
\one_{[a,b]}(H_0^{\nu})=\cF_0^{\nu\pm}\;\!\one_{[a,b]}(\Q^2)\;\!(\cF_0^{\nu \mp})^\t
\end{equation}
in $\B\big(L^2(\R_+)\big)$. In addition, $\one_{[a,b]}(H_0^\nu)$ is a projection.
\end{proposition}

\begin{proof}
The proof can be mimicked from the one of Proposition \ref{prop_diag}.
The new necessary information are the kernel of $\one_{[a,b]}(H_0^\nu)(x,y)$,
which is provided in \eqref{def_H_0_nu_ker}, and the equality recalled in \eqref{eq_to_be_used}.
\end{proof}

Finally, observe that the equality
\begin{equation*}
\one_{\Xi}(H_0^{\nu})=\cF_0^{\nu\pm}\;\!\one_{\Xi}(\Q^2)\;\!(\cF_0^{\nu \mp})^\t
\end{equation*}
extends \eqref{to_be_extended} to any Borel subset $\Xi$ of $\R_+$.

\subsection{M{\o}ller operators and scattering operator}

In this section we extend the results obtained for the M{\o}ller
and the scattering operators
to the family of operators $H_0^\nu$.
As before, we shall consider any $\nu\in \C$ which is not an exceptional value.

For the pair $(H_0^\nu,H_0^{\nu'})$ using the Hankel
transformations we define
\begin{equation}
W_{0;0}^{\nu;\nu';\pm}=  \cF_{0}^{\nu\pm}\;\! (\cF_{0}^{\nu' \mp})^\t.
\label{moller2}
\end{equation}
Note that the M{\o}ller operators for any pairs $(H_0^\nu,H_{m,\kappa})$ or $(H_{m,\kappa},H_0^\nu)$
could be defined similarly,
but the corresponding expressions can also be directly obtained
by using the chain rule. For that reason, we shall not analyze them separately.

Based on these definitions and on the results obtained so far one easily infers
some identities:

\begin{proposition}
The following identities hold:
\begin{align}
\label{inve2} \big(W_{0;0}^{\nu;\nu';\mp}\big)^\t \;\!W_{0;0}^{\nu;\nu';\pm}
& = \one_{\R_+}(H_0^{\nu'}), \\
\nonumber W_{0;0}^{\nu;\nu';\pm} \;\!\big(W_{0;0}^{\nu;\nu';\mp}\big)^\t&= \one_{\R_+}(H_0^\nu), \\
\nonumber \big(W_{0;0}^{\nu;\nu';\pm}\big)^\t & = W_{0;0}^{\nu';\nu;\mp},\\
\nonumber W_{0;0}^{\nu;\nu';\pm} H_0^{\nu'} & = H_{0}^\nu W_{0;0}^{\nu;\nu';\pm}.
\end{align}
\end{proposition}

By \eqref{inve2},
$\big(W_{0;0}^{\nu;\nu';-}\big)^\t$ can be treated as an inverse of
$W_{0;0}^{\nu;\nu';+}$. Therefore, we define
the scattering operator as
\begin{equation*}
S_{0;0}^{\nu;\nu'}:=\big(W_{0;0}^{\nu;\nu';-}\big)^\t W_{0;0}^{\nu;\nu';-}.
\end{equation*}
Clearly, the scattering operator can be expressed in terms of the
Hankel transform:
\begin{equation*}
S_{0;0}^{\nu;\nu'}=\cF_0^{\nu' +}\;\! (\cF_0^{\nu -})^\t\;\!\cF_0^{\nu -} \;\!(\cF_0^{\nu' +})^\t.
\end{equation*}

In order to  analyze the scattering operator it is convenient to
introduce the operators
\begin{equation}\label{reso4}
\cG_0^{\nu \mp}:=(\cF_0^{\nu \mp})^\t\;\!\cF_0^{\nu \mp}.
\end{equation}

\begin{proposition}
The following equalities hold:
\begin{equation*}
\cG_0^{\nu\mp} =  \frac{\gamma+\ln\big(\frac{\Q}{2}\big)-\nu\pm \i \frac{\pi}{2}}{\gamma+\ln\big(\frac{\Q}{2}\big)-\nu\mp \i \frac{\pi}{2}}.
\end{equation*}
\end{proposition}

\begin{proof}
Starting with \eqref{eq_hp1} and \eqref{eq_hp2}, and
by taking successively the equalities \eqref{eq_Xi_0}, \eqref{eq_Xi_1} and \eqref{eq_Gamma} into account, one infers that
\begin{align*}
(\cF_0^{\nu \mp})^\t\;\!\cF_0^{\nu \mp}
& = \Ju \Big\{
\one  \pm \frac{\i \frac{\pi}{2}}{\gamma-\ln(2\Q)-\nu\mp \i \frac{\pi}{2}}
\;\! \frac{\e^{\mp \frac{\pi}{2}A}}{\cosh\big(\frac{\pi}{2}A\big)} \\
& \quad \pm \frac{\e^{\pm \frac{\pi}{2}A}}{\cosh\big(\frac{\pi}{2}A\big)} \;\!\frac{\i \frac{\pi}{2}}{\gamma-\ln(2\Q)-\nu\mp \i \frac{\pi}{2}} \\
& \quad + \frac{\i \frac{\pi}{2}}{\gamma-\ln(2\Q)-\nu\mp \i \frac{\pi}{2}}\;\!
\frac{1}{\cosh^2\big(\frac{\pi}{2}A\big)}\;\!
\frac{\i \frac{\pi}{2}}{\gamma-\ln(2\Q)-\nu\mp \i \frac{\pi}{2}}\Big\}\Ju.
\end{align*}
By observing then that
\begin{align*}
& \frac{\i \frac{\pi}{2}}{\gamma-\ln(2\Q)-\nu\mp \i \frac{\pi}{2}}
\;\! \frac{\e^{\mp \frac{\pi}{2}A}}{\cosh\big(\frac{\pi}{2}A\big)} \\
& = \frac{\e^{\mp \frac{\pi}{2}A}}{\cosh\big(\frac{\pi}{2}A\big)} \;\!
\frac{\i \frac{\pi}{2}}{\gamma-\ln(2\Q)-\nu\mp \i \frac{\pi}{2}} \\
&\quad  + \frac{\pi}{2}\frac{1}{\gamma-\ln(2\Q)-\nu\mp \i \frac{\pi}{2}}
\Big[i\frac{\e^{\mp \frac{\pi}{2}A}}{\cosh\big(\frac{\pi}{2}A\big)},\ln(\Q)\Big]
\frac{1}{\gamma-\ln(2\Q)-\nu\mp \i \frac{\pi}{2}}
\end{align*}
and by using \eqref{eq_comput} together with some algebraic manipulations, one directly infers that
\begin{equation*}
(\cF_0^{\nu \mp})^\t\;\!\cF_0^{\nu \mp}
= \Ju \Big\{\one \pm 2 \frac{\i \frac{\pi}{2}}{\gamma-\ln(2\Q)-\nu\mp \i \frac{\pi}{2}} \Big\}\Ju
= \Ju \frac{\gamma-\ln(2\Q)-\nu\pm \i \frac{\pi}{2}}{\gamma-\ln(2\Q)-\nu\mp \i \frac{\pi}{2}} \Ju.
\end{equation*}
The statement follows then by the definition of $\Ju$.
\end{proof}

Let us stress once again that $\cG_0^{\nu \mp}$ are simply functions of $\Q$.
Finally, by using the Hankel transformations, one can obtain a diagonal
representation of scattering operators. These operators are expressed in terms
of the operators \eqref{reso4}, namely
\begin{equation*}
(\cF_{0}^{\nu' -})^\t \;\!S_{0;0}^{\nu;\nu'}\;\!\cF_{0}^{\nu' +}
= \cG_0^{\nu -}\cG_0^{\nu'+}
= (\cF_{0}^{\nu' +})^\t \;\!S_{0;0}^{\nu;\nu'} \;\!\cF_{0}^{\nu'-}.
\end{equation*}

In the non-exceptional case, the operator $H_0^\nu$ generates a bounded
one-parameter group, at least on the range of the projection
$\one_{\R_+}(H_0^\nu)$, by the formula
\begin{equation*}
\e^{\i tH_0^\nu}\one_{\R_+}(H_0^\nu)=
\one_{\R_+}(H_0^\nu)  \e^{\i tH_0^\nu}
:=\cF_0^{\nu\pm}\e^{\i t\Q^2}(\cF_0^{\nu \mp})^\t.
\end{equation*}
We finally show that
$W_{0;0}^{\nu;\nu';\pm}$ correspond to the M{\o}ller operators
as usually defined by the time-dependent scattering theory.

\begin{proposition}
For any $\nu,\nu'\in \C$ not exceptional
the M{\o}ller  operators exist and coincide with the operators
defined in \eqref{moller2}:
\begin{equation*}
\slim_{t\to \pm \infty}
\one_{\R_+}(H_0^\nu)\e^{\i tH_0^\nu}\e^{-\i tH_0^{\nu'}}\one_{\R_+}(H_0^{\nu'})
=  W_{0;0}^{\nu;\nu';\pm}.
\end{equation*}
\end{proposition}

\begin{proof}
The proof is parallel to the proof of Proposition \ref{prop_wave1}.
One first observes that
\begin{align*}
(\cF_{0}^{\nu \mp})^\t \cF_{0}^{\nu' \pm}
& = \Big(\Xi_0(-A) \pm \frac{\i \frac{\pi}{2}}{\gamma+\ln\big(\frac{\Q}{2}\big)-\nu\mp \i \frac{\pi}{2}}
\;\!\Xi_0^\pm(-A)\Big)  \\
& \quad \times
\Big(\Xi_0(A) \mp \Xi_0^\mp(A) \frac{\i \frac{\pi}{2}}{\gamma +\ln\big(\frac{\Q}{2}\big)-\nu' \pm \i \frac{\pi}{2}}
\Big).
\end{align*}
By using the explicit expression for the products of $\Xi_0$ and $\Xi_0^\pm$ as given in \eqref{eq_Xi_0},
\eqref{eq_Xi_1} and \eqref{eq_Xi_2},
one then infers that for the four maps $t\mapsto \Xi_{0}(-t)\Xi_{0}(t)$,
$t\mapsto \Xi_{0}^\pm(-t)\Xi_{0}(t)$, $t\mapsto \Xi_{0}(-t)\Xi_{0}^\mp(t)$
and $t\mapsto \Xi_{0}^\pm(-t)\Xi_{0}^\mp(t)$ belong to
$C\big([-\infty,\infty]\big)$ with
\begin{eqnarray*}
\Xi_{0}(-\infty)\Xi_{0}(\infty) = 1 & \hbox{and} & \Xi_{0}(\infty)\Xi_{0}(-\infty)=1, \\
\Xi_{0}^+(-\infty)\Xi_{0}(\infty) = 2 & \hbox{and} & \Xi_{0}^+(\infty)\Xi_{0}(-\infty)= 0,\\
\Xi_{0}^-(-\infty)\Xi_{0}(\infty) = 0 & \hbox{and} & \Xi_{0}^-(\infty)\Xi_{0}(-\infty)= 2,\\
\Xi_{0}^+(-\infty)\Xi_{0}^-(\infty) = 4 & \hbox{and} & \Xi_{0}^+(\infty)\Xi_{0}^-(-\infty)=0.
\end{eqnarray*}

Based on these observations one directly deduces from Lemma \ref{lem_limit_A} that
\begin{align*}
&\slim_{t\to \pm \infty} \e^{\i t\Q^2}    (\cF_{0}^{\nu \mp})^\t \cF_{0}^{\nu' \pm} \e^{-\i t \Q^2} \\
&=\slim_{t\to \pm \infty} \e^{\i t\Q^2}
\Big(\Xi_0(-A) \pm \frac{\i \frac{\pi}{2}}{\gamma+\ln\big(\frac{\Q}{2}\big)-\nu\mp \i \frac{\pi}{2}}
\;\!\Xi_0^\pm(-A)\Big) \times \\
& \qquad \qquad \qquad \times
\Big(\Xi_0(A) \mp \Xi_0^\mp(A) \frac{\i \frac{\pi}{2}}{\gamma +\ln\big(\frac{\Q}{2}\big)-\nu' \pm \i \frac{\pi}{2}}
\Big) \e^{-\i t \Q^2} \\
& = \one .
\end{align*}
The remaining argument is similar to the one of the proof of Proposition \ref{prop_wave1}.
\end{proof}

\appendix

\section{Bessel family for dimension $1$}\label{secB1}
\setcounter{equation}{0}

In this section, we gather several properties of the Bessel family for dimension $1$.
If no additional restriction is imposed, the parameter $m$ is an arbitrary element of $\C$,
but a special attention is often required when $m\in \{\dots,-2,-1\}$.
All the described properties follow from the theory of the
usual Bessel family of dimension 2 described in literature. However,
in general, we outline an independent derivation.

\subsection{The function $\Ia_m$}

The \emph{modified Bessel function for dimension $1$}, denoted $\Ia_m$, is defined by an everywhere convergent power series
or by the so-called Schl\"afli integral representation:
\begin{align}\label{asym2}
\Ia_m(z)& = \sum_{n=0}^\infty\frac{\sqrt\pi\left(\frac{z}{2}\right)^{2n+m+\frac12}}{n!\Gamma(m+n+1)}\\
\nonumber &=\frac{\sqrt{z}}{\sqrt{2\pi}}\int_0^{\pi}
\e^{z\cos(\phi)}\cos(m\phi)\;\!\d\phi\\
\nonumber & \quad -\frac{\sqrt{z}}{\sqrt{2\pi}}\sin(\pi m)\int_0^\infty
\e^{-z\cosh(\beta)-m\beta}\;\!\d\beta \quad \hbox{ for }\Re(z)>0.
\end{align}
Note that since $z\mapsto 1/\Gamma(z)$ is an entire function, the first
expression is meaningful even for $m\in \{\dots,-2,-1\}$.
Clearly, the function $\Ia_m$ is is analytic on $\cc\setminus ]-\infty,0]$, and it is analytic on $\C$ for $m\in \Z+\frac{1}{2}$.

The function $\Ia_m$ has the following asymptotics for $z$ near $0$:
\begin{equation}\label{eq_Im_asym}
\Ia_m(z)=  \frac{\sqrt\pi}{\Gamma(m+1)} \left(\frac{z}{2}\right)^{m+\frac12} + O(|z|^{\Re(m)+\frac{5}{2}}).
\end{equation}
On the other hand, for $m\in \Z$ we have
\begin{equation*}
\Ia_{-m}(z)=\Ia_{m}(z)
=\frac{\sqrt{z}}{\sqrt{2\pi}}\int_0^\pi\e^{z\cos(\phi)}\cos(m\phi)\;\!\d\phi.
\end{equation*}
For any $m$, the analytic continuation around $0$ by the angle $\pm\pi$ multiplies
$\Ia_m$ by a phase factor, namely
\begin{equation*}
\Ia_m(\e^{\pm\i\pi}z)=\e^{\pm\i \pi (m+\frac12)}\Ia_m(z).
\end{equation*}
The value $\Ia_m(z)$ is real for $z>0$ and for $m\in\rr$, and more generally one has
\begin{equation*}
\overline{\Ia_m(z)}=\Ia_{\bar m}(\bar z).
\end{equation*}
There also exists a generating function, namely for any $t\in \C^\times$ and any $z\in \cc\setminus ]-\infty,0]$~:
\begin{equation*}
\frac{\sqrt{\pi z}}{\sqrt{2}}
\exp\left(\frac{z}{2}(t+t^{-1})\right)=\sum_{m=-\infty}^\infty
t^m\Ia_m(z).
\end{equation*}

\subsection{The function $\Ka_m$}\label{subsec_K_m}

The \emph{MacDonald function for dimension $1$}, denoted $\Ka_m$,
can be defined for $\Re(z)>0$ by the integral representation
\begin{equation}
\Ka_m(z):=\frac{\sqrt{z}}{\sqrt{2\pi}}\int_0^\infty\exp
\left(-\frac{z}{2}(s+s^{-1})\right)s^{-m-1}\d s.
\label{defo}\end{equation}
It extends by analytic continuation onto a larger domain with a
possible branch point at $0$.
Outside of the basic region $\Re(z)>0$ one can obtain
other integral formulas by an appropriate
deformation of the contour of integration in \eqref{defo},
see for example \cite[Sec.~6.22]{W}.

For $m\not \in \Z$ the functions $\Ia_m$ and $\Ka_m$ are connected by the equality
\begin{equation}\label{macdo1}
\Ka_m(z)= \frac{1}{\sin(\pi m)}\big(\Ia_{-m}(z)-\Ia_{m}(z)\big).
\end{equation}
For $m\in \Z$, the relation \eqref{macdo1} can be extended by
l'H\^opital's rule:
\begin{equation}\label{lap10}
\Ka_m(z) =\frac{(-1)^{m+1}}{\pi}
\left(\frac{\d}{\d n}\Ia_n(z)\Big|_{n=-m}
+\frac{\d}{\d n}\Ia_n(z)\Big|_{n=m}\right).
\end{equation}
By setting $h(n):=\sum\limits_{k=1}^{n-1}\frac{1}{k}$ and $h(1)=0$
the equality \eqref{lap10} leads
for $m\in $
to
\begin{align}
\nonumber \Ka_m(z)&=(-1)^{m+1} \frac{2}{\pi}\;\!\Big(\ln\Big(\frac{z}{2}\Big)+\gamma\Big)\;\!\Ia_m(z)\\
\nonumber &\quad +\frac{1}{\sqrt\pi}\sum_{k=0}^{m-1}
(-1)^k \frac{(m-k-1)!}{k!} \Big(\frac{z}{2}\Big)^{2k-m+\frac12} \\
&\quad +\frac{(-1)^m}{\sqrt\pi}\sum_{k=0}^\infty
\frac{h(k+1)+h(m+k+1)}{k!(m+k)!}\Big(\frac{z}{2}\Big)^{2k+m+\frac12},
\label{bess5}
\end{align}
where $\gamma$ is the Euler's constant.
 Note that only the first and the third terms
are present in the special case $m=0$.

The function $\Ka_m$ is analytic for $\cc\setminus]-\infty,0]$
and satisfies $\Ka_{-m}(z)=\Ka_m(z)$.
We also have
\begin{equation*}
\overline{\Ka_m(z)}=\Ka_{\bar m}(\bar z),
\end{equation*}
and the value $\Ka_m(z)$ is real for $z>0$ and for $m\in\rr$ or $m\in\i\rr$.

The function $\Ka_m$ has a well-defined asymptotic at infinity, namely for
any $\epsilon>0$ and $|\arg(z)|<\pi-\epsilon$:
\begin{equation}\label{tri4a}
\Ka_m(z) = \e^{-z}\big(1+O(|z|^{-1})\big).
\end{equation}

The asymptotics near $0$ can be obtained from \eqref{macdo1} and
\eqref{bess5}.
We present the asymptotics in the strip $|\Re(m)|<1$:
\begin{equation}\label{eq:besselfunc2_1}
\Ka_m(z) =
\left\{ \begin{array}{lcl}
-\frac{\sqrt{2  z}}{\sqrt\pi}\big(\ln\left(\frac{z}{2}\right)+\gamma\big)+O\big(|z|^{\frac52}\ln(|z|)\big)
& {\rm if} & m=0,\\
\frac{\Gamma(m)}{\sqrt\pi}
\left(\frac{z}{2}\right)^{-m+\frac12}+O(|z|^{\Re(m)+\frac{1}{2}}) &   {\rm if} & 1>\Re (m)>0,\\
\frac{\Gamma(-m)}{\sqrt\pi}
\left(\frac{z}{2}\right)^{m+\frac12} +O(|z|^{-\Re(m)+\frac{1}{2}})  & {\rm if} & -1<\Re (m)<0, \\
\frac{\sqrt{\pi}}{\sin(\pi m)} \left(\frac{z}{2}\right)^{\frac12}
\big(\frac{(z/2)^{-m}}{\Gamma(1-m)}-\frac{(z/2)^m}{\Gamma(1+m)}\big) +O(|z|^{\frac{5}{2}})
& \rm if & m\in \i \R^\times.
\end{array} \right.
\end{equation}
Actually, we will only need  the following estimates, valid for
all $m$: For $|z|<1$ we have
\begin{equation*}
|\Ka_m(z)| \leq
\left\{ \begin{array}{lcl}
C_0|z|^{\frac12}(1+|\ln(z)|)
& {\rm if} & m=0,\\
C_m|z|^{-{ |\Re (m)|} +\frac12}
&  {\rm if} & m\neq 0,
\end{array} \right.
\end{equation*}
for some constants $C_0$ and $C_m$ independent of $z$.

Let us also mention
another relation between the functions $\Ia_m$ and $\Ka_m$, namely
\begin{equation*}
\Ia_m(z)=\frac12\big(\Ka_m(\e^{-\i\pi}z) -\e^{\i\pi(m-\frac12)}\Ka_m(z)\big).
\end{equation*}

\subsection{Additional properties of $\Ia_m$ and $\Ka_m$}\label{subsec_Wronsk}

As already mentioned in Section \ref{sec_Bessel_eq} the functions
$\Ia_m$ and $\Ka_m$ are solutions of the modified Bessel equation for dimension $1$:
\begin{align}
\left(\partial_z^2-\Big(m^2-\frac14\Big)\frac{1}{z^2}-1\right)v(z) & = 0.
\end{align}
They form a basis of solutions of this equation.

For the three functions $\Ia_m$, $\Ia_{-m}$ and $\Ka_m$ their respective
Wronskians \eqref{eq_Wronsk} can be computed and are independent of $z$, namely:
\begin{align*}
W_z(\Ka_m,\Ia_m)&=1,\\
W_z(\Ka_m,\Ia_{-m})&=1,\\
W_z(\Ia_m,\Ia_{-m})&=-\sin(\pi m).
\end{align*}

Let $\La_m$ denote either the functions $\Ia_m$
or the function $\e^{i\pi m}\Ka_m$.
Then the following
\emph{contiguous relations} are satisfied:
\begin{align*}
2m \La_m(z)&=z\La_{m-1}(z)-z\La_{m+1}(z), \\
2m \partial_z \La_m(z)&=\Big(m+\frac12\Big)\La_{m-1}(z)+\Big(m-\frac12\Big)
\La_{m+1}(z).
\end{align*}
They imply the following \emph{recurrence relations}~:
\begin{align*}
\partial_z\left(z^{m-\frac12} \La_m(z)\right)& =z^{m-\frac12}\La_{m-1}(z), \\
\left(\partial_z+\Big(m-\frac12\Big)\frac{1}{z}\right)\La_m(z)& =\La_{m-1}(z), \\
\partial_z\left(z^{-m-\frac12} \La_m(z)\right)&=z^{-m-\frac12}\La_{m+1}(z), \\
\left(\partial_z+\Big(-m-\frac12\Big)\frac{1}{z}\right)\La_m(z)&= \La_{m+1}(z).
\end{align*}

Let us finally observe that for $m=\pm\frac12$ the functions
$\Ia_m$ and $\Ka_m$ coincide with well-known elementary functions:
\begin{align*}
\Ka_{\pm\frac12}(z)&= \e^{-z},\\
\Ia_{-\frac12}(z)&= \cosh(z),\\
\Ia_{\frac12}(z)&= \sinh(z).
\end{align*}
The simplicity of these relations is one of the motivations
for introducing the Bessel family for dimension $1$.

\subsection{The function $\Ja_m$}

The \emph{Bessel function for dimension $1$}, denoted  $\Ja_m$, is defined by the following relations
\begin{align}
\label{forsure11} \Ja_m(z)& = \e^{\pm\i\frac{\pi}{2}(m+\frac12)}\Ia_m(\e^{\mp\i\frac{\pi}{2}}z),\\
\nonumber &= \sum_{n=0}^\infty\frac{(-1)^n\sqrt\pi\left(\frac{z}{2}\right)^{2n+m+\frac12}}
{n!\Gamma(m+n+1)} \\
\label{forsure1} & = \frac{1}{2}\left(\e^{-\i\frac{\pi}{2}(m+\frac12)}\Ka_m(\e^{-\i\frac{\pi}{2}} z)+
\e^{\i\frac{\pi}{2}(m+\frac12)}\Ka_m(\e^{\i\frac{\pi}{2}} z) \right).
\end{align}
This function is clearly analytic on $\cc\setminus ]-\infty,0]$,
and it is analytic on $\C$ for $m\in\zz+\tfrac12$.

Some additional properties of this function are
\begin{equation*}
\Ja_m(\e^{\pm \i\pi}z)=\e^{\pm\i\pi (m+\frac12)}\Ja_m(z)
\end{equation*}
and
\begin{equation*}
\overline{\Ja_m(z)}=\Ja_{\bar m}(\bar z).
\end{equation*}
From the Taylor expansion, one infers that near $0$ one has
\begin{equation*}
\Ja_m(z)= \frac{\sqrt{\pi}}{\Gamma(m+1)}\Big(\frac{z}{2}\Big)^{m+\frac{1}{2}}+O(|z|^{\Re(m)+\frac{5}{2}})\ .
\end{equation*}
For large $z$ with $|\arg(z)|<\frac{\pi}{2}-\epsilon$ for some $\epsilon>0$ one also has
\begin{equation*}
\Ja_m(z)=\cos\Big(z-\frac{1}{2}\pi m -\frac{1}{4}\pi \Big) + \e^{|\Im(z)|}O(|z|^{-1})\ .
\end{equation*}

\subsection{The functions $\Ha_m^\pm$}

The \emph{Hankel functions for dimension $1$}, denoted $\Ha_m^\pm$, are essentially analytic continuations of the function
$\Ka_m$, one for the lower part and the other one for the upper
part of the complex plane. Indeed, the following relations are satisfied:
\begin{align*}
\Ha_m^{\pm}(z) & = \e^{\mp\i\frac{\pi}{2}(m+\frac12)}
\Ka_m(\e^{\mp\i\frac{\pi}{2}} z)\\
& = \frac{\e^{\mp\i\frac{\pi}{2}(m+\frac12)}}{\sin(\pi m)}
\Big(\Ia_{-m}(\e^{\mp\i\frac{\pi}{2}}z)
-\Ia_{m}(\e^{\mp\i\frac{\pi}{2}}z)\Big),
\end{align*}
from which one also infers that
\begin{equation}\label{eq_H_pm_m}
\Ha_{-m}^{\pm}(z)=\e^{\pm \i \pi m}\Ha_{m}^{\pm}(z).
\end{equation}

Some additional relations between $\Ja_m$ and $\Ha_m^\pm$ are:
\begin{align*}
\Ja_m(z) & =\frac{1}{2}\left(\Ha_m^{+}(z)+\Ha_m^{-}(z)\right),\\
\Ja_{-m}(z) & =\frac{1}{2}\left(\e^{\i \pi m}\Ha_m^{+}(z)+
\e^{-\i \pi m}\Ha_m^{-}(z)\right),\\
\Ha_m^{\pm}(z) &=\frac{-\e^{\mp\i \pi(m+\frac12)}\Ja_m(z)+\e^{\mp\i\frac{\pi}{2}} \Ja_{-m}(z)}
{\sin (\pi m)}.
\end{align*}

The following asymptotic expansions will also be necessary: from
\eqref{eq:besselfunc2_1} one infers that for any $\theta$,
$|\arg(z)|<\theta$, as  $z\to 0$,
\begin{equation}\label{eq_asym_Hm}
\Ha_m^\pm(z)=\left\{ \begin{array}{lcl}
\pm \i \e^{\mp \i \frac{\pi}{2}m}\frac{\sqrt{2}}{\sqrt{\pi}} z^{\frac{1}{2}} \big(\ln(z)+\gamma\mp \i \frac{\pi}{2}\big) + O\big(|z|^{\frac52}\ln(|z|)\big)
& {\rm if} & m=0,\\
\mp \i\frac{\Gamma(m)}{\sqrt\pi}
\left(\frac{z}{2}\right)^{-m+\frac12} + O(|z|^{\Re(m)+\frac{1}{2}})&   {\rm if} & 1>\Re (m)>0,\\
\mp \i \e^{\mp \i \pi m} \frac{\Gamma(-m)}{\sqrt\pi}
\left(\frac{z}{2}\right)^{m+\frac12} + O(|z|^{-\Re(m)+\frac{1}{2}}) & {\rm if} & -1<\Re (m)<0, \\
\mp \i \frac{\sqrt{\pi}}{\sin(\pi m)} \left(\frac{z}{2}\right)^{\frac12}
\big(\frac{(z/2)^{-m}}{\Gamma(1-m)}-\frac{\e^{\mp \i \pi m}(z/2)^m}{\Gamma(1+m)}\big) + O(|z|^{\frac{5}{2}})
& \rm if & m\in \i \R^\times.
\end{array} \right.
\end{equation}

On the other hand, for any $\epsilon >0$ the following asymptotic formulas are true for $|z|\to \infty$ with
$\arg(z)\not \in [\mp\tfrac{\pi}{2}-\epsilon,\mp\tfrac{\pi}{2}+\epsilon]$:
\begin{equation}\label{trib}
\Ha_m^\pm(z) = \e^{\pm\i (z-\frac{1}{2}\pi m-\frac{1}{4}\pi)}\big(1+O(|z|^{-1})\big)\ .
\end{equation}

\subsection{Function $\Ya_m$}\label{sec_Neumann}

The \emph{Neumann function for dimension $1$}, denoted $\Ya_m$, is defined by
\begin{equation*}
\Ya_m(z)=\frac{1}{2\i}\big(\Ha_m^+(z)-\Ha_m^-(z)\big),\qquad
\Ha_m^\pm = \Ja_m\pm \i\Ya_m.
\end{equation*}

The  function $\Ya_m$ is especially useful for $m\in\{0,1,2,\dots\}$, when we have
\begin{align*}
\Ya_m(z)=&\frac{2}{\pi}\Big(\log(\frac{z}{2})+\gamma\Big)\Ja_m(z)\\
&-\frac{1}{\sqrt\pi}\sum\limits_{k=0}^{m-1}\frac{(m-k-1)!}{k!}\Big(\frac{z}{2}\Big)^{2k-m+\frac12}\\
&-\frac{1}{\sqrt\pi}
\sum\limits_{k=0}^\infty (-1)^k\;\!\frac{h(k+1)+h(m+k+1)}{k!(m+k)!}\Big(\frac{z}{2}\Big)^{2k+m+\frac12},
\end{align*}
with the function $h$ introduced in Section \ref{subsec_K_m}.

\subsection{Additional properties of $\Ja_m$, $\Ha_m^\pm$ and $\Ya_m$}

As already mentioned in Section \ref{sec_Bessel_eq}
the functions $\Ja_m$, $\Ha_m^\pm$ and $\Ya_m$ are solutions of the Bessel equation for dimension $1$:
\begin{align}
\left(\partial_z^2-\Big(m^2-\frac14\Big)\frac{1}{z^2}+1\right)v(z) & = 0.
\end{align}
In addition their respective
Wronskian can be computed and are independent of $z$, for example,
\begin{align*}
W_z(\Ja_m,\Ja_{-m})&=-\sin(\pi m),\\
W_z(\Ha_m^-,\Ha_{m}^+)&=2\i.
\end{align*}

Let us still observe that for $m=\pm\frac12$ the resulting functions
coincide with well-known elementary functions:
\begin{align*}
\Ja_{{\frac12}}(z)&= \sin(z), \\
\Ja_{-{\frac12}}(z)&= \cos(z),\\
\Ha_{{\frac12}}^{\pm}(z)&= \e^{\pm \i(z-\frac{\pi}{2})},\\
\Ha_{-{\frac12}}^{\pm}(z)&= \e^{\pm \i z}.
\end{align*}

\subsection{Integral identities}

The following indefinite integrals follow  from the recurrence relations of Section \ref{subsec_Wronsk}:

\begin{align*}
&\int_y^\infty \Ka_m(ax)\Ka_m(bx)\d x\\
&=\frac{1}{a/b-b/a}
\Big(\frac{1}{b}\Ka_{m-1}(ay)\Ka_m(by)-\frac{1}{a}\Ka_m(ay)\Ka_{m-1}(by)\Big),\quad \Re(a+b)>0;\\
& \int_y^\infty \Ka_m(ax)^2\d x \\
& = -\frac{y}{2}\Ka_m(ay)^2+\frac{m}{a}\Ka_m(ay)\Ka_{m-1}(ay)+
\frac{y}{2}\Ka_{m-1}(ay)^2, \quad  \Re(a)>0.
\end{align*}
They imply the following definite integrals:
\begin{eqnarray}
\nonumber \int_0^\infty \Ka_m(ax)\Ka_m(bx)\d x
&=&\frac{1}{\sin(\pi m)} \frac{(a/b)^{m}-(b/a)^{m}}{\sqrt{ab}(a/b-b/a)},\\
\nonumber && \ m\neq0, \ |\Re (m)|<1,\ \Re(a+b)>0,\\
\nonumber \int_0^\infty \Ka_0(ax)\Ka_0(bx)\d x
&=&\frac{2}{\pi}\frac{\ln(a)-\ln(b)}{\sqrt{ab}(a/b-b/a)},\ \ \Re(a+b)>0,\\
\label{int3} \int_0^\infty \Ka_m(ax)^2\d x&=&\frac{ m}
{\sin (\pi m) a},  \\
\nonumber && \ m\neq0, \ |\Re (m)|<1,\ \Re (a)>0,\\
\label{int4} \int_0^\infty \Ka_0(ax)^2\d
x&=&\frac{1}{\pi a},\ \ \Re (a)>0.
\end{eqnarray}
In the same vein, let us also mention the definite integral
\begin{equation}\label{beqo}
\int_0^\infty  \Ka_m(ax) \Ja_m(bx)\d x =
\frac{(a/b)^m}{\sqrt{ab}(a/b+b/a)},
\end{equation}
see for example \cite[Eq.~6.521]{GR}.
We also derive an additional relation which will be useful later on.
For Bessel functions for dimension $2$ this result corresponds to \cite[Eq. 6.541]{GR}.

\begin{proposition}
For any $m\in \C$ with $\Re(m)>-1$ one has
\begin{equation}\label{manipu}
\frac{2}{\pi}\int_0^\infty \Ja_m(xp)\;\!\Ja_m(yp)
\frac{1}{(p^2+k^2)}\;\!\d p\\
=\frac{1}{k}\left\{\begin{matrix}\! \Ia_m(kx)\;\! \Ka_m(ky) & \hbox{ if } 0 < x < y, \\
\Ia_m(ky)\;\! \Ka_m(kx) & \hbox{ if } 0 < y < x.
\end{matrix}\right.
\end{equation}
\end{proposition}

\begin{proof}
For $0<x<y$ one has by \eqref{forsure11} and \eqref{forsure1}
\begin{align*}
&\frac{2}{\pi}\int_0^\infty \Ja_m(xp)\;\!\Ja_m(yp)
\frac{1}{(p^2+k^2)}\;\!\d p\\
&=\frac{1}{\pi}\int_0^\infty\Big(\Ia_m(\e^{-\i\frac\pi2}xp)
\Ka_m(\e^{-\i\frac\pi2}yp)+ \Ia_m(\e^{\i\frac\pi2}xp)
\Ka_m(\e^{\i\frac\pi2}yp)\Big)\frac{\d p}{(p^2+k^2)} \\
&=\frac{1}{\pi}\int_{-\infty}^\infty\Ia_m(-\i xp)
\Ka_m(-\i yp)\frac{\d p}{(p^2+k^2)} \\
&=\frac{2\pi\i}{\pi}\frac{\Ia_m(xk)\Ka_m(yk)}{2\i k} \\
&=\frac{1}{k} \Ia_m(kx)\Ka_m(ky).
\end{align*}
In the last step we used the fact that the integral over the
semicircle on the upper half plane $p=R\e^{\i\phi}$, with $\phi\in[0,\pi]$, goes to
zero as $R\to\infty$. Besides, we have a single singularity of the
integrand inside the contour of integration at $p=\i k$, which is a
simple pole, whose residue has been evaluated.

A similar proof holds for $0<y<x$.
\end{proof}

\subsection{Barnes identities}

Integral representations of Bessel-type functions in terms of the
Gamma function are sometimes called
{\em Barnes identities } from the name of their discoverer. Before we
present them, let us quote a useful result about  asymptotics of the
Gamma function given in  \cite[Cor.~1.4.4]{AAR}, which is a consequence of the
Stirling formula.

\begin{lemma}\label{stirling}
Let $a,b \in \R$ with $a_1\leq a \leq a_2$ for two constants $a_1,a_2$. Then one has
\begin{equation*}
|\Gamma(a+\i  b)|=\sqrt{2\pi}|b|^{a-\frac12}\e^{-\frac{\pi}{2}|b|}\big(1+O(|b|^{-1})\big),
\end{equation*}
where the constant implied by the term $O(|\cdot|)$ depends only on $a_1$ and $a_2$.
\end{lemma}

Let $m\in\C$ and $c\in\R$ with
\begin{equation}\label{stirling1}
c<\frac{\Re (m)}{2}.
\end{equation}
The following representation is a reformulation of
an identity found in \cite[Ch.~VI.5]{W}:
\begin{equation} \label{barnes1}
\Ja_m(x) =  \frac{1}{4\i\sqrt{\pi}} \int_{\gamma}
\frac{\Gamma(c+\frac{s}{2})}{\Gamma(m+1-c-\frac{s}{2})}
\left(\frac{x}{2}\right)^{m+\frac12-2c-s} \d s,
\end{equation}
where $\gamma$ is a contour which asymptotically coincides with
the vertical line
${]-\i\infty,+\i\infty[}$
and passes on the right of $-2c$.
Note that by Lemma \ref{stirling}
\begin{equation*}
\left|\frac{\Gamma(c+\i\frac{t}{2})}{\Gamma(m+1-c-\i\frac{t}{2})}\right|
\leq C (1+|t|)^{2c-\Re (m)-1},
\end{equation*}
hence the condition
\eqref{stirling1} implies the integrability of
the r.h.s.~of \eqref{barnes1}.

If we want that the contour is a straight vertical line, we need to
assume that  $ c\in\, \big]0,\frac{\Re (m)}{2}\big[$
(which implies $c>0$ and $\Re (m)>0$),
and then we can rewrite
\eqref{barnes1} as
\begin{equation} \label{barnes}
\Ja_m(x) =  \frac{1}{4\sqrt{\pi}} \int_{-\infty}^{+\infty}
\frac{\Gamma(c+\i\frac{t}{2})}{\Gamma(m+1-c-\i\frac{t}{2})}
\left(\frac{x}{2}\right)^{m+\frac12-2c-\i t} \d t.
\end{equation}

As shown in the proof of \cite[Lem.~6.3]{BDG}, the validity of \eqref{barnes} can then be extended
in the sense of distribution up to  $\Re(m)>-1$ and $0< c<\Re(m)+1$.
In particular, by choosing $c=\frac{m+1}{2}$ one infers that in the sense of
distributions for $\Re(m)>-1$ we have
\begin{equation}\label{eq_Barnes}
\Ja_m(x)=\frac{1}{4\sqrt\pi}\int_{-\infty}^{+\infty} \frac{\Gamma(\frac{m+\i
t+1}{2})}{\Gamma(\frac{m-\i t+1}{2})}\left(\frac{x}{2}\right)^{-\i
t-\frac12}\d t.
\end{equation}

Let us also consider a representation of Hankel functions similar to
\eqref{barnes1} and valid under the condition
\eqref{stirling1}. The next formula follows from \cite[Sec.~6.5]{W}:
\begin{equation}\label{barnes2}
\Ha_m^\pm(x) = \frac{\e^{\mp\i\pi(m+\frac12)\pm\i \pi c}}{4\i\pi^{\frac32}}
\int_{\gamma'}\Gamma\Big(c+\frac{s}{2}\Big)\Gamma\Big(c-m+\frac{s}{2}\Big)
\e^{\pm \i\frac{\pi}{2}s}\Big(\frac{x}{2}\Big)^{-2c-s+m+\frac12}\d s,
\end{equation}
where $\gamma'$ is a contour which asymptotically coincides with
the vertical line ${]-\i\infty,+\i\infty[}$ and
passes on the right of $-2c$ and $-2c+2m$. As a consequence of Lemma \ref{stirling} one infers that
\begin{equation*}
\left|  \Gamma\Big(c+\frac{\i t}{2}\Big)\Gamma\Big(c-m+\frac{\i t}{2}\Big)
\e^{\mp \frac{\pi}{2}t }\right|\leq
C (1+|t|)^{2c-\Re(m)-1}\e^{-\frac{\pi}{2}|t|}\e^{\mp
\frac{\pi}{2}t},
\end{equation*}
which guarantees the integrability of
\eqref{barnes2}. Here, we cannot choose $\gamma'$ to be a straight
line. However, if we are satisfied with the interpretation of the integral
\eqref{barnes2} in the sense of distributions, then under conditions
 $0<c$ and $\Re( m)<c$ a straight vertical line will work and we can
rewrite \eqref{barnes2} as
\begin{align*}
&\cH_m^\pm(x)\\
&= \frac{\e^{\mp\i
\pi(m+\frac12)\pm\i \pi c}}{4\pi^{\frac32}}
\int_{-\infty}^{+\infty}\Gamma\Big(c+\frac{\i t}{2}\Big)\Gamma\Big(c-m+\frac{\i t}{2}\Big)
\e^{\mp \frac{\pi}{2}t}\Big(\frac{x}{2}\Big)^{-2c-\i t+m+\frac12}\d t.
\end{align*}
In particular, for $-1<\Re(m)<1$ and by setting $c=\frac{\Re(m)+1}{2}$ we obtain after a few manipulations and in the
sense of distributions
\begin{equation}\label{eq_H_0}
\cH_m^\pm(x) = \frac{\e^{\mp\i
\frac{\pi}{2} m}}{4\pi^{\frac32}}
\int_{-\infty}^{+\infty}\Gamma\Big(\frac{-m+1+\i t}{2}\Big)\Gamma\Big(\frac{m+1+\i t}{2}\Big)
\e^{\mp \frac{\pi}{2}t}\Big(\frac{x}{2}\Big)^{-\i t-\frac12}\d t.
\end{equation}

\section{Propagation of the generator of dilations}
\setcounter{equation}{0}

We derive some relations between the generator of dilations $A$ and the
multiplication operator $\Q^2$. Note first that
\begin{equation*}
\e^{\i t\ln(\Q)} \;\!\e^{\i\tau A}\;\! \e^{-\i t\ln(\Q)} = \e^{\i\tau(A-t)}.
\end{equation*}
Therefore $\ln(\Q)$ and $A$ satisfy the canonical commutation
relations, which determines their properties up to unitary equivalence.
It is easy to see that for $\psi \in C\big([-\infty,\infty]\big)$
one has
\begin{equation*}
\e^{\i t\ln(\Q)} \psi(-A) \e^{-\i t\ln(\Q)} = \psi(-A+t)
\end{equation*}
and consequently
\begin{equation}\label{eq_help}
\slim_{t\to \pm \infty} \e^{\i t\ln(\Q)} \psi(-A) \e^{-\i t\ln(\Q)} = \psi(\pm \infty).
\end{equation}

The next lemma contains a similar result with the operator $\ln(\Q)$
replaced by $\Q^2$. It can be obtained by an abstract argument
based on Mourre theory, see \cite{R} for the details. We will give an
alternative elementary proof.
Note that since $\Q^2 = \varphi(\ln(\Q))$ with $\varphi(u)=\e^{2u}$ for any $u\in \R$ and since
$\varphi'>0$, the following statement has a flavor similar to the invariance principle in scattering theory.

\begin{lemma}\label{lem_limit_A}
For any $\psi\in C\big([-\infty,\infty]\big)$ the following equalities hold:
\begin{equation}\label{eq_to_be_proved}
\slim_{t\to \pm \infty} \e^{\i t\Q^2} \psi(-A)\;\! \e^{-\i t \Q^2} = \psi(\pm \infty).
\end{equation}
\end{lemma}

\begin{proof}
Let us first note that it is enough to show that
\begin{equation}\label{eq_to_be_proved1}
\wlim_{t\to \pm \infty} \e^{\i t\Q^2} \psi(-A)\;\! \e^{-\i t \Q^2} = \psi(\pm \infty).
\end{equation}
Indeed, this easily follows from the equality
\begin{align*}
& \big\|\big(\e^{\i t\Q^2} \psi(-A)\;\! \e^{-\i t \Q^2}-\psi(\pm\infty)\big)f\big\|^2 \\
& = \big(f|\e^{\i t\Q^2} |\psi|^2(-A)\;\! \e^{-\i t \Q^2}f\big) + \|\psi(\pm\infty)f\|^2\\
& \quad -\big(\psi(\pm\infty)f|\e^{\i t\Q^2} \psi(-A)\;\! \e^{-\i t \Q^2} f\big)
-\big(\e^{\i t\Q^2} \psi(-A)\;\! \e^{-\i t \Q^2} f|\psi(\pm\infty)f\big).
\end{align*}

We now introduce the unitary transformation $W: L^2(\R)\to L^2(\R_+)$ by
\begin{align*}
(Wf)(x)&= x^{-\frac12}f\big(\ln(x)\big)\qquad \forall x\in \R_+,\\
(W^{-1}g)(t)&= \e^{\frac{t}{2}}g\big(\e^t\big) \qquad \forall t\in \R,
\end{align*}
and check that
\begin{equation*}
W^{-1}X^2W =\e^{2Q}\qquad \hbox{and }\qquad  W^{-1}AW = P
\end{equation*}
with $Q$ and $P$ the usual self-adjoint operators of position and momentum in $L^2(\R)$.
Therefore, one infers that
\begin{equation*}
W^{-1}\e^{\i tX^2}\psi(-A)\e^{-\i tX^2}W
=\e^{\i t\e^{2Q}}\psi(-P)\e^{-\i t\e^{2Q}}.
\end{equation*}

For $f_1,f_2\in L^2(\R)$ with compact support one then observes that
\begin{align}\label{pij}
\nonumber &\big(f_1|\e^{\i t\e^{2Q}}\psi(-P)\e^{-\i t\e^{2Q}} f_2\big)\\
\nonumber &= \frac{1}{2\pi}\int\d x \int \d\xi \int\d y \;\!\bar{f_1(x)} \;\!\e^{\i t\e^{2x}}\;\!\psi(-\xi)\;\!\e^{\i(x-y)\xi}\;\!\e^{-\i t\e^{2y}} \;\!f_2(y)\\
\nonumber &= \frac{1}{2\pi} \int\d x\int\d\xi\int\d y \;\!\bar{f_1(x)} \;\!\psi(-\xi)\;\! \exp\left(\i(x-y)\big(\xi+  t\tfrac{\e^{2x}- \e^{2x}}{x-y}\big)\right) \;\!f_2(y)\\
&= \frac{1}{2\pi} \int\d x\int\d\xi\int\d y \;\!\bar{f_1(x)} \psi\left(-\xi+ t\tfrac{\e^{2x}-\e^{2x}}{x-y} \right)\e^{\i(x-y)\xi} f_2(y).
\end{align}
Clearly, for any $x,y$, one has
\begin{equation*}
\frac{\e^{2x}-\e^{2y}}{x-y}>0,
\end{equation*}
and thus for $y\in\supp f_2$ and $x\in\supp f_1$ there exists a strictly positive $c_0$ such that
\begin{equation*}
\frac{\e^{2x}-\e^{2y}}{x-y}\geq c_0>0.
\end{equation*}
Finally, one easily obtains that \eqref{pij} converges as $t\to \pm \infty$ to
\begin{equation*}
\big(f_1|\psi(\pm\infty)f_2\big),
\end{equation*}
which shows \eqref{eq_to_be_proved1} by a density argument.
\end{proof}

\end{document}